\documentclass[preprint,11pt,authoryear]{elsarticle}
\usepackage{amsmath}%
\usepackage{amsfonts}%
\usepackage{amssymb}%
\usepackage{graphicx}
\usepackage{amssymb}
\usepackage{color}
\usepackage{multirow}
\usepackage{booktabs}
\usepackage{float}
\usepackage{comment}
\usepackage{enumitem}
\usepackage{placeins}
\usepackage[table,dvipsnames]{xcolor}
\usepackage{bbm}

\usepackage[subrefformat=parens, labelfont=up]{subcaption}
\usepackage{amsthm}
\usepackage{thmtools}
\usepackage{hyperref}
\usepackage[margin=2.5cm]{geometry}

\newcounter{panel}

\newcounter{group}

\setcitestyle{round}

\theoremstyle{plain}

\newtheorem{proposition}{Proposition}[section]

\newtheorem{assumption}{Assumption}[section]

\declaretheoremstyle[
  spaceabove=1.1\baselineskip, 
  spacebelow=1.1\baselineskip, 
  headfont=\bfseries,           
  notefont=\itshape,            
  bodyfont=\normalfont,         
  headpunct=.,                  
  postheadspace=1em
]{estimatorstyle}

\theoremstyle{plain}
\newtheorem{theorem}{Theorem}[section]
\newtheorem{corollary}{Corollary}[section]
\newtheorem{lemma}{Lemma}[section]

\theoremstyle{remark}
\newtheorem{remark}{Remark}[section]
\newtheorem{ex}{Example}[section]

\newcommand\be{\begin{eqnarray}}
\newcommand\ee{\end{eqnarray}}
\newcommand\bee{\begin{eqnarray*}}
\newcommand\eee{\end{eqnarray*}}

\newcommand\un{\underline}

\newcommand{\R}{\mathbb{R}}
\newcommand{\Z}{\mathbb{Z}}

\newcommand{\E}{\mathbb{E}}

\newcommand{\eones}{\underline{\mathbbm{1}}}
\newcommand{\Var}{\textup{Var}}
\newcommand{\Cov}{\textup{Cov}}
\newcommand{\vect}{\textup{vec}}
\newcommand{\trace}{\textup{tr}}
\newcommand{\diag}{\textup{diag}}

\newcommand{\du}{\underline{\delta}}

\newcommand{\gu}{\underline{g}}
\newcommand{\hu}{\underline{h}}
\newcommand{\tu}{\underline{t}}
\newcommand{\su}{\underline{s}}
\newcommand{\ru}{\underline{r}}
\newcommand{\uu}{\underline{u}}
\newcommand{\vu}{\underline{v}}
\newcommand{\pu}{\underline{p}}
\newcommand{\qu}{\underline{q}}

\newcommand{\lagset}{H_{\un n}}
\newcommand{\nhset}{G_{\un n}(\un h)}
\newcommand{\nset}{G_{\un n}}
\newcommand{\Cctr}{C_Z}
\newcommand{\Cbar}{\overline{C}}
\newcommand{\Xbar}{\overline{X}}
\newcommand{\ngset}{G_{\un n}(\un g)}

\definecolor{m1color}{RGB}{255, 200, 200}  
\definecolor{m2color}{RGB}{200, 200, 255}  
\definecolor{centercolor}{RGB}{255, 255, 0} 

\usepackage{nicematrix}
\usepackage{tikz}
\usepackage{graphicx} 

\numberwithin{equation}{section}

\begin{document}

\begin{frontmatter}

\title{Nonparametric covariance estimation with bias correction for spatial lattice processes}

\author{Antonio Fioravanti\corref{cor1}}
\ead{fioravanti@statistik.tu-dortmund.de}
\author{Carsten Jentsch}
\address{Department of Statistics, TU Dortmund University, 44221 Dortmund, Germany}
\cortext[cor1]{Corresponding author}
\ead{jentsch@statistik.tu-dortmund.de}

\date{\today}

\begin{abstract}
Accurate covariance estimation is crucial for spatial data analysis. While parametric methods can suffer from model misspecification leading to wrong conclusions, nonparametric approaches are often neglected in practice as they rely on the estimation of a large number of covariance parameters and often face finite-sample bias issues.

In this paper, we study the bias properties of sample covariance estimators for stationary lattice processes on $\mathbb{Z}^2$, when the true mean parameter is unknown and has to be estimated. We derive exact formulas for the finite-sample biases of sample covariance estimators and show that their expectations are linear combinations of population covariances determined by spatial lag and sample size. Based on this characterization, we propose jointly bias-corrected covariance estimators that are nearly unbiased. 
Additionally, we derive formulas for the mean-squared error and prove asymptotic normality results that show asymptotic equivalence for the estimators with and without bias correction. Simulations demonstrate substantial reductions of bias and often also in terms of MSE, particularly under strong spatial dependence. A kriging application to topography data further illustrates the practical benefits of the proposed approach.
\end{abstract}


\end{frontmatter}


\section{Introduction}\label{introduction}

The study of spatially dependent processes, that is, of random variables that are indexed by a set of discrete spatial coordinates of dimension larger than one, has become increasingly important across different scientific disciplines. Among others, these include geostatistics, environmental science, epidemiology, image analysis and econometrics. When studying spatial data, a key aspect is the analysis of the spatial dependence, that is, how strong the random variables are correlated with others in nearby locations and how fast this (linear) stochastic dependence actually decays when the distance between two locations becomes larger. The accurate estimation of the dependence a central problem in the analysis of any kind of dependent processes. Compared with time series, where dependence evolves along a single ordered dimension, spatial processes exhibit dependence across several directions. This makes the dependence structure more complex and accurate covariance estimation becomes even more crucial to guarantee meaningful inference on the underlying process and reliable prediction at unobserved locations. For a broad review of spatial statistics and spatial dependence concepts, see \citet{Cressie1993, sherman2011spatial, schabenberger2017statistical}. 

Two major approaches are commonly used to model spatial dependence. On the one hand, parametric approaches are used to specify models for the covariance function based on a (low-dimensional) parameter vector, such as, e.g., the Mat{\'e}rn, exponential or spherical model, among others \citep{Cressie1993, Guttorp2006} and estimates its parameters from the data, typically by maximum likelihood  \citep{mardia1984} or by fitting an empirical variogram using weighted least squares \citep{Cressie1985}. Parametric models are attractive, because valid covariance families guarantee positive semidefiniteness of the estimated covariance function and boil down the problem to estimating some (low-dimensional) parameter vector that determines the covariance function. However, parametric approaches impose potentially restrictive assumptions on the covariance structure, and the effect of covariance misspecification, e.g., on spatial prediction depends critically on the appropriateness of the assumed covariance models \citep{Stein1989}. Sometimes even under correct specification, e.g. of a Mat{\'e}rn covariance model, individual covariance parameters may fail to be consistently estimable under fixed-domain asymptotics \citep{Zhang2004}.

On the other hand, to address the possible restrictiveness of parametric assumptions, nonparametric approaches estimate the covariances or (semi-)variograms directly from the data without imposing any specific parametric form. Classical geostatistics predominately favors empirical variogram estimation \citep{matheron1962traite}; while more flexible methods include kernel-methods \citep{hall1994nonparametric},  bootstrap methods \citep{Bjrnstad2001}, sieve maximum likelihood \citep{wang2023nonparametric}, among others. A recent review of nonparametric methods for the estimation of the covariance for spatial data is provided by \citet{Bilchouris2025}, see also the references therein. Relying on the common method-of-moment principle, the most direct nonparametric technique is based on the sample covariance. This estimator, however, is well known to be biased in finite samples, when the true mean is unknown \citep{sherman2011spatial}.

For lattice covariance estimation, the focus has been on the bias caused by edge or boundary effects, which are typical for lattices and higher-dimensional data arrays. \citet{guyon1982parameter} showed that such a bias is not asymptotically negligible in a lattice of dimension larger or equal to 2, and proposed to normalize the sample covariance by the number of observed pairs at its respective lag rather than using the total sample size, which in this paper is referred to as the \emph{standard estimator}. This solved the issue of the bias induced by the edge effect \citep{ripley1984spatial}. Later, \citet{dahlhaus1987edge} introduced data tapering to reduce the large variance caused by estimating covariances at high lags. More recently, \citet{guillaumin2022debiased} propose a debiasing of the Whittle quasi-likelihood estimator by replacing the spectral density by the expected periodogram. However, all these methods operate in settings where the process mean is known and often set to zero. In practice, this is not the case and the unknown mean needs to replaced by a suitable estimator (typically the sample mean) in the estimation process. In comparison to the edge effect bias, this introduces a different kind of bias \cite{sherman2011spatial}, which is the subject of the present paper.

The same bias problem arises also in time series analysis, 
where for univariate time series, \citet{vogelsang2016exactly} proposed a matrix inversion framework that yields nearly unbiased sample covariance estimation, subsequently extended also to the multivariate time series case by \citet{li2024correcting}. To the best of our knowledge, no analogous bias correction exists for spatial lattice processes, where the two dimensional structure introduces additional complexity and the number of covariances to be estimated becomes large already for small sample sizes. The present paper aims to fill this gap. 

The remainder of the paper is organized as follows. In Section \ref{sec:exact_bias}, we introduce the spatial lattice setup and provide the notation used throughout the paper. Then, we present the two versions of sample covariance estimators studied in this paper, which differ only in their prefactors, and derive formulas for their exact finite sample bias and also jointly for the whole matrix of sample covariances at all observed lags. 
Based on these findings, in Section~\ref{sec:mainresult}, we propose a joint bias-corrected sample covariance (matrix) estimator that is demonstrated to be nearly unbiased. We quantify the remaining bias, calculate the mean squared error and suggest two data-driven rules for the selection of a required cutoff parameter. Moreover, we
establish that the correction approach applied to both versions of sample covariance estimators studied in this paper lead exactly to the same bias-corrected estimator and we present asymptotic normality results.
In Section~\ref{sec:simulations}, we report the results of a Monte Carlo simulation study for separable and nonseparable Gaussian random fields. A real-data application is shown in Section~\ref{sec:dataapp}, where we apply bias-corrected covariance estimation to conduct ordinary kriging on a topographic dataset. 
Section~\ref{sec:conclusion} concludes and states future research directions. All proofs are collected in the Appendix.

\section{The exact bias of sample covariances for lattice processes}\label{sec:exact_bias}
In this section, after presenting the notation and assumptions employed in this paper, we introduce the two sample covariance estimators (which differ only in the prefactor used), we derive the finite sample bias of each of them, first for a single scalar estimator, and then for a joint matrix estimator of covariances. Finally, we describe the second-order structure of the estimators assuming Gaussianity.

\subsection{Spatial lattice processes: setup and notation}\label{subsec:setup}

We consider a real-valued univariate spatial process $\mathbf{X}=\{X_{\underline{t}},\; \underline{t}=(t_1,t_2)\in\mathbb{Z}^2\}$, indexed on the two-dimensional integer lattice\footnote{Here and throughout the manuscript, vector-valued quantities are denoted by underlined lower-case letters $\un v$ and matrices by bold upper-case letters $\mathbf{M}$.}. A vector $\underline{h}=(h_1,h_2)\in\mathbb{Z}^2$ is called a spatial \emph{lag} and represents the shift separating a location $\un t$ from the location $\un t+\un h$. We state two commonly imposed assumptions in the following.

\begin{assumption}[Second-order stationarity]
\label{ass:stationarity}
$\mathbf{X}$ is second-order (or weakly) stationary with constant mean $\E(X_{\un t})=\mu$ for all $\un t\in\mathbb{Z}^2$ and with spatial covariance function $C:\mathbb{Z}^2\rightarrow \mathbb{R}$, such that
\begin{align}\label{eq:covariance}
 C(\un h) &:=\Cov(X_{\un t+\un h}, X_{\un t}) = \E\bigl(\left(X_{\un t +\un h}-\mu\right)\left(X_{\un t}-\mu\right)\bigr), \quad  \forall \un t, \un h \in \mathbb{Z}^2.
\end{align}
In particular, it holds that $C(\un h)= C(-\un h)$.
\end{assumption}

\begin{assumption}[Short-range dependence] 
\label{ass:abs_summability}
The covariance function $C$ from \eqref{eq:covariance} is absolutely summable over $\mathbb{Z}^2$. That is, we have
\begin{equation}\label{eq:abs_summability}
\|C\|_{1} =\sum_{\hu\in\Z^{2}}|C(\hu)|<\infty.
\end{equation}
    
\end{assumption}
In particular, under Assumption~\ref{ass:stationarity}, a Gaussian random field is fully characterized by its first and second moments, while Assumption~\ref{ass:abs_summability} imposes the common short-range dependence condition that excludes long-memory lattices fields.

Next, we define how data from the process $\mathbf{X}$ is observed. Suppose that for $\un n = (n_1, n_2)\in \mathbb{N}^2$, we obtain a lattice realization and collect the observations in the data matrix $\mathbf{X}_{\underline{n}}:=\bigl(X_{\underline{t}},\; \underline{t}\in G_{\un n}\bigr)$ of the spatial process $\mathbf{X}$ on a $(n_1\times n_2)$ rectangular grid, whose unique coordinate pairs are collected in the index set
\begin{equation}\label{eq:grid_set}
G_{\underline{n}} = \bigl\{\un t = (t_1, t_2) : 1 \le t_i \le n_i,\; i =1,2 \bigr\} \subset \mathbb{Z}^2.
\end{equation}
In the following, we denote its cardinality $|\nset|$, i.e., the total amount of unique coordinate pairs in the grid $\nset$, also as the \emph{total} sample size and we define
\begin{equation}\label{eq:total_sample_size}
N := |\nset| = n_1n_2.
\end{equation}
Accordingly, all lags realized by $G_{\underline n}$ are collected in
\begin{equation}\label{eq:lag_set}
H_{\underline n}
:=\{\underline h=(h_1,h_2): |h_i|\leq n_i-1,\ i \in \{1,2\}\}
\subset\mathbb Z^2.
\end{equation}
Note that not every location in $G_{\underline n}$ admits every lag in $H_{\underline n}$: a location $\un t$ and lag $\un h$ has an observation in $\mathbf{X}_{\underline{n}}$ if and only if the location $\un t + \un h$ is also contained in the observed grid $\nset$. Thus, we define also the subset of all valid locations for the lag $\un h$ as
\begin{equation}\label{eq:effective_set}
\nhset:=\bigl\{\un t\in \nset : \un t+\un h\in \nset\bigr\}\subseteq \nset,
\qquad \un h\in H_{\un n}.
\end{equation}
Its cardinality $|\nhset|$, is the \emph{effective} sample size corresponding to lag $\un h$, i.e., the number of pairs $(X_{\un t}, X_{\un t + \un h})$ separated by the lag $\un h$ such that both $X_{\un t}$ and $X_{\un t+ \un h}$ are contained in the grid $G_{\underline n}$. This leads to  
\begin{align}
N_h &:=|\nhset| =  (n_1-|h_1|)(n_2-|h_2|).
\label{eq:effective_sample_size}
\end{align}
Finally, we introduce some set notation. First, for arbitrary integers $a,b\in\mathbb{Z}$ with $a\leq b$, we use the short notation
\begin{align}\label{eq:I_set_definition}
I(a:b):= \{a,\, a+1,\, a+2,\, \ldots,\,b-2,\,b-1,\,b\}
\end{align}
for the set of all integers from $a$ to $b$. Hence, for $i\in \{1,2\}$, the lags of $\un h \in H_{\un n}$ run along the interval $I\bigl((-n_i+1):(n_i-1)\bigr)$ and the locations $\un t\in \nhset$ run along the interval $
I\bigl(\max(1,1-h_i): \min(n_i, n_i-h_i)\bigr)$. 

Throughout the paper, we will deal with the vectorization operator $\vect(\cdot)$ that stacks a rectangular lattice column-wise. Accordingly, the mapping that allows to recover the correct double indices of the lattice from the single index of the vectorized lattice can be defined as follows: for a lattice $F = I(a_1:b_1)\times I(a_2:b_2)$ with set cardinalities $c_i := |I(a_i:b_i)|= b_i-a_i+1$, $i\in\{1,2\}$ and total cardinality $|F| = c_1c_2$, we use the mapping
\begin{equation}\label{eq:mapping_vec}
    \un{\ell}_F:\{1, \dots, |F|\} \rightarrow F, \quad \un{\ell}_F(k):= \Biggl(\bigl((k-1) \bmod c_1\bigr)+a_1, \; \Bigg\lfloor\frac{k-1}{c_1}\Bigg\rfloor+a_2\Biggr),
\end{equation}
where $\lfloor x\rfloor$ denotes the largest integer smaller or equal $x$, and $a\bmod b$ denotes the remainder of $a$ upon division by $b$.

We illustrate this spatial lattice setup with a Gaussian example, which also serves as the data-generating process in our simulation study in Section~\ref{sec:simulations}.

\begin{ex}[Gaussian spatial lattice process]\label{example_Gaussian}
Suppose the spatial lattice process $\mathbf{X}$ is Gaussian. Further, suppose that the covariance $C(\hu)$ between arbitrary points $X_{\tu+\hu}$ and $X_{\tu}$ having spatial lag $\hu=(h_1,h_2)$ is given by the exponential covariance function
\begin{align}\label{eq:trueautocov}
C_{\lambda}(h_1, h_2) = \sigma^2 \exp \left( - \frac{|h_{1}|}{\lambda} - \frac{|h_{2}|}{\lambda}-\beta\frac{|h_1-h_2|}{\lambda} \right),    \quad   \hu\in H_{\underline{n}},
\end{align}
where $\sigma^2=\Var(X_{\tu})=C_{\lambda}(0, 0)$ is the variance of the lattice process $\mathbf{X}$, the parameter $\lambda$ controls the persistence of the spatial dependence and $\beta$ the degree of (non)-separability. In this paper, we consider $\lambda \in \{2,4,6,8\}$, where $\lambda = 2$ gives the weakest and $\lambda = 8$ the strongest spatial dependence. Further, $\beta=0$ leads to a separable covariance function, while $\beta\in (0,1]$ implies non-separability.
\end{ex}

In Figure \ref{fig:realization_sp}, we show realizations $\mathbf{X}_{\underline{n}}$ for $n_1=n_2=30$ of two (Gaussian) spatial processes $\mathbf{X}$ for the exponential covariance function defined in \eqref{eq:trueautocov} for $\lambda=4$ (moderate dependence) and $\lambda=8$ (large dependence), respectively, and $\beta=0$ (separable case).  

\begin{figure}[t]
    \centering
\includegraphics[width=0.9\textwidth]{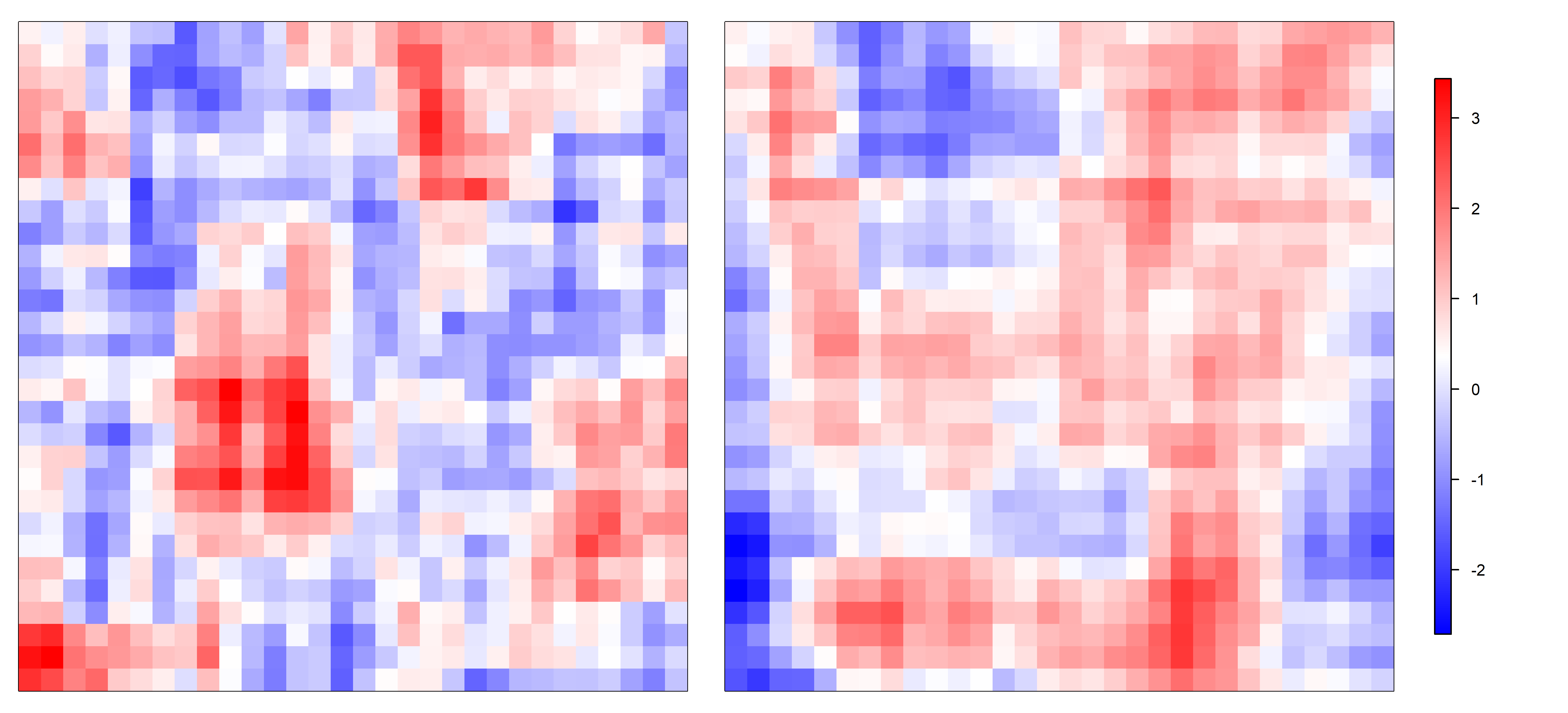}
\caption{Realizations $\mathbf{X}_{\underline{n}}$ with $n_1=n_2=30$ of spatial (Gaussian) processes $\mathbf{X}$ with exponential covariance function $C_\lambda$ defined in \eqref{eq:trueautocov} with $\lambda=4$ (left panel) and $\lambda=8$ (right panel) and $\beta=0$ (separable case) in both panels.}
\label{fig:realization_sp}
\end{figure}

\subsection{Sample covariance estimators for spatial lattice processes}\label{subsec:sample_covariances}

We consider two versions of spatial sample covariance estimators that have been proposed in the literature. As defined in the following, they differ only in their prefactors. On the one hand, using the prefactor $1/N_{\un h}$ that scales by the \emph{effective} sample size $N_{\un h}$ from \eqref{eq:effective_sample_size} for lag $\un h$, we define 
\begin{align}
\widehat{C}(\un h)
&= \frac{1}{N_{\un h}}\sum_{\tu\in \nhset}
   \bigl(X_{\tu+\un h}-\overline X\bigr)\bigl(X_{\tu}-\overline X\bigr),    \quad   \un h\in H_{\underline{n}}, \label{eq:sampleAutocov}
\end{align}
and call $\widehat{C}(\un h)$ a \emph{standard covariance estimator}. On the other hand, using the prefactor $1/N$ that scales by the \emph{total} sample size $N$ from \eqref{eq:total_sample_size} for all lags $\un h$, we define 
\begin{align}
\widetilde{C}(\un h)
&= \frac{1}{N}\sum_{\tu\in \nhset}
   \bigl(X_{\tu+\un h}-\overline X\bigr)\bigl(X_{\tu}-\overline X\bigr),    \quad   \un h\in H_{\underline{n}}.
\label{eq:sampleAutocov_const}
\end{align}
In the following, $\widetilde{C}(\un h)$ is called a \emph{constant prefactor covariance estimator}. In the above, $\nhset$ is the set of coordinates that admit the lag $\un h$ from \eqref{eq:effective_set}, $\overline X$ denotes the sample mean of the observed data on the whole grid $G_{\un n}$, defined as
\begin{equation}\label{eq:samplemean}
 \overline{X}=\frac{1}{N}\sum_{\un u\in \nset }X_{\un u},
 \end{equation} 
and, as usual, we set $\widehat{C}(\hu)=\widetilde{C}(\hu):=0$ if $\hu\notin H_{\underline{n}}$.

The standard estimator \eqref{eq:sampleAutocov} is normalized by the effective sample size $N_h$ from \eqref{eq:effective_sample_size} and is the version most commonly used in geostatistics \citep{Cressie1993}. The constant prefactor estimator \eqref{eq:sampleAutocov_const} instead divides every lag by the total sample size $N$. 
Although the latter approach is quite uncommon in the spatial statistics literature, it can be motivated and is largely inspired by analogous practices in time series analysis. In particular, the choice of the reciprocal sample size as a universal prefactor for all sample covariances (irrespective of the lag $\un h$) guarantees positive definiteness of the resulting sample covariance function (and all sample covariance matrices); see \cite{brockwell1991, kreiss2006einfuhrung}. 

The standard estimator $\widehat C(\un h)$ \eqref{eq:sampleAutocov} is mainly motivated by the fact that its \emph{infeasible} version 
\begin{align}
\widehat{C}_{\mu}(\un h)
&= \frac{1}{N_{\un h}}\sum_{\tu\in \nhset}
   \bigl(X_{\tu+\un h}- \mu\bigr)\bigl(X_{\tu}- \mu\bigr),    \quad   \un h\in H_{\underline{n}}, \label{eq:sampleAutocov_infeasible}
\end{align}
where $\overline{X}$ is replaced by $\mu$, 
is an unbiased estimator of the covariance function $C(\un h)$, i.e., $\E\bigl(\widehat{C}_{\mu}(\un h)\bigr) = C(\un h)$; see \citet{guyon1982parameter}. In contrast, the analogously defined \emph{infeasible} version $\widetilde{C}_\mu(\un h)$ of the constant prefactor estimator $\widetilde{C}(\un h)$ from \eqref{eq:sampleAutocov_const} does still exhibit a bias (unless $\un h=(0,0)$) due to the edge or boundary effect, i.e., $\E\bigl(\widetilde{C}_{\mu}(\un h)\bigr) \neq C(\un h)$ \citep{guyon1982parameter}. Still when comparing $\widehat{C}(\un h)$ and $\widetilde{C}(\un h)$ under the unknown mean setting, $\widehat{C}(\un h)$ exhibits a smaller bias at small lags, thus the standard estimator is often preferred over the constant prefactor estimator.

The issue on which we focus in this paper arises when the true mean $\mu$ is \emph{unknown} and the sample mean estimator $\overline{X}$ from \eqref{eq:samplemean} (which is itself an unbiased estimator for $\mu$) is used
in \eqref{eq:sampleAutocov}-\eqref{eq:sampleAutocov_const}. In this case, also $\widehat C(\un h)$ is biased, i.e., $\E(\widehat{C}(\un h)) \neq C(\un h)$, see \citet{sherman2011spatial}.

\subsection{Finite sample bias of a single covariance estimator }\label{subsec:finite_sample_bias}

In this section, we quantify the finite-sample biases of both versions of sample covariance estimators $\widehat{C}(\un h)$ and $\widetilde{C}(\un h)$ defined in \eqref{eq:sampleAutocov} and \eqref{eq:sampleAutocov_const}, respectively, that arise when the true mean $\mu$ of the process $\mathbf{X}$ is unknown and the sample mean $\overline X$ from \eqref{eq:samplemean} is plugged-in for the estimation.
Consequently, both estimators are computed from the deviations of the observations from the sample mean. Under Assumption~\ref{ass:stationarity}, this implies that their moments are governed by the second-order structure of the demeaned process, rather than that of the process $\mathbf{X}$. By taking a closer look at the covariance between $X_{\tu}-\overline{X}$ and $X_{\su}-\overline{X}$ for some $\tu, \su \in G_{\underline{n}}$, we get 
\begin{equation}\label{eq:C_centered}
\Cctr(\tu, \su):= \Cov(X_{\tu}-\overline{X}, X_{\su}-\overline{X}) = C(\tu-\su) - \overline{C}(\tu) - \overline{C}(\su) + \Var(\overline{X}),
\end{equation}
where, for $\tu \in G_{\underline{n}}$, 
\begin{equation}\label{eq:Cbar}
\overline{C}(\un t) := \Cov(X_{\un t}, \overline{X}) = \frac{1}{N}\sum_{\un u\in \nset} C(\un t-\un u)
\end{equation}
denotes the covariance between a single observation $X_{\un t}$ and the sample mean $\overline{X}$. Further, for the variance of the sample mean, we get 
\begin{equation}\label{eq:D}
  D:=\Var(\overline{X}) = \frac{1}{N}\sum_{\un v\in \nset} \overline{C}(\un v)=  \frac{1}{N^2}\sum_{\un v\in \nset}\sum_{\un u\in \nset} C(\un v-\un u).
\end{equation}
We refer to $\Cctr$ as the \emph{centered covariance function}  of the process $\mathbf{X}$. In contrast to $C$, which takes as arguments the lags $\hu$, the centered covariance function $C_Z$ is a function of the two locations. That is, $C_Z(\un t, \un s)$ does depend on the positions $\tu$ and $\su$ and not only on their difference $\tu-\su$.

Now, with the notation introduced above, the expectation of the standard covariance estimator $\widehat{C}(\hu)$ becomes
\begin{equation}\label{eq:single_st_cov_expectation}
\E\big(\widehat{C}(\hu)\big) = \frac{1}{N_h}\sum_{\tu \in \nhset}C_Z(\tu+\hu, \tu) =  C(\hu) - A_{\hu} - B_{\hu} + D,
\end{equation}
where $D$ is the variance of the sample mean \eqref{eq:D}  and where
\begin{align}
A_{\hu} &:= \frac{1}{N_{\hu}}\sum_{\tu \in \nhset}\overline{C}(\tu+\hu)=\frac{1}{N_{\hu}}\sum_{\tu \in \nhset}\Cov\big(X_{\tu+\hu},\overline X\big)
 = \frac{1}{N_{\hu}N}\sum_{\tu\in \nhset}\ \sum_{\uu\in \nset} C(\tu+\hu-\uu),
 \label{eq:Ah}\\[4pt]
B_{\hu} &:= \frac{1}{N_{\hu}}\sum_{\tu \in \nhset}\overline{C}(\tu)=\frac{1}{N_{\hu}}\sum_{\tu \in \nhset}\Cov\big(X_{\tu},\overline X\big)
 = \frac{1}{N_{\hu}N}\sum_{\tu\in \nhset}\ \sum_{\uu\in \nset} C(\tu-\uu).
 \label{eq:Bh}
\end{align}
The representation of $\E(\widehat{C}(\hu))$ in \eqref{eq:single_st_cov_expectation} makes the source of the bias explicit. The terms $A_h$ and $B_h$ from \eqref{eq:Ah} and \eqref{eq:Bh} are averages of the covariances between the observations entering the sample covariance estimator, namely $X_{\tu +\hu}$ and
$X_{\tu}$, respectively, and the sample mean, while $D$ is the variance of the sample mean itself. 
Note that each of the two terms \eqref{eq:Ah} and \eqref{eq:Bh} are double sums over the set $\nhset$ and $\nset$. Collecting the terms accordingly reduces each double sum to
a single sum over $\ru\in\lagset$, weighted by counts that factorize across the two lattice
dimensions. Details are given in the Appendix, and the result is summarized in the following theorem.

\begin{theorem}[Exact bias of a single $\widehat{C}(\hu)$]\label{prop:single_bias}
Suppose Assumption~\ref{ass:stationarity} holds and let $\hu \in \lagset$. Then,
the bias of the standard covariance estimator $\widehat{C}(\hu)$ defined in
\eqref{eq:sampleAutocov} is
\begin{align}
b(\hu) := \E\bigl(\widehat{C}(\hu)\bigr)-C(\hu)
= -\frac{1}{N_{\hu}}\sum_{\tu \in \nhset}
   \left[\Cov(X_{\tu+\hu}, \Xbar)+ \Cov(X_{\tu}, \Xbar)\right]+\Var(\Xbar)
\label{eq:biasautocov_explicit_mainprop}
\end{align}
or, equivalently represented, we have
\begin{align}
b(\hu) = -\frac{S(\hu)+S(-\hu)}{N N_{\hu}}+\frac{S(\un 0)}{N^{2}},
\label{eq:biasautocov}
\end{align}
where $S(\hu)$ is a weighted sum running over the full lag set $\lagset$ defined by
\begin{align}
S(\hu) := \sum_{\ru \in \lagset} C(\ru)\,\prod_{i=1}^2\nu_{n_i}(r_i,h_i),
\label{eq:S-def}
\end{align}
and, with $x^{+}:=\max(x,0)$, the counting factors $\nu_{n_i}(r_i,h_i)$ are given by
\begin{align}
\nu_{n_i}(r_i,h_i) := \bigl(n_i-\max(|r_i|,|h_i|,|r_i+h_i|)\bigr)^{+},\quad i \in \{1,2\}.
\label{eq:counting_factor}
\end{align}
\end{theorem}

A similar bias formula holds for the sample covariance estimator with constant prefactor $\widetilde C(\hu)$ from \eqref{eq:sampleAutocov_const}, which we state as a simple corollary.

\begin{corollary}[Exact bias of a single $\widetilde{C}(\hu)$]\label{cor:single_bias_const}
Under the conditions of Theorem~\ref{prop:single_bias}, the bias of the constant prefactor
estimator $\widetilde{C}(\hu)$ defined in \eqref{eq:sampleAutocov_const} becomes
\begin{eqnarray}\label{eq:biasautocov_constant}
\widetilde{b}(\hu)
:= \E\bigl(\widetilde{C}(\hu)\bigr)-C(\hu) &=& \left(\frac{N_h}{N}-1\right)C(\hu) +\frac{N_h}{N}b(\hu)\notag \\
&=& \left(\frac{N_{\hu}}{N}-1\right)C(\hu)
   - \frac{S(\hu)+S(-\hu)}{N^{2}}
   + \frac{N_{\hu}}{N^{3}}\,S(\un 0),
\label{eq:biasautocov_constant}
\end{eqnarray}
with $S(\un h)$ as defined in \eqref{eq:S-def}.
\end{corollary}

Some remarks on the biases presented in Theorem~\ref{prop:single_bias} and Corollary~\ref{cor:single_bias_const} are in order.

\begin{remark}[On the bias formulas]\label{remark:single_bias} \ 
\begin{enumerate}[label=(\roman*)]
\item Both biases $b(\hu)$, $\widetilde{b}(\hu)$ from \eqref{eq:biasautocov} and \eqref{eq:biasautocov_constant} are weighted sums of all covariances between any two random variables in $\mathbf{X}_{\underline{n}}$, where the weights depend only on the sample sizes $(n_1,n_2)$ and the spatial lags $(h_1,h_2)$. In particular, the bias does \emph{not} depend on the unknown mean $\mu=\E(X_{\tu})$. 
\item If the true mean $\mu$ is known and replaces the sample mean $\Xbar$ in the standard estimator $\widehat{C}(\hu)$ (leading to $\widehat{C}_{\mu}(\hu)$ defined in \eqref{eq:sampleAutocov_infeasible}), we get an unbiased estimator as all the terms on the right-hand side of \eqref{eq:biasautocov_explicit_mainprop} become zero in that case.
\item Conversely, the constant prefactor estimator $\widetilde{C}(\hu)$ remains biased even when the true mean $\mu$ is known and replaces the sample mean $\Xbar$ (leading to $\widetilde{C}_{\mu}(\hu)$), as that information does not remove the term $(N_h/N-1)C(\hu)$ from the bias in \eqref{eq:biasautocov_constant}. This leftover bias in the $\widetilde{C}(\un h)$ estimator is caused by the \emph{edge effect}, as it represents the fraction of lattice pairs lost to the edge or boundary of the lattice, as studied by \citet{guyon1982parameter}.
\item Under positive spatial dependence, i.e., whenever $C(\hu) \geq 0$ for all $\hu \in H_{\underline n}$, the bias is negative for both estimators, meaning that they systematically \emph{underestimate} the true covariance. 
\item For the sample covariance at lag $\underline{h}=(0,0)$, i.e., when estimating the variance $\Var(X_{\tu})$, we have $N_h=N$, and in this case $b(\hu)$ and $\widetilde{b}(\hu)$ from \eqref{eq:biasautocov} and\eqref{eq:biasautocov_constant} coincide.
\end{enumerate}
\end{remark}

According to Remark \ref{remark:single_bias} (i), for arbitrary spatial lags $\hu$, the bias $b(\hu)$ of $\widehat C(\hu)$ (also for bias $\widetilde b(\hu)$ of $\widetilde C(\hu)$) does depend on \emph{all} covariances between any two random variables in $\mathbf{X}_{\underline{n}}$. Hence, with the goal of bias corrected estimation in view, it appears to be advisable to consider sample covariances and their biases \emph{jointly}, which is done in the following.

\FloatBarrier
\subsection{Joint bias of a matrix of the standard sample covariances}
\label{subsec:bias_matrix_sample_covariance}

In this section, we adopt a matrix-based approach that allows for a simultaneous treatment of \emph{all} spatial lags contained in $\un h\in H_{\underline{n}}$ at once. We focus first on the standard estimator $\widehat{C}(\hu)$ from \eqref{eq:sampleAutocov}. For this purpose, we organize the population covariances in a comprehensive $(L_1\times L_2)$ matrix $\boldsymbol{\Gamma}$, where $L_i = 2n_i - 1$ for $i\in\{1,2\}$ and $L_{\un n} = L_1L_2$, encompassing all lag combinations between any two random variables in $\mathbf{X}_{\underline{n}}$. That is, we define
\begin{align}\label{eq:Gamma_pq}
    \boldsymbol{\Gamma} = \left(\begin{array}{c}\boldsymbol{\Gamma}_{p,q},  \\  p=1,\ldots,L_1  \\  q=1,\ldots,L_2 \end{array}\right) = \left(\begin{array}{c}C(h_1, h_2),  \\  h_1\in I(1-n_1:n_1-1)  \\  h_2\in I(1-n_2:n_2-1) \end{array}\right), \quad   \bigl[\boldsymbol{\Gamma}\bigr]_{p,q} = C(p-n_1,\,q-n_2),
\end{align}
such that the matrix $\mathbf{\Gamma}$ contains all covariances $C(\un h)$, $\un h\in H_{\underline{n}}$. In the same way, we define the corresponding matrix of sample covariances $\boldsymbol{\widehat \Gamma}$ by replacing all $C(h_1,h_2)$'s by $\widehat C(h_1,h_2)$'s in \eqref{eq:Gamma_pq}. 
Moreover, for $i\in\{1,2\}$, we arrange the counting factors $\nu_{n_i}(r_i, h_i)$ from \eqref{eq:counting_factor} into three $(L_i\times L_i)$ matrices, whose rows are indexed by the lag $h_i = p-n_i$ and columns by the lag $r_i = q-n_i$. In the bias representation \eqref{eq:biasautocov}, the lag $\hu$ enters as $S(\hu)$, $S(-\hu)$, $S(\un 0)$, and thus affects only the second argument of the counting factor $\nu_{n_i}(r_i, h_i)$ from \eqref{eq:counting_factor}. Accordingly, for $i\in\{1,2\}$ and $p,q=1,\ldots,L_i$, we
define the \emph{raw counting matrices} $\mathbf A_i$, $\mathbf B_i$, and $\mathbf D_i$ having entries
\begin{align}\label{eq:counting_matrices}
  \bigl[\mathbf A_i\bigr]_{p,q} := \nu_{n_i}\bigl(q-n_i,\; p-n_i\bigr),
  \quad
  \bigl[\mathbf B_i\bigr]_{p,q} := \nu_{n_i}\bigl(q-n_i,\; n_i-p\bigr),
  \quad
  \bigl[\mathbf D_i\bigr]_{p,q} := \nu_{n_i}\bigl(q-n_i,\; 0\bigr),
\end{align}
\begin{figure}[t]
\begin{center}
\renewcommand{\arraystretch}{0.9}
\footnotesize
\begin{minipage}{0.3\textwidth}
\centering
\textbf{Matrices $A_1 = A_2$}
\vspace{2mm}
\resizebox{\linewidth}{!}{%
\begin{NiceTabular}{c|*{7}{c}}[cell-space-limits=0.5pt]
$h_i \backslash r_i$ & -3 & -2 & -1 & 0 & 1 & 2 & 3 \\ \hline
-3 & 0 & 0 & 0 & 1 & 1 & 1 & 1 \\
-2 & 0 & 0 & 1 & 2 & 2 & 2 & 1 \\
-1 & 0 & 1 & 2 & 3 & 3 & 2 & 1 \\
0 & 1 & 2 & 3 & 4 & 3 & 2 & 1 \\
1 & 1 & 2 & 3 & 3 & 2 & 1 & 0 \\
2 & 1 & 2 & 2 & 2 & 1 & 0 & 0 \\
3 & 1 & 1 & 1 & 1 & 0 & 0 & 0
\end{NiceTabular}
}
\end{minipage}
\hfill
\begin{minipage}{0.3\textwidth}
\centering
\textbf{Matrices $B_1 = B_2$}
\vspace{2mm}
\resizebox{\linewidth}{!}{%
\begin{NiceTabular}{c|*{7}{c}}[cell-space-limits=0.5pt]
$h_i \backslash r_i$ & -3 & -2 & -1 & 0 & 1 & 2 & 3 \\ \hline
-3 & 1 & 1 & 1 & 1 & 0 & 0 & 0 \\
-2 & 1 & 2 & 2 & 2 & 1 & 0 & 0 \\
-1 & 1 & 2 & 3 & 3 & 2 & 1 & 0 \\
0 & 1 & 2 & 3 & 4 & 3 & 2 & 1 \\
1 & 0 & 1 & 2 & 3 & 3 & 2 & 1 \\
2 & 0 & 0 & 1 & 2 & 2 & 2 & 1 \\
3 & 0 & 0 & 0 & 1 & 1 & 1 & 1
\end{NiceTabular}
}
\end{minipage}
\hfill
\begin{minipage}{0.3\textwidth}
\centering
\textbf{Matrices $D_1 = D_2$}
\vspace{2mm}
\resizebox{\linewidth}{!}{%
\begin{NiceTabular}{c|*{7}{c}}[cell-space-limits=0.5pt]
$h_i \backslash r_i$ & -3 & -2 & -1 & 0 & 1 & 2 & 3 \\ \hline
-3 & 1 & 2 & 3 & 4 & 3 & 2 & 1 \\
-2 & 1 & 2 & 3 & 4 & 3 & 2 & 1 \\
-1 & 1 & 2 & 3 & 4 & 3 & 2 & 1 \\
0 & 1 & 2 & 3 & 4 & 3 & 2 & 1 \\
1 & 1 & 2 & 3 & 4 & 3 & 2 & 1 \\
2 & 1 & 2 & 3 & 4 & 3 & 2 & 1 \\
3 & 1 & 2 & 3 & 4 & 3 & 2 & 1
\end{NiceTabular}
}
\end{minipage}
\end{center}
\caption{The matrices $A_1 = A_2$ (left panel), $B_1 = B_2$ (center panel) and $D_1 = D_2$ (right panel) for $n_1=n_2=4$. As $L_i=2n_i-1$ for $i \in \{1,2\}$, all matrices are of dimension $7\times 7$.}
\label{fig:ABD_example}
\end{figure}
which are displayed in Figure~\ref{fig:ABD_example} for $n_1=n_2=4$. These raw counting matrices have to be normalized in order to incorporate the prefactors $(1/NN_h)$ and $1/N^2$ appearing in \eqref{eq:biasautocov}. Since
$N N_{\hu}=\prod_{i \in \{1,2\}} n_i\bigl(n_i-|h_i|\bigr)$, the relevant
normalization in direction $i$ is obtained by using the $(L_i\times L_i)$ diagonal matrix
\begin{align*}
  \mathbf N_i := \diag\bigl(n_i(n_i-|p-n_i|)\bigr), \quad  p=1,\ldots,L_i
\end{align*}
and the \emph{normalized count matrices} are defined by
\begin{align}\label{eq:normalized_counting_matrices}
  \mathbf A^*_i := \mathbf N_i^{-1}\mathbf A_i,
  \quad
  \mathbf B^*_i := \mathbf N_i^{-1}\mathbf B_i,
  \quad
  \mathbf D^*_i := n_i^{-2}\,\mathbf D_i .
\end{align}
Accordingly, their entries can be shown by using the counting factors from \eqref{eq:counting_factor} and by indexing the matrices with $h_i, r_i\in I(1-n_i: n_i-1)$. That is, we have
\begin{equation}\label{eq:normalised_counting_matrices_form_rh}
[\mathbf A^{*}_i]_{h_i,r_i}
=\frac{\nu_{n_i}(r_i,h_i)}{n_i(n_i-|h_i|)},\qquad
[\mathbf B^{*}_i]_{h_i,r_i}
=\frac{\nu_{n_i}(r_i,-h_i)}{n_i(n_i-|h_i|)},\qquad
[\mathbf D^{*}_i]_{h_i,r_i}
=\frac{\nu_{n_i}(r_i,0)}{n_i^{2}}.
\end{equation}
Now, making use of \eqref{eq:single_st_cov_expectation}, we can express the expectation of the whole matrix of sample covariances $\widehat{\boldsymbol{\Gamma}}$ compactly as
\begin{equation}\label{eq:matrix_bias}
\E\big(\widehat{\boldsymbol{\Gamma}}\big)
= \boldsymbol{\Gamma}
- \mathbf{A}_1^* \boldsymbol{\Gamma} \mathbf{A}_2^{*^\top}
- \mathbf{B}_1^* \boldsymbol{\Gamma} \mathbf{B}_2^{*^\top}
+ \mathbf{D}_1^* \boldsymbol{\Gamma} \mathbf{D}_2^{*^\top}.
\end{equation}
Finally, by applying the vectorization operator $\vect(\cdot)$ on both sides of Equation~\eqref{eq:matrix_bias} and using the Kronecker product properties\footnote{For matrices $\mathbf{A}\in \mathbb{R}^{m\times n}$ and $\mathbf{B}\in\mathbb{R}^{p\times q}$, $\mathbf{A}\otimes \mathbf{B}$ indicates the Kronecker product, which is the $(mp)\times(nq)$ matrix obtained by multiplying each element $a_{ij}$ of $\mathbf{A}$ by the matrix $\mathbf{B}$. In particular, for a fitting matrix $\mathbf{M}\in \mathbb{R}^{q\times n}$ it holds that
$\vect(\mathbf{B}\mathbf{M}\mathbf{A}^\top)
=
(\mathbf{A}\otimes\mathbf{B})\vect(\mathbf{M})$.
Moreover, for matrices $\mathbf{C}\in\mathbb{R}^{n\times r}$ and $\mathbf{D}\in\mathbb{R}^{q\times s}$, it holds that
$(\mathbf{A}\otimes\mathbf{B})(\mathbf{C}\otimes\mathbf{D})
=
(\mathbf{A}\mathbf{C})\otimes(\mathbf{B}\mathbf{D})$.}, we can represent the expectation of the vectorized sample covariance matrix $\widehat{\boldsymbol{\Gamma}}$ as a linear combination of all population covariances in $\boldsymbol{\Gamma}$ and provide its joint bias formula, which is summarized in the following theorem.

\begin{theorem}[Joint mean and bias of $\vect(\widehat{\boldsymbol{\Gamma}})$]\label{prop:joint_bias} Under Assumption~\ref{ass:stationarity}, using the vectorization operator $\vect(\cdot)$, Equation \eqref{eq:matrix_bias} can equivalently expressed as
\begin{equation}
\E\big(\vect(\widehat{\boldsymbol{\Gamma}})\big) = \mathbf{W}_{\underline{n}}^* \, \vect(\boldsymbol{\Gamma}),
\label{eq:linear_bias}
\end{equation}
where the $(L_{\un n}\times L_{\un n})$ weight matrix
\begin{equation}\label{eq:W_n_star}
\mathbf{W}_{\underline{n}}^* = \mathbf{I}_{L_{\un n}} - \mathbf{A}_2^* \otimes \mathbf{A}_1^* - \mathbf{B}_2^* \otimes \mathbf{B}_1^* + \mathbf{D}_2^* \otimes \mathbf{D}_1^*
\end{equation}
is a rank-deficient deterministic matrix with $\textup{rank}(\mathbf{W}_{\underline n}^*) \leq L_{\un n} - 1$ that only depends on the grid sizes $(n_1,n_2)$. Further, the joint bias vector $\un{b}_{\vect(\widehat{\boldsymbol{\Gamma}})}$ of $\vect(\widehat{\boldsymbol{\Gamma}})$ becomes
\begin{align}\label{eq:bias_naive_vector}
\un{b}_{\vect(\widehat{\boldsymbol{\Gamma}})} := \E\big(\vect(\widehat{\boldsymbol{\Gamma}})\big) - \vect(\boldsymbol{\Gamma}) =  \left(\mathbf{W}_{\underline{n}}^*-\mathbf{I}_{L_{\un n}}\right)\vect(\boldsymbol{\Gamma}).
\end{align}
Moreover, under Assumption~\ref{ass:abs_summability}, the bias \eqref{eq:bias_naive_vector} is uniformly at most of order $1/N$. That is, we have
\begin{equation}\label{eq:naive_bias_bound}
\bigl\lVert \un{b}_{\vect(\widehat{\boldsymbol{\Gamma}})}\bigr\rVert_\infty
\leq\frac{1}{N}\bigl\lVert\vect(\boldsymbol{\Gamma})\bigr\rVert_1,
\end{equation}
where, for a vector $\un{v}\in\mathbb{R}^d$,
$\lVert\un{v}\rVert_\infty:=\max_{1\leq i\leq d}|v_i|$ denotes its largest absolute component, whereas
$\lVert\un{v}\rVert_1:=\sum_{i=1}^d|v_i|$ denotes the sum of the absolute values of its components.
\end{theorem}

Some remarks on the matrices $\mathbf{A}_i$, $\mathbf{B}_i$ and $\mathbf{D}_i$, $i \in \{1,2\}$ from \eqref{eq:counting_matrices} and the weight matrix $\mathbf{W}_{\underline{n}}^*$ \eqref{eq:W_n_star} and their properties are in order.

\begin{remark}[On $\mathbf{A}_i$, $\mathbf{B}_i$ and $\mathbf{D}_i$, and $\mathbf{W}^*_{\underline{n}}$ ] \ 
\begin{enumerate}[label=(\roman*)]
    \item All matrices $\mathbf{A}_i$, $\mathbf{B}_i$ and $\mathbf{D}_i$ are ``counting matrices'' that have non-negative and integer-valued entries. All entries of $\mathbf{D}_i$ are strictly positive.
    \item Despite 
    $\mathbf{W}_{\underline{n}}^*$ may appear as a diagonally dominant matrix, for arbitrary sample sizes $(n_1,n_2)$, all its row sums equal zero. For the case of $n_1=n_2=4$, a heatmap of the resulting $(49\times 49)$ weight matrix $\mathbf{W}_{\underline{n}}^*$ is shown in Figure~\ref{fig:W_n_star_comparison} (left panel).
    \item By construction, the entries of $\mathbf{N}_i$ are exactly the row sums of $\mathbf{A}_i$ and of $\mathbf{B}_i$. That is, for each fixed $h_i$, summing $\nu_{n_i}(r_i, h_i)$ over $r_i$ recovers $n_i(n_i-|h_i|)$ admissible index pairs for direction $i$, and correspondingly $n_i^2$ when $h_i=0$, which is the constant row sum of $\mathbf{D}_i$. 
    Hence, the normalized count matrices $\mathbf{A}_i^*$, $\mathbf{B}_i^*$ and $\mathbf{D}_i^*$ from \eqref{eq:normalised_counting_matrices_form_rh} are row-stochastic. That is, we have
\begin{equation*}\label{eq:row_stochastic_norm_count}
\mathbf{A}_i^*\underline{\mathbbm{1}}_i=\mathbf{B}_i^*\underline{\mathbbm{1}}_i=\mathbf{D}_i^*\underline{\mathbbm{1}}_i=\underline{\mathbbm{1}}_i, \qquad \textup{for} \quad  i\in\{1,2\},
\end{equation*}
where $\underline{\mathbbm{1}}_i=(1,1,\dots, 1)^\top$ denotes the $(L_i \times 1)$ vector of ones.
\end{enumerate}
\end{remark}

We illustrate further the introduced quantities and the derived formulas with a toy example for the case of a Gaussian spatial lattice process that was used already in Example \ref{example_Gaussian}.

\begin{ex}[Joint bias formulas for spatial lattice process]\label{example_Gaussian_bias}
Let $\mathbf{X}$ be the Gaussian spatial lattice process from Example \ref{example_Gaussian} with exponential covariance function $C_\lambda(\hu)$ defined in \eqref{eq:trueautocov} for some $\lambda\in\{2,4,6,8\}$ and with $\beta\in [0,1]$. Further, suppose $n_1=n_2=4$, that is, $\mathbf{X}_{\underline{n}}$ is a $(4\times 4)$ data matrix and $\boldsymbol{\Gamma}$ defined in \eqref{eq:Gamma_pq} becomes a $(7\times 7)$ matrix of covariances. Figure~\ref{fig:true_gamma_all} illustrates different instances of $\boldsymbol{\Gamma}$, where the panels show varying levels of spatial dependence, controlled by the parameter $\lambda$, for both separable ($\beta=0$) and non-separable ($\beta=1$) settings. For this setup, the $(7\times 7)$ raw counting matrices $\mathbf{A}_1=\mathbf{A}_2$, $\mathbf{B}_1=\mathbf{B}_2$ and $\mathbf{D}_1=\mathbf{D}_2$ are shown in Figure \ref{fig:ABD_example}. Moreover, in Figure \ref{fig:W_n_star_comparison} (left panel), we show a heatmap of the corresponding $(49\times 49)$ weight matrix $\mathbf{W}_{\underline{n}}^*$.
\end{ex}

\begin{figure}[t]
\centering

\captionsetup[subfigure]{skip=6pt}

\begin{minipage}[c]{0.06\textwidth}
  \centering
  \rotatebox{90}{$\beta = 0$}
\end{minipage}\hfill
\begin{minipage}[c]{0.92\textwidth}
  \centering

  \begin{subfigure}[t]{0.24\textwidth}
    \centering
    \includegraphics[width=\linewidth,height=0.18\textheight,keepaspectratio]{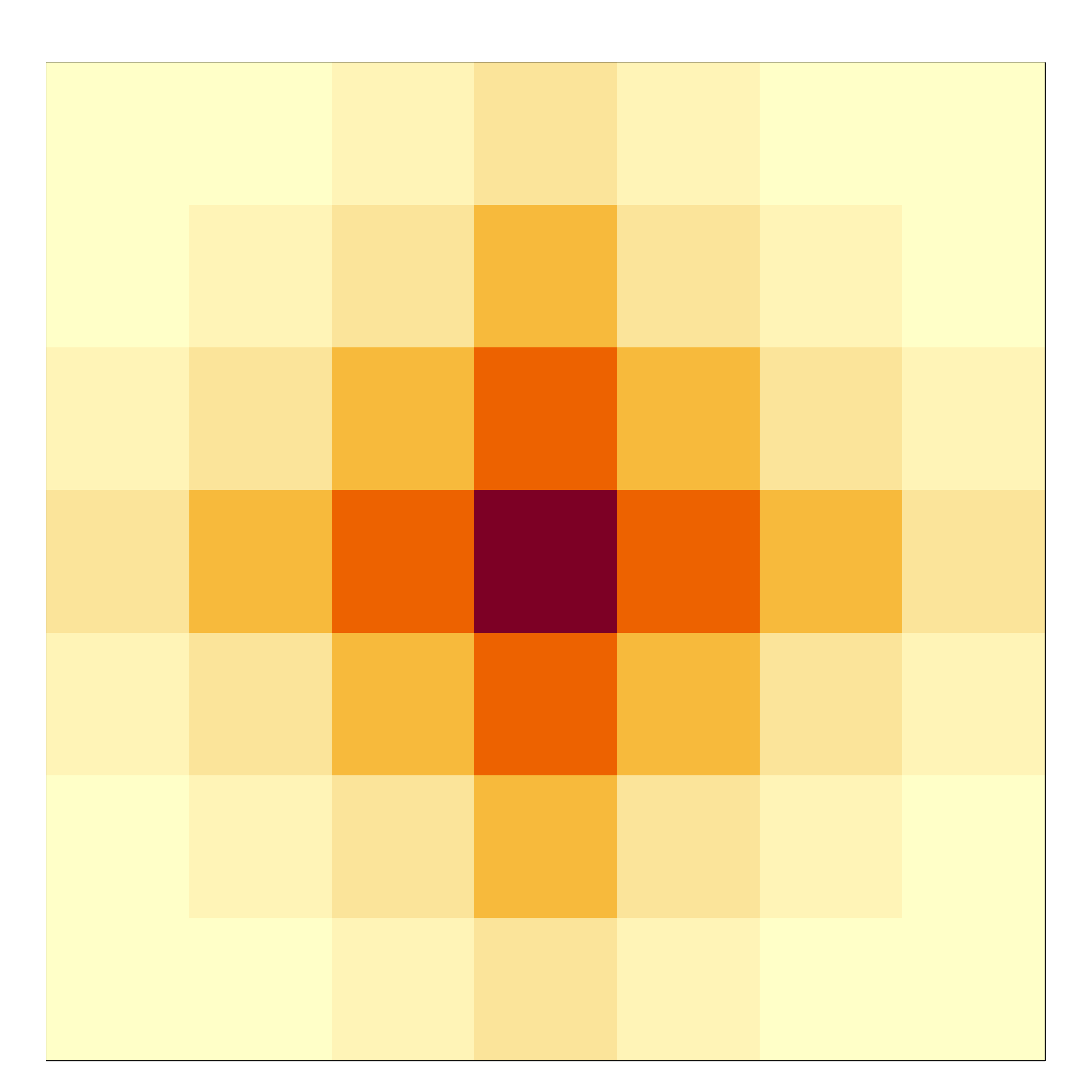}
    \caption{$\lambda=2$}
  \end{subfigure}\hfill
  \begin{subfigure}[t]{0.24\textwidth}
    \centering
    \includegraphics[width=\linewidth,height=0.18\textheight,keepaspectratio]{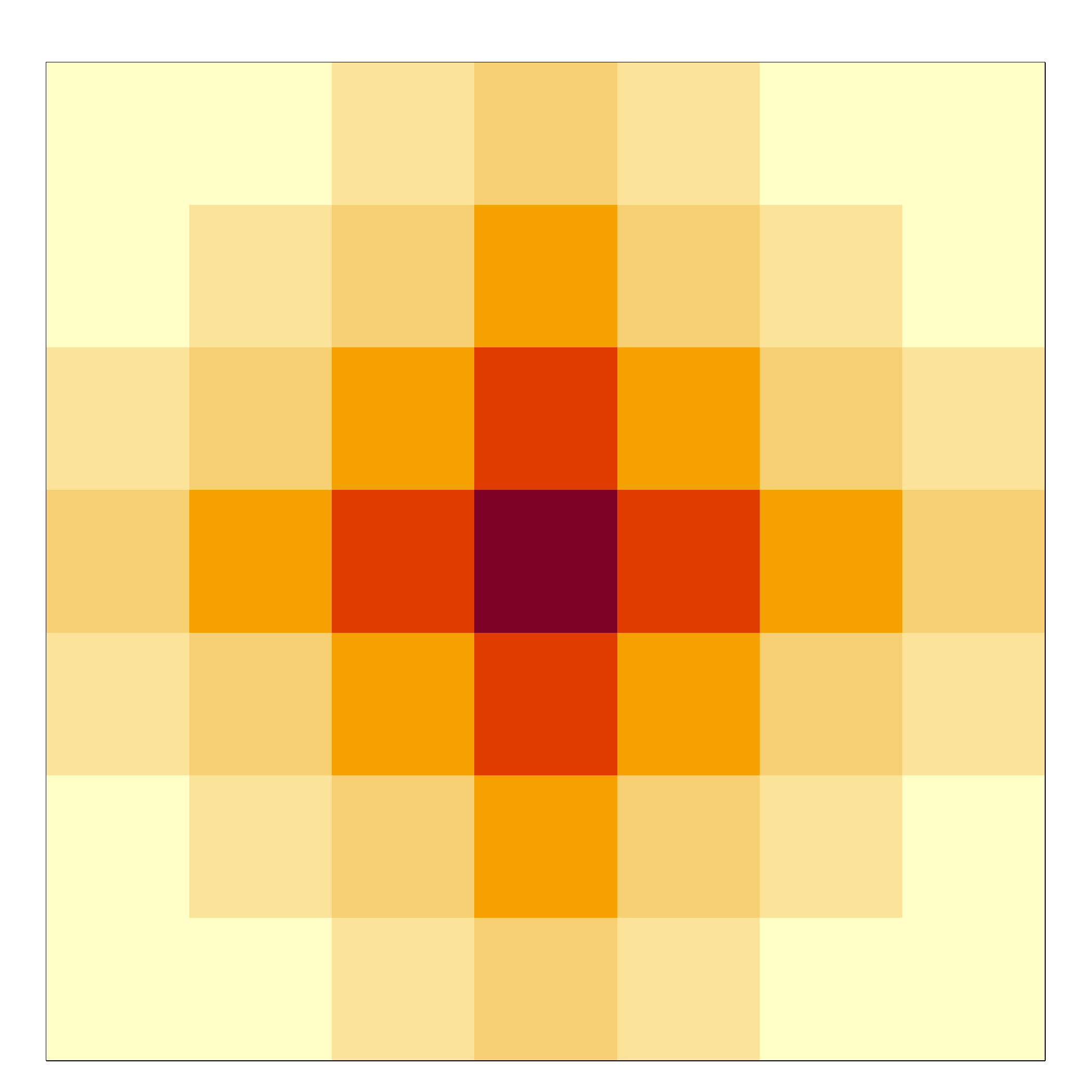}
    \caption{$\lambda=4$}
  \end{subfigure}\hfill
  \begin{subfigure}[t]{0.24\textwidth}
    \centering
    \includegraphics[width=\linewidth,height=0.18\textheight,keepaspectratio]{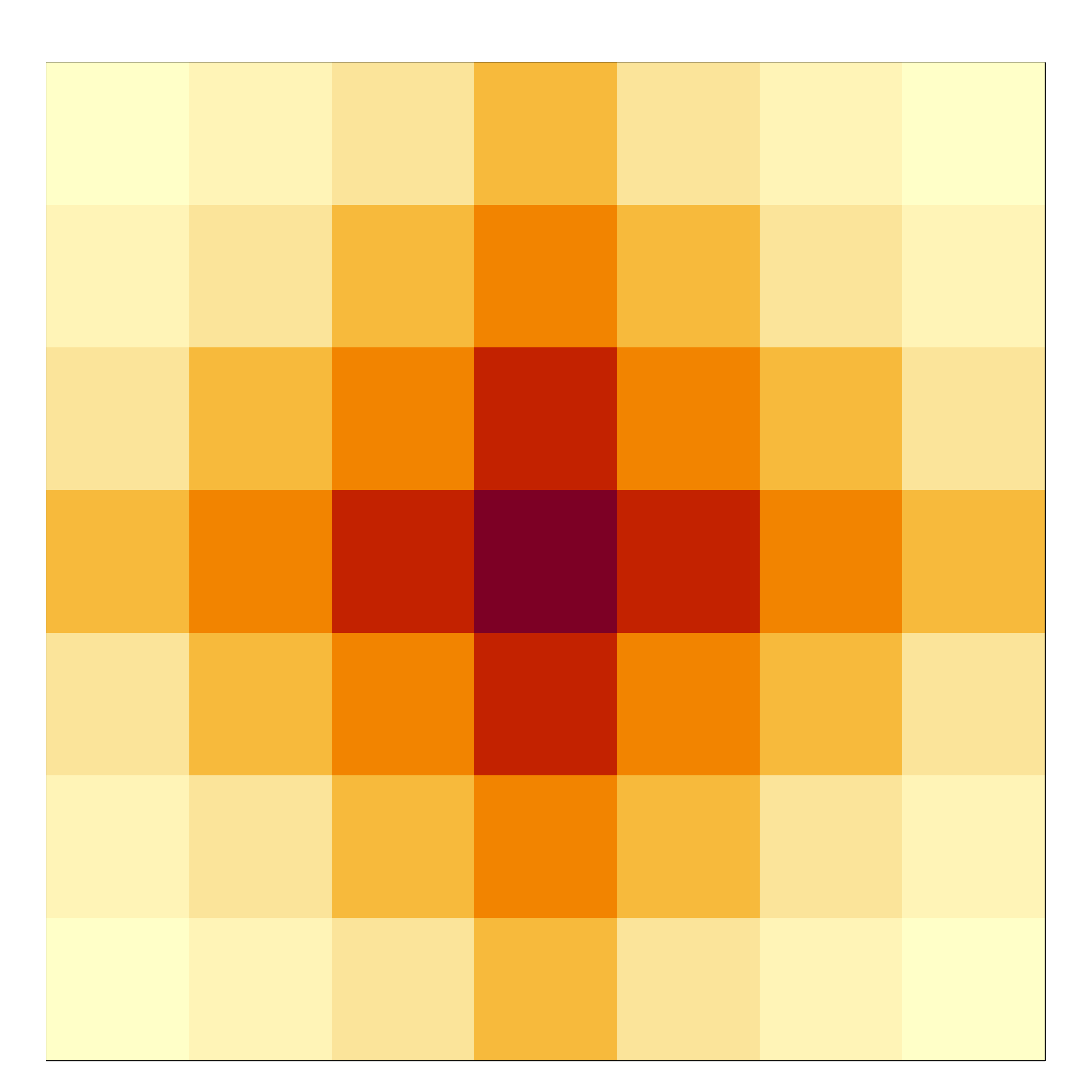}
    \caption{$\lambda=6$}
  \end{subfigure}\hfill
  \begin{subfigure}[t]{0.24\textwidth}
    \centering
    \includegraphics[width=\linewidth,height=0.18\textheight,keepaspectratio]{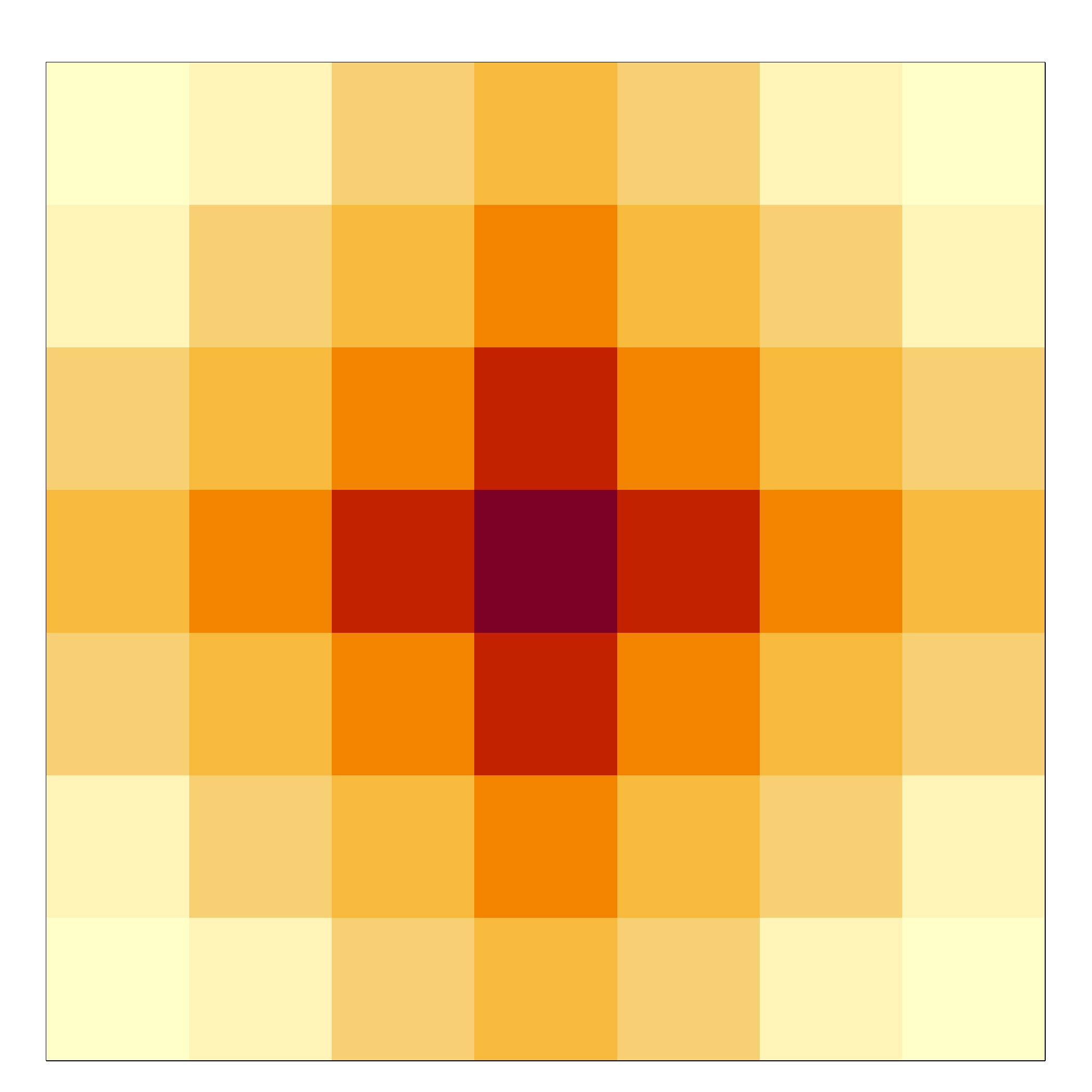}
    \caption{$\lambda=8$}
  \end{subfigure}

\end{minipage}

\vspace{1em}

\begin{minipage}[c]{0.06\textwidth}
  \centering
  \rotatebox{90}{$\beta = 1$}
\end{minipage}\hfill
\begin{minipage}[c]{0.92\textwidth}
  \centering

  \begin{subfigure}[t]{0.24\textwidth}
    \centering
    \includegraphics[width=\linewidth,height=0.18\textheight,keepaspectratio]{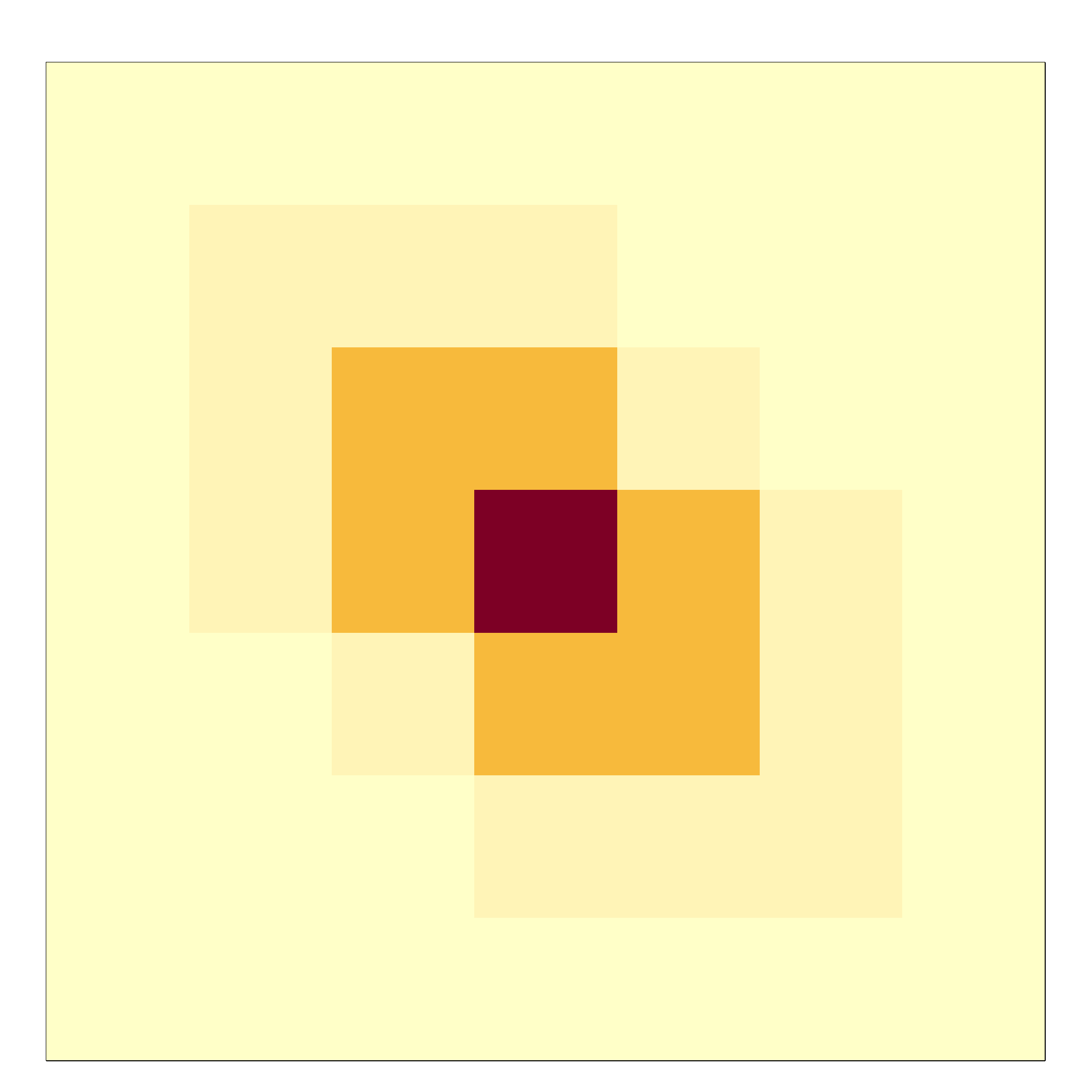}
    \caption{$\lambda=2$}
  \end{subfigure}\hfill
  \begin{subfigure}[t]{0.24\textwidth}
    \centering
    \includegraphics[width=\linewidth,height=0.18\textheight,keepaspectratio]{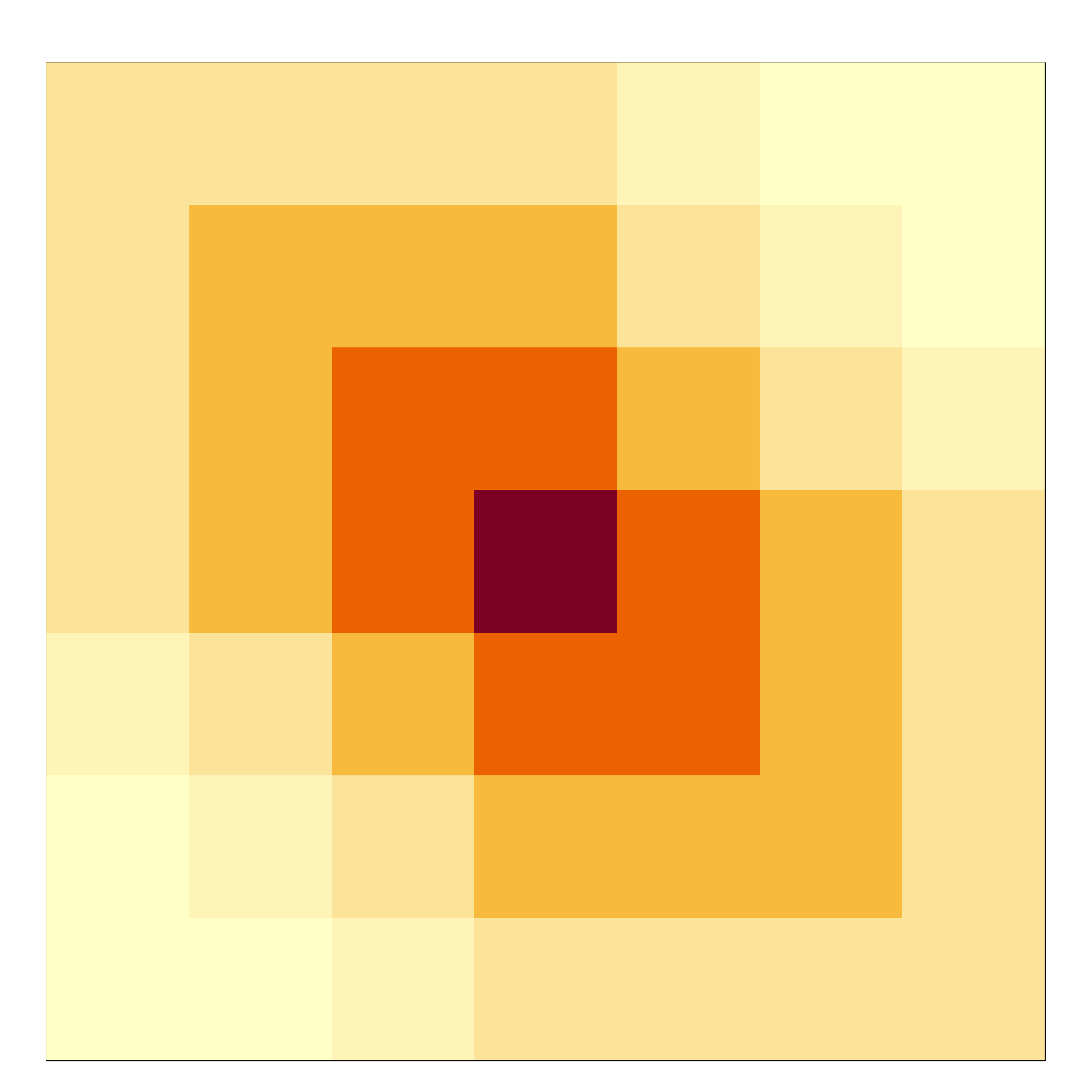}
    \caption{$\lambda=4$}
  \end{subfigure}\hfill
  \begin{subfigure}[t]{0.24\textwidth}
    \centering
    \includegraphics[width=\linewidth,height=0.18\textheight,keepaspectratio]{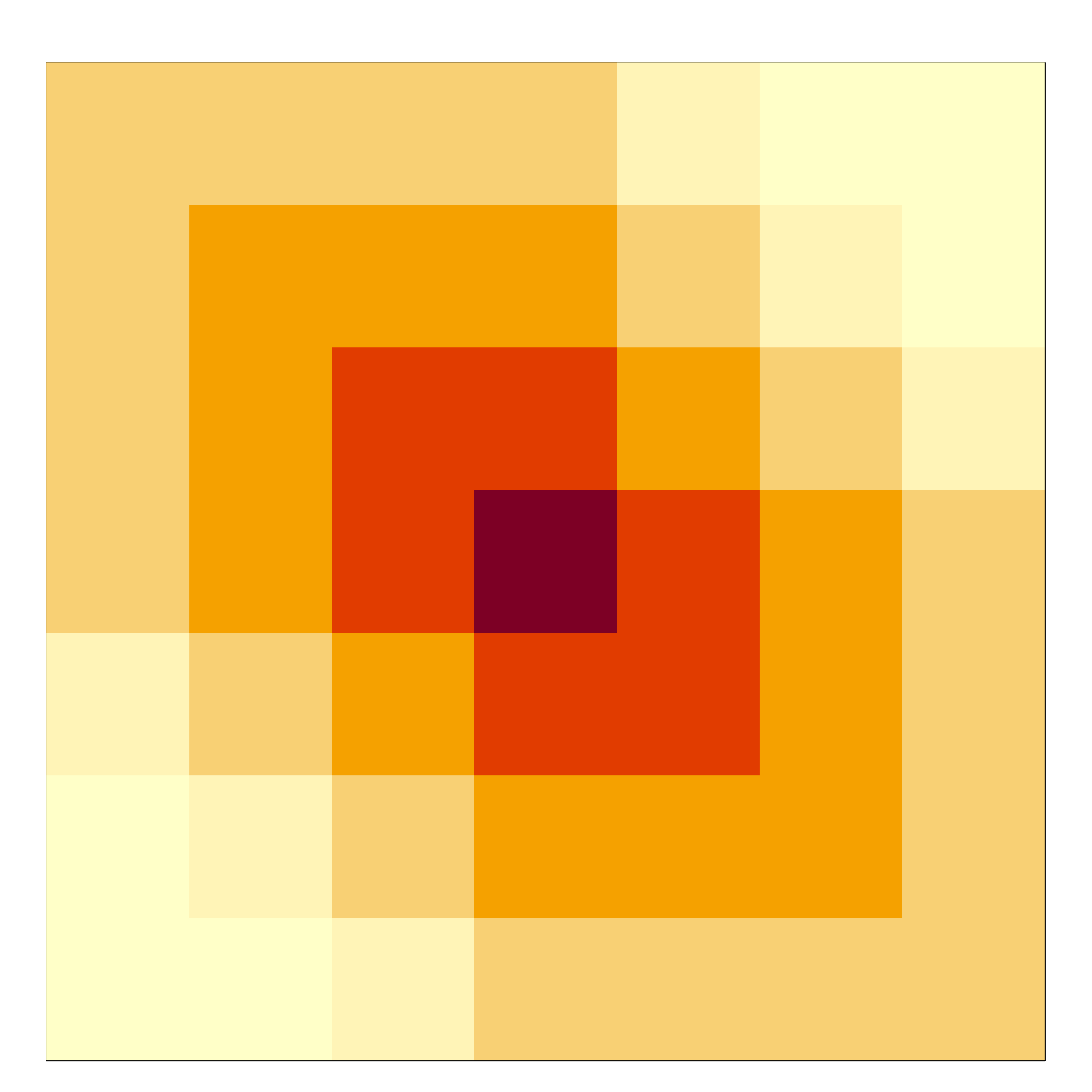}
    \caption{$\lambda=6$}
  \end{subfigure}\hfill
  \begin{subfigure}[t]{0.24\textwidth}
    \centering
    \includegraphics[width=\linewidth,height=0.18\textheight,keepaspectratio]{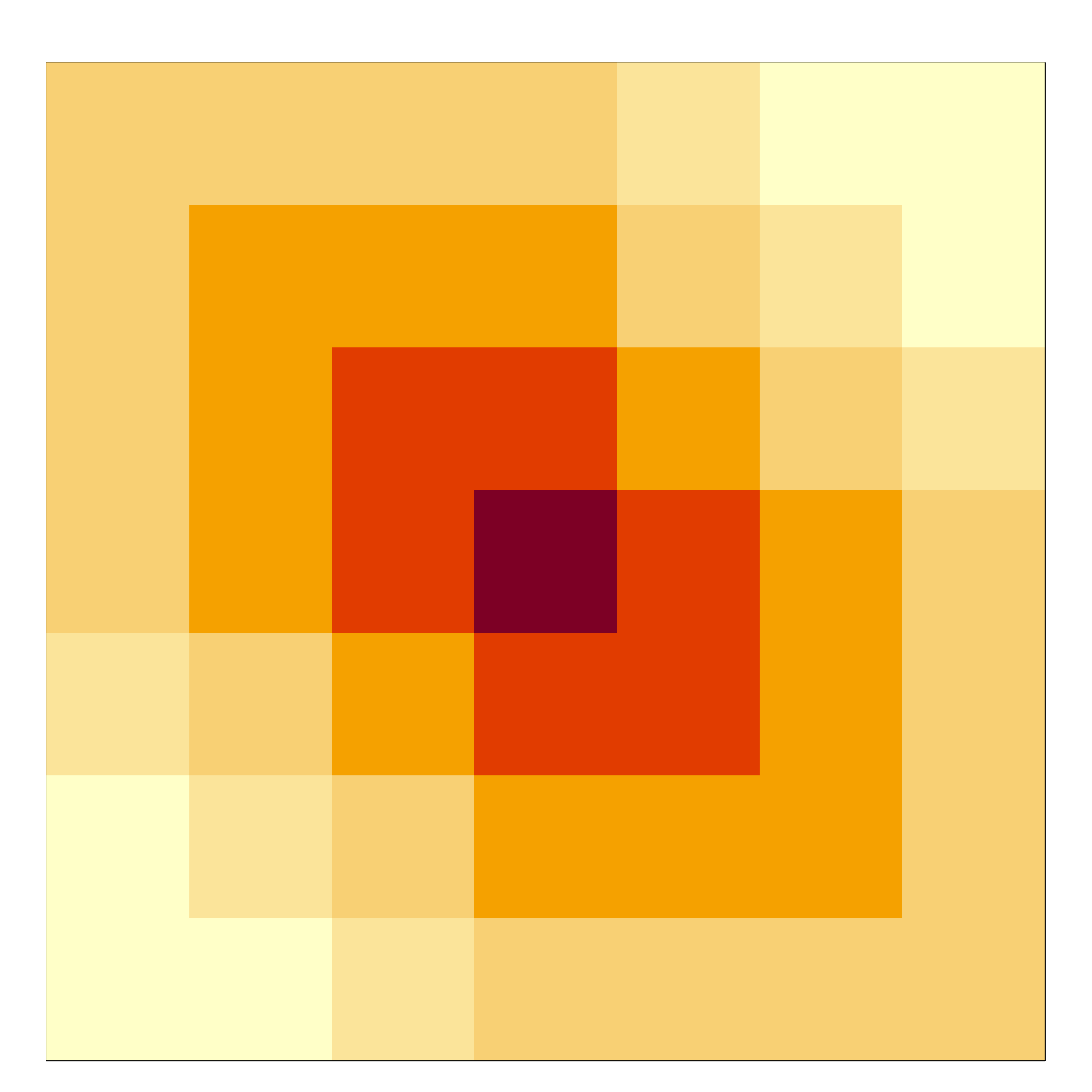}
    \caption{$\lambda=8$}
  \end{subfigure}
\end{minipage}

\caption{Covariance matrices $\mathbf{\Gamma}$ from \eqref{eq:Gamma_pq} of size $(7\times 7)$ for $\mathbf{X}_{\underline{n}}$ with $n_1=n_2=4$ from a spatial (Gaussian) process $\mathbf{X}$ with exponential covariance function $C_\lambda$ defined in \eqref{eq:trueautocov}. $\mathbf{\Gamma}$ is shown for separable ($\beta=0$; top row) and non-separable ($\beta=1$; bottom row) dependence with increasing spatial dependence $\lambda\in\{2,4,6,8\}$.}
\label{fig:true_gamma_all}
\end{figure}

\subsection{Joint bias of a matrix of the constant prefactor sample covariances}
\label{subsec:bias_matrix_sample_covariance_const}

The characterization of the joint mean and bias for the standard covariance estimator $\widehat C(\un h)$ in Theorem \ref{prop:joint_bias} can be easily transferred also to the case of the constant prefactor estimator from $\widetilde C(\un h)$. This is because the expectations of the two estimators are linked by a lag dependent scaling $\E(\widetilde C(\hu)) = N_h/N\E(\widehat{C}(\hu))$, see \eqref{eq:connection_standard_constant}. This property can be used to define diagonal $(L_i\times L_i)$ matrices $\boldsymbol{\Delta}_i$, $i\in\{1,2\}$ that encode the relationship jointly. We define
\begin{equation}
\label{eq:bartlett_matrices}
    \boldsymbol{\Delta}_i := \diag\left(1-\frac{|p-n_i|}{n_i},\; p=1, \dots, L_i\right) = \diag\left(1-\frac{|h_i|}{n_i},\; h_i \in I(1-n_i:n_i-1)\right).
\end{equation}
That is, these matrices essentially collect the Bartlett windows \citep{bartlett1950periodogram} on their main diagonal. Hence, we refer to $\boldsymbol{\Delta}_i$, $i\in\{1,2\}$ as \emph{Bartlett window matrices}.
Since the \emph{total} and \emph{effective} sample size factorizes across the two lattice dimensions, i.e., $N=n_1n_2$ and $N_h=(n_1-|h_1|)(n_2-|h_2|)$, their ratio factorizes as well. That is, for $(h_1,h_2) = (p-n_1,q-n_2)$, we have
\begin{equation*}
    \frac{N_h}{N}=\frac{n_1-|h_1|}{n_1}\cdot\frac{n_2-|h_2|}{n_2}  = \left[\boldsymbol{\Delta}_1\right]_{p,p}\cdot \left[\boldsymbol{\Delta}_2\right]_{q,q}, 
\end{equation*}
such that we can exploit the above together with $\E(\widetilde C(\hu)) = N_h/N\E(\widehat{C}(\hu))$ to
connect the two estimators (for the whole covariance matrix) in the following corollary.

\begin{corollary}[Joint mean and bias of $\vect(\widetilde{\boldsymbol{\Gamma}})$]\label{cor:joint_bias_const}
Let $\widetilde{\boldsymbol{\Gamma}}$ be the matrix of constant prefactor
sample covariances that is defined as $\widehat{\boldsymbol{\Gamma}}$ below \eqref{eq:Gamma_pq}, but uses $\widetilde{C}$ from \eqref{eq:sampleAutocov_const} instead of
$\widehat{C}$. Then, we have
\begin{equation}\label{eq:expectation_const_est}
\E\bigl(\widetilde{\boldsymbol{\Gamma}}\bigr)
=
\boldsymbol{\Delta}_1
\E\bigl(\widehat{\boldsymbol{\Gamma}}\bigr)
\boldsymbol{\Delta}_2,
\qquad
\E\bigl(\vect(\widetilde{\boldsymbol{\Gamma}})\bigr)
=
\widetilde{\mathbf{W}}^{*}_{\underline{n}}
\vect(\boldsymbol{\Gamma}).
\end{equation}
with joint bias vector 
\begin{align}\label{eq:const_bias_vector}
b_{\vect(\widetilde{\boldsymbol{\Gamma}})}
:= \E\bigl(\vect(\widetilde{\boldsymbol{\Gamma}})\bigr) - \vect(\boldsymbol{\Gamma})
= \bigl(\widetilde{\mathbf{W}}^{*}_{\underline{n}} - \mathbf{I}_{L_1 L_2}\bigr)\vect(\boldsymbol{\Gamma}).
\end{align}
Here, the $(L_{\un n} \times L_{\un n})$ constant-prefactor weight matrix
$\widetilde{\mathbf{W}}^{*}_{\underline{n}}$ can be directly linked to
$\widehat{\mathbf{W}}^{*}_{\underline{n}}$ in \eqref{eq:W_n_star} by using the
Bartlett window matrices $\boldsymbol{\Delta}_i$ defined in \eqref{eq:bartlett_matrices}. Precisely, we have
\begin{align}
\widetilde{\mathbf{W}}_{\underline{n}}^* &:= (\boldsymbol{\Delta}_2\otimes\boldsymbol{\Delta}_1)\,\mathbf{W}_{\underline{n}}^*
\label{eq:W_n_tilde}\\
&= (\boldsymbol{\Delta}_2\otimes\boldsymbol{\Delta}_1)
- (\boldsymbol{\Delta}_2\mathbf{A}_2^*)\otimes(\boldsymbol{\Delta}_1\mathbf{A}_1^*)
- (\boldsymbol{\Delta}_2\mathbf{B}_2^*)\otimes(\boldsymbol{\Delta}_1\mathbf{B}_1^*)
+ (\boldsymbol{\Delta}_2\mathbf{D}_2^*)\otimes(\boldsymbol{\Delta}_1\mathbf{D}_1^*),\notag
\end{align}
where  $\boldsymbol{\Delta}_i\mathbf{A}_i^* = n_i^{-2}\mathbf{A}_i$ and $\boldsymbol{\Delta}_i\mathbf{B}_i^* = n_i^{-2}\mathbf{B}_i$, while $\boldsymbol{\Delta}_i\mathbf{D}_i^*=\diag\big((n_i-|h_i|)/n_i^{3}\big)\mathbf{D}_i$ and $\widetilde{\mathbf{W}}^*_{\underline{n}}$ inherits the same rank deficiency as $\mathbf{W}_{\underline{n}}^*$, i.e., $\textup{rank}(\widetilde{\mathbf{W}}^*_{\underline{n}} ) \leq L_{\un n}-1$.
\end{corollary}

\subsection{Exact finite sample second-order structure of the sample covariances}\label{subsec:finite_second_order}

The bias analysis in the previous section was distribution free as it relied only on the weak stationarity Assumption~\ref{ass:stationarity}. In this section, we present also the second-order structure of single covariance estimators $\widehat{C}(\un h)$ and $\widetilde{C}(\un h)$. For this purpose, some information on the (finite) fourth-order moments is required to describe the variability of the estimators. In the following, we impose Gaussianity to facilitate the derivation and the presentation of the resulting exact formulas for the covariances of sample covariance estimators.
\begin{assumption}[Gaussianity]\label{ass:gaussianity} The spatial lattice process $\mathbf{X}$ is Gaussian.
\end{assumption}
    
Focusing first on the standard estimator $\widehat{C}(\un h)$, Isserlis' theorem \citep{Isserlis1918} applied to the centered random variables of $\mathbf{X}$ leads to the exact formula presented in the following theorem.    
\begin{theorem}[Exact finite-sample covariances of the standard covariance estimator]\label{prop:covofcov}
Let Assumptions \ref{ass:stationarity} and \ref{ass:gaussianity} hold and let $\hu = (h_1,h_2)$ and $\un g=(g_1, g_2)$ be two lag vectors in $H_{\un n}$ \eqref{eq:lag_set} with effective sample sizes $N_h=(n_1-|h_1|)(n_2-|h_2|)$ and $N_g=(n_1-|g_1|)(n_2-|g_2|)$, respectively. Then, using the centered form representation based on $C_Z$ defined in $\eqref{eq:C_centered}$, the covariances of two standard sample covariance estimators $\widehat C(\hu)$ and $\widehat C(\gu)$ defined in \eqref{eq:sampleAutocov} can be represented as 
\begin{equation}\label{eq:cov_of_cov}
    \Cov\bigl(\widehat C(\hu),\widehat C(\gu)\bigr)
    =\frac{1}{N_hN_g}
     \sum_{\tu\in\nhset} \sum_{\su\in\ngset} \left[\Cctr(\tu+\hu,\su+\gu)\,\Cctr(\tu,\su)+\Cctr(\tu+\hu,\su)\,\Cctr(\tu,\su+\gu)\right].
    \end{equation}
    Alternatively, the lag form representation with $\Cbar$ and $D$ from \eqref{eq:Cbar} and \eqref{eq:D} and the shift $\du:=\su-\tu$ becomes    \begin{align}
    & \Cov\bigl(\widehat C(\hu),\widehat C(\gu)\bigr)   \label{eq:cov_of_cov_lag}   \\
    =&\frac{1}{N_{\hu}N_{\gu}}\sum_{\tu\in\nhset}\ \sum_{\su\in\ngset}\Bigl\{\bigl[C(\du+\gu-\hu)-\Cbar(\tu+\hu)-\Cbar(\su+\gu)+D\bigr]
    \bigl[C(\du)-\Cbar(\tu)-\Cbar(\su)+D\bigr]\notag \\
    &\quad\quad\quad+\bigl[C(\du-\hu)-\Cbar(\tu+\hu)-\Cbar(\su)+D\,\bigr]
     \bigl[C(\du+\gu)-\Cbar(\tu)-\Cbar(\su+\gu)+D\bigr]\Bigr\}. \notag
    \end{align}
    Notably, the right-hand sides of \eqref{eq:cov_of_cov} and \eqref{eq:cov_of_cov_lag} does not depend on the mean $\mu$.
    \end{theorem}
A remark on Theorem \ref{prop:covofcov} and the interpretation of the exact representation of the (finite sample) covariance of two standard sample covariance estimators is in order.

\begin{remark}[On the representation of $\Cov(\widehat C(\hu),\widehat C(\gu))$ in \eqref{eq:cov_of_cov}]\label{remark:cov_of_cov}\   
\begin{enumerate}[label=(\roman*)]
    \item If the true mean $\mu$ is known and replaces the sample mean $\Xbar$ in the standard estimators $\widehat{C}(\hu)$ (leading to $\widehat{C}_{\mu}(\hu)$) in \eqref{eq:cov_of_cov},
   $C_Z(\un t,\un s)$ collapses to $C(\un t-\un s)$ and \eqref{eq:cov_of_cov_lag} becomes\begin{equation}\label{eq:bartlett_cov-of-cov}
        \frac{1}{N_{\hu}N_{\gu}}\sum_{\tu\in\nhset}\ \sum_{\su\in\ngset}
\left[C(\du+\gu-\hu)C(\du)+C(\du-\hu)C(\du+\gu)\right],
    \end{equation}
which is the two-dimensional lattice version of the Bartlett formula  \citep{Bartlett1946}, which gives the covariance of sample covariances in the univariate time series setting. That is, all other terms in \eqref{eq:cov_of_cov} arise from estimating $\mu$ by $\overline{X}$ in $\widehat C(\hu)$ as defined in \eqref{eq:sampleAutocov}.
\end{enumerate}
\end{remark}

Again, using the simple link between the two sample covariances $\widetilde C(\hu)=N_{\hu}/N\,\widehat C(\hu)$, we can state the second-order structure of the constant prefactor estimator in the following corollary.

\begin{corollary}[Exact finite sample covariance of the constant estimator]\label{cor:cov_const}
Under the conditions of Theorem~\ref{prop:covofcov}, and with $\widetilde C$ from  \eqref{eq:sampleAutocov_const}, we have
\begin{equation}\label{eq:cov_of_cov_const}
\begin{aligned}
\Cov\bigl(\widetilde C(\hu),\widetilde C(\gu)\bigr)
&=\frac{N_{\hu}N_{\gu}}{N^{2}}\Cov\bigl(\widehat C(\hu),\widehat C(\gu)\bigr)\\[2pt]
&=\frac{1}{N^{2}}\sum_{\tu\in\nhset}\ \sum_{\su\in\ngset}
\Bigl[\Cctr(\tu+\hu,\su+\gu)\Cctr(\tu,\su)+\Cctr(\tu+\hu,\su)\Cctr(\tu,\su+\gu)\Bigr].
\end{aligned}
\end{equation}
\end{corollary}

\section{Bias-corrected covariance estimation}\label{sec:mainresult}

In this section, we use the linear bias representation derived in the previous Section \ref{sec:exact_bias} to construct a framework for bias-corrected estimation of covariances from stationary spatial lattice data. Indeed, Equation~\eqref{eq:linear_bias} suggests a natural route for bias correction. If the deterministic weight matrix $\mathbf{W}_{\underline{n}}^*$ \emph{were} invertible, we could obtain an unbiased estimator for the sample covariances by left-multiplying $\vect(\widehat{\boldsymbol{\Gamma}})$ with $(\mathbf{W}_{\underline{n}}^*)^{-1}$. 

However, according to Proposition \ref{prop:joint_bias}, the deterministic weight matrix $\mathbf{W}_{\underline{n}}^*$ from \eqref{eq:W_n_star} is singular. Hence, since the inverse of the full weight matrix does not exist, we cannot recover such a correction directly. To remedy this issue, we propose an adaptation of the methods proposed by \citet{vogelsang2016exactly} for univariate time series and extended by \citet{li2024correcting} to multivariate time series to the spatial lattice case. That is, instead of using the whole sample covariance matrix $\widehat{\boldsymbol{\Gamma}}$, we consider suitable lower-dimensional (central) 
submatrices, denoted by $\widehat{\boldsymbol{\Gamma}}_{\underline{m}}$, 
of $\widehat{\boldsymbol{\Gamma}}$. As this also reduces the dimension of the corresponding weight matrix, it becomes invertible.

\subsection{Bias correction for submatrices of the matrix of sample covariances} \label{subsec:bias_correction}

Recall that $L_i = 2n_i-1$ for $i \in \{1,2\}$, and $L_{\un n } = L_1L_2$. We define $\underline{m}=(m_1,m_2)$ to be a vector of \emph{cutoff parameters} with $m_i\leq n_i-2$ and set $M_i = (2m_i+1)$, for $i \in \{1,2\}$, and $M_{\un m}:= M_1M_2$. In the following, we split the full lag set $H_{\underline{n}}$ from \eqref{eq:lag_set} into the set of \emph{selected} (or \emph{within-window}) lags 
\begin{equation*}
    \label{eq:m_lag_set}
    W_{\underline{m}} := \{\un h \in \mathbb{Z}^2: |h_i|\leq m_i, \; i \in \{1,2\} \} \subset H_{\underline{n}},\quad |W_{\underline{m}}| = M_{\un m}
\end{equation*}
and the set of the \emph{non-selected} (or \emph{out-of-window}) lags
\begin{equation*}\label{eq:not_selected_lags}
  W_{\underline{m}}^c:= H_{\underline{n}} \setminus W_{\underline{m}}, \qquad |W_{\underline{m}}^c|=L_{\un n}-M_{\un m}.
\end{equation*}
For $i \in\{1,2\}$, let the  $\mathbf{P}_{i,m_i}\in \mathbb{R}^{M_i\times L_i}$ be the matrix that retains the $M_i$ central rows of $\mathbf{I}_{L_i}$. We call them \emph{partition matrices} and their entries become
\begin{equation}\label{eq:partition_matrix} 
[\mathbf{P}_{i,m_i}]_{p,q} =
\mathbbm{1}\{q=n_i-m_i+p-1\},
\qquad p=1,\ldots,M_i,\quad q=1,\ldots,L_i,
\end{equation}
where $\mathbbm{1}(\cdot)$ is the indicator function.
In this way, the $p$-th row of $\mathbf{P}_{i,m_i}$ selects the lag $h_i = p-m_i-1$ and the partition matrices have a single one in each row and no repeated column. Analogously, let $\mathbf{P}^{c}_{i,m_i}\in \mathbb{R}^{(L_i-M_i)\times L_{i}}$ be the complementary matrices
that retain the remaining $L_{i}- M_{i}$ rows of $\mathbf{I}_{L_i}$ in increasing order, that is,
\begin{equation}\label{eq:partition_matrix_c}
[\mathbf{P}^{c}_{i,m_i}]_{p,q}
=
\begin{cases}
1, & q=p,
     \qquad\qquad\;\, p=1,\dots,n_i-m_i-1,\\[6pt]
1, & q=p+M_i,
     \qquad p=n_i-m_i,\dots,K_i,\\[6pt]
0, & \textup{otherwise}.
\end{cases}
\end{equation}
By construction, $\mathbf{P}_{i,m_i}$ and $\mathbf{P}^{c}_{i,m_i}$ together consist of all rows of $\mathbf{I}_{L_i}$.
We apply \eqref{eq:partition_matrix} on both sides of $\boldsymbol{\Gamma}$ and $\boldsymbol{\widehat{\Gamma}}$ defined in \eqref{eq:Gamma_pq} 
to extract an $(M_1 \times M_2)$ principal submatrix centered around the lag $(0,0)$ that contains all the covariances with lags $\hu \in W_{\underline{m}}$. That is, we define
\begin{equation}\label{eq:gammas_m}
\widehat{\mathbf{\Gamma}}_{\underline{m}} :=\mathbf{P}_{1,m_1}\,\widehat{\mathbf{\Gamma}}\,\mathbf{P}_{2,m_2}^{\top},\quad \mathbf{\Gamma}_{\underline{m}} :=\mathbf{P}_{1,m_1}\,\mathbf{\Gamma}\,\mathbf{P}_{2,m_2}^{\top}.
\end{equation}
In Figure \ref{fig:selection_example}, we give examples of $\widehat{\mathbf{\Gamma}}_{\underline{m}}$ for $n_1,n_2=4$ and three choices of cutoff parameters $m_1=m_2=0$, $m_1=m_2=1$ and $m_1=m_2=2$.

For the vectorized versions of \eqref{eq:gammas_m}, we exploit the vectorization property of the Kronecker product and combine the partition matrices into a single \emph{selection matrix} $\mathbf{S}_{\underline{m}}$ and its complement $\mathbf{S}^{c}_{\un m}$. Precisely, we define\begin{equation}\label{eq:selection_matrices}
  \mathbf{S}_{\underline{m}}:= (\mathbf{P}_{2,m_2}\otimes\mathbf{P}_{1,m_1}) \in \mathbb{R}^{M_{\un m}\times L_{\un n}}, \quad \mathbf{S}^{c}_{\un m}:=\begin{pmatrix}\mathbf{I}_{L_2}\otimes\mathbf{P}^{c}_{1,m_1}\\[2pt]
\mathbf{P}^{c}_{2,m_2}\otimes\mathbf{P}_{1,m_1}
\end{pmatrix}
   \in \R^{(L_{\un n}-M_{\un m})\times L_{\un n}}.
 \end{equation}
 where the two blocks of $\mathbf{S}^{c}_{\un m}$ in
\eqref{eq:selection_matrices} correspond to the disjoint decomposition of $W^{c}_{\un m}$ into the lags
with $|h_1|>m_1$, for which $h_2$ is unrestricted, and those with $|h_1|\le m_1$ and $|h_2|>m_2$.
 Since, by construction, $\mathbf{S}_{\un m}$ and $\mathbf{S}_{\un m}^c$ consist of distinct rows of $\mathbf{I}_{L_i}$, exhausting them together, they have orthonormal rows and mutually orthogonal row spaces. That is, we have
\begin{equation}\label{eq:selection_properties}
\mathbf{S}_{\un m}\mathbf{S}_{\un m}^\top = \mathbf{I}_{M_{\un m}},
\qquad
\mathbf{S}_{\un m}\mathbf{S}^{c\top}_{\un m} = \mathbf{0},
\qquad
\mathbf{S}_{\un m}^\top\mathbf{S}_{\un m}
+ \mathbf{S}^{c\top}_{\un m}\mathbf{S}^{c}_{\un m} = \mathbf{I}_{L_{\un n}}.
\end{equation}
Hence, pre-multiplying the selection matrix $\mathbf{S}_{\un m}$ and its complementary $\mathbf{S}_{\un m}^c$ from 
\eqref{eq:selection_matrices} to the $\vect(\boldsymbol{\Gamma})$ and $\vect(\widehat{\boldsymbol{\Gamma}})$ yields
\begin{equation}\label{eq:Cm_Cmc}
\begin{pmatrix}
\un C_{\un m}\\
\un C_{\un m}^{c}
\end{pmatrix}:=
\begin{pmatrix}
\mathbf S_{\un m}\\
\mathbf S_{\un m}^{c}
\end{pmatrix}
\vect(\boldsymbol{\Gamma}),\quad
\begin{pmatrix}
\un{\widehat{C}}_{\un m}\\
\un{\widehat{C}}_{\un m}^{c}
\end{pmatrix}:=
\begin{pmatrix}
\mathbf S_{\un m}\\
\mathbf S_{\un m}^{c}
\end{pmatrix}
\vect(\widehat{\boldsymbol{\Gamma}}),
\end{equation}
where $\un C_{\un m}\in\mathbb R^{M_{\un m}}$ contains the population covariances with lags in $W_{\un m}$ and $\un C_{\un m}^{c}\in\mathbb R^{(L_{\un n}-M_{\un m})}$ contains those with lags in $W_{\un m}^{c}$. By the last identity in \eqref{eq:selection_properties}, we can finally decompose the $L_{\un n}$-dimensional vector of population covariances into the vector of selected and non-selected covariances leading to
\begin{equation}
\label{eq:gamma_decomposition}
  \vect(\boldsymbol{\Gamma}) = \mathbf{S}_{\underline{m}}^\top \mathbf{S}_{\underline{m}}\vect(\boldsymbol{\Gamma}) + \mathbf{S}_{\underline{m}}^{c\top} \mathbf{S}_{\underline{m}}^c\vect(\boldsymbol{\Gamma}).
\end{equation}
Now, using \eqref{eq:linear_bias}, we examine the expectation of the vectorized vector of partitioned sample covariances $\vect(\widehat{\boldsymbol{\Gamma}})$. By plugging in the decomposition of $\vect(\boldsymbol{\Gamma})$ \eqref{eq:gamma_decomposition} and using the notation in \eqref{eq:Cm_Cmc}, we get
\begin{align}
\label{eq:linear_bias_m}
\E[\vect(\widehat{\boldsymbol{\Gamma}}_{\underline{m}})] = \mathbf{S}_{\underline{m}}\mathbf{W}_{\underline{n}}^* \, \vect(\boldsymbol{\Gamma}) &= \mathbf{S}_{\underline{m}}\mathbf{W}_{\underline{n}}^*\mathbf{S}_{\underline{m}}^\top \mathbf{S}_{\underline{m}}\vect(\boldsymbol{\Gamma}) + \mathbf{S}_{\underline{m}}\mathbf{W}_{\underline{n}}^*\mathbf{S}_{\un m}^{c\top}\mathbf{S}_{\un m}^{c}\vect(\boldsymbol{\Gamma})\notag\\
&= \mathbf{W}^{*}_{\un n,\un m}\,\un C_{\un m}
      + \mathbf{W}^{*,c}_{\un n,\un m}\,\un C^{c}_{\un m},  
\end{align}
where the two factors in \eqref{eq:linear_bias_m} are the \emph{reduced weight matrix}, namely a central submatrix of the weight matrix in \eqref{eq:W_n_star}, and the rectangular \emph{leakage matrix}
\begin{align}\label{eq:Wnm_star}
\mathbf{W}_{\underline{n},{\underline{m}}}^{*}
:= \mathbf{S}_{\underline{m}}\,\mathbf{W}_{\underline{n}}^*\,\mathbf{S}_{\underline{m}}^{\top} \in\R^{M_{\un m}\times M_{\un m}},
\quad 
\mathbf{W}^{*,c}_{\un n,\un m}
 := \mathbf{S}_{\un m}\mathbf{W}_{\un n}^{*}\mathbf{S}^{c\top}_{\un m}
 \in \R^{M_{\un m}\times (L_{\un n}-M_{\un m})}
\end{align}
The reduced weight matrix $\mathbf{W}_{\underline{n},{\underline{m}}}^{*}$ inherits the Kronecker structure of $\mathbf{W}_{\underline{n}}^{*}$, such that, for $i\in \{1,2\}$, we have
\[
\mathbf{A}_{i,m_i}^{*} := \mathbf{P}_{i,m_i}\mathbf{A}_i^{*}\mathbf{P}_{i,m_i}^{\top},\quad
\mathbf{B}_{i,m_i}^{*} := \mathbf{P}_{i,m_i}\mathbf{B}_i^{*}\mathbf{P}_{i,m_i}^{\top},\quad
\mathbf{D}_{i,m_i}^{*} := \mathbf{P}_{i,m_i}\mathbf{D}_i^{*}\mathbf{P}_{i,m_i}^{\top},
\]
which are the central $(M_{\un m}\times M_{\un m})$ blocks of the normalized counting matrices in \eqref{eq:normalized_counting_matrices} and the reduced weight matrix is given by
\begin{equation}\label{eq:reducedWeightMatrix} 
\mathbf{W}_{\underline{n},\un{m}}^{*}
=\mathbf{I}_{M_{\un m}}-
\mathbf{A}_{2,m_2}^{*}\otimes \mathbf{A}_{1,m_1}^{*}-
\mathbf{B}_{2,m_2}^{*}\otimes \mathbf{B}_{1,m_1}^{*}+
\mathbf{D}_{2,m_2}^{*}\otimes \mathbf{D}_{1,m_1}^{*},
\end{equation}
which becomes invertible for $m_i\leq n_i-2$ small enough.
Finally, for estimating $\underline{C}_{\underline{m}}$, we define the \emph{bias-corrected vector of (central) sample covariance estimators} $\widehat{\underline{C}}_{\underline{m}}^{*}$ by solving the reduced system in the first term of the last right-hand side of \eqref{eq:linear_bias_m}. That is, we define 
\begin{equation}\label{eq:bias_corrected_estimator}
\widehat{\underline{C}}_{\underline{m}}^{*}
=
\bigl(\mathbf{S}_{\underline{m}}\mathbf{W}_{\underline{n}}^{*}\mathbf{S}_{\underline{m}}^{\top}\bigr)^{-1}
\un{\widehat{C}}_{\un m} = \bigl(\mathbf{W}_{\underline{n},\underline{m}}^{*}\bigr)^{-1}\un{\widehat{C}}_{\un m}.
\end{equation}
Since \eqref{eq:bias_corrected_estimator} inverts only the first of the two terms in \eqref{eq:linear_bias_m}, the estimator $\widehat{\underline{C}}_{\underline{m}}^{*}$ is not exactly unbiased unless the second term vanishes. We formalize this in the following theorem.

\begin{theorem}[Finite sample bias of $\widehat{\underline{C}}_{\underline{m}}^{*}$]\label{prop:bias_of_bce} Let $\un m =(m_1,m_2)$ with $m_i\leq n_i-2, \; i\in \{1,2\}$, and suppose $\mathbf{W}_{\un n, \un m}$ is invertible. Then, the bias vector of the bias-corrected  estimator
$\widehat{\underline{C}}_{\underline{m}}^{*}$ defined in \eqref{eq:bias_corrected_estimator} becomes
\begin{align}
\underline{b}^{*}_{\un m}
&:=\E\bigl(\widehat{\underline{C}}^{*}_{\un m}\bigr)-\underline{C}_{\un m}
=\bigl(\mathbf{W}^{*}_{\un n,\un m}\bigr)^{-1}\mathbf{W}^{*,c}_{\un n,\un m}\,\underline{C}^{c}_{\un m}.
\label{eq:bias_corrected_m}
\end{align}
Moreover, if $M_{\un m}/N < 1$, then 
\begin{equation}\label{eq:bc_bias_bound}
\bigl\lVert\underline{b}^{*}_{\un m}\bigr\rVert_\infty \leq \frac{1}{N-M_{\un m }}\bigl\lVert\underline{C}^{c}_{\un m}\bigr\rVert_1. 
\end{equation}
\end{theorem}

A few remarks on the result of Theorem~\ref{prop:bias_of_bce} are in order.

\begin{remark}[On the bias correction]\
\begin{enumerate}[label=(\roman*)]
\item The bias correction thus removes the contribution to the bias of the selected covariances exactly, whatever the covariance function is, while the leakage bias is pre-multiplied by the inverse of the reduced weight matrix $(\mathbf{W}^{*}_{\un n,\un m})^{-1}$. 
\item The bias-corrected estimator is exactly unbiased if and only if $\mathbf{W}^{*,c}_{\un n,\un m}\un C^{c}_{\un m}=\un 0$. This holds whenever $\un C^{c}_{\un m}=\un 0$,
that is, whenever the field $\mathbf{X}$ is $\un m$-dependent, so that $C(\un h)=0$ for every
$\un h\notin W_{\un m}$.
\item Both factors in \eqref{eq:linear_bias_m} are governed by the deviation of
$\mathbf{W}^{*}_{\un n}$ from the identity. Indeed, by the first two identities in
\eqref{eq:selection_properties}, the identity part of $\mathbf{W}^{*}_{\un n}$ contributes exactly
$\mathbf{I}_{M_{\un m}}$ to $\mathbf{W}^{*}_{\un n,\un m}$ and nothing at all to
$\mathbf{W}^{*,c}_{\un n,\un m}$, so that
\begin{equation}\label{eq:both_blocks}
  \mathbf{W}^{*}_{\un n,\un m}-\mathbf{I}_{M_{\un m}}
   = \mathbf{S}_{\un m}\bigl(\mathbf{W}^{*}_{\un n}-\mathbf{I}_{L_{\un n}}\bigr)
     \mathbf{S}_{\un m}^{\top},
  \qquad
  \mathbf{W}^{*,c}_{\un n,\un m}
   = \mathbf{S}_{\un m}\bigl(\mathbf{W}^{*}_{\un n}-\mathbf{I}_{L_{\un n}}\bigr)
     \mathbf{S}^{c\top}_{\un m} .
\end{equation}
Both are therefore built from the counting matrices generated by the estimation of the mean. In
particular, if $\mu$ were known and used in place of $\overline X$ (leading to $\widehat{C}_{\mu}(\hu)$), we would have
$\mathbf{W}^{*}_{\un n}=\mathbf{I}_{L_{\un n}}$, both matrices in \eqref{eq:both_blocks} would vanish,
and $\widehat{\un C}^{*}_{\un m}$ would coincide with $\widehat{\un C}_{\un m}$ and be unbiased.
\end{enumerate}
\end{remark}
 
\begin{figure}[t]
\begin{center}
\renewcommand{\arraystretch}{1.1}
\resizebox{\linewidth}{!}{%
\begin{NiceTabular}{c|*{7}{c}}[cell-space-limits=1pt]
$h_1 \backslash h_2$
  & -3 & -2 & -1 & 0 & 1 & 2 & 3 \\ \hline
-3 & $\widehat{C}(-3,-3)$ & $\widehat{C}(-3,-2)$ & $\widehat{C}(-3,-1)$ & $\widehat{C}(-3,0)$ & $\widehat{C}(-3,1)$ & $\widehat{C}(-3,2)$ & $\widehat{C}(-3,3)$ \\
-2 & $\widehat{C}(-2,-3)$ & \cellcolor{blue!15} $\widehat{C}(-2,-2)$ & \cellcolor{blue!15}  $\widehat{C}(-2,-1)$ & \cellcolor{blue!15} $\widehat{C}(-2,0)$ &\cellcolor{blue!15}  $\widehat{C}(-2,1)$ & \cellcolor{blue!15} $\widehat{C}(-2,2)$ & $\widehat{C}(-2,3)$ \\
-1 & $\widehat{C}(-1,-3)$ &\cellcolor{blue!15}  $\widehat{C}(-1,-2)$ & \cellcolor{red!15}$\widehat{C}(-1,-1)$ & \cellcolor{red!15}$\widehat{C}(-1,0)$ & \cellcolor{red!15}$\widehat{C}(-1,1)$ & \cellcolor{blue!15} $\widehat{C}(-1,2)$ & $\widehat{C}(-1,3)$ \\
 0 & $\widehat{C}(0,-3)$  & \cellcolor{blue!15} $\widehat{C}(0,-2)$  & \cellcolor{red!15}$\widehat{C}(0,-1)$  & 
 \cellcolor{green!15}{$\widehat{C}(0,0)$}  & \cellcolor{red!15}$\widehat{C}(0,1)$  & \cellcolor{blue!15} $\widehat{C}(0,2)$  & $\widehat{C}(0,3)$ \\
 1 & $\widehat{C}(1,-3)$  & \cellcolor{blue!15} $\widehat{C}(1,-2)$  & \cellcolor{red!15}$\widehat{C}(1,-1)$  & \cellcolor{red!15}$\widehat{C}(1,0)$  & \cellcolor{red!15}$\widehat{C}(1,1)$  & \cellcolor{blue!15}$\widehat{C}(1,2)$  & $\widehat{C}(1,3)$ \\
 2 & $\widehat{C}(2,-3)$  & \cellcolor{blue!15} $\widehat{C}(2,-2)$  & \cellcolor{blue!15} $\widehat{C}(2,-1)$  & \cellcolor{blue!15} $\widehat{C}(2,0)$  & \cellcolor{blue!15} $\widehat{C}(2,1)$  & \cellcolor{blue!15} $\widehat{C}(2,2)$  & $\widehat{C}(2,3)$ \\
 3 & $\widehat{C}(3,-3)$  & $\widehat{C}(3,-2)$  & $\widehat{C}(3,-1)$  & $\widehat{C}(3,0)$  & $\widehat{C}(3,1)$  & $\widehat{C}(3,2)$  & $\widehat{C}(3,3)$
\end{NiceTabular}
}
\vspace{2mm}
\noindent
\begin{tabular}{@{\hspace{2.2cm}}c@{\hspace{1cm}}c@{\hspace{1cm}}c@{}}
\cellcolor{green!15}\phantom{xx} $m = 0$ &
\cellcolor{red!15}\phantom{xx} $m = 1$ & \cellcolor{blue!15}\phantom{xx} $m = 2$
\end{tabular}
\end{center}
\caption{For $n_1=n_2=4$, the $(L_1\times L_2)=(7\times 7)$ sample covariance matrix $\widehat{\mathbf{\Gamma}}$ is shown, where extracted submatrices $\widehat{\mathbf{\Gamma}}_{\underline{m}}$ are indicated for $m_1=m_2=0$ (green) $m_1=m_2=1$ (green and red) and $m_1=m_2=2$ (green, red and blue), respectively.}
\label{fig:selection_example}
\end{figure}

To illustrate the bias correction approach and the quantities introduced above, we conclude this section with a continuation of the toy example setup considered already in Examples \ref{example_Gaussian} and \ref{example_Gaussian_bias}.

\begin{ex}[Bias correction for spatial lattice process]\label{example_Gaussian_bias_correction}
As in the setup of Examples \ref{example_Gaussian} and \ref{example_Gaussian_bias}, suppose again $n_1=n_2=4$. That is, we observe the $(4\times 4)$ data matrix $\mathbf{X}_{\underline{n}}$ such that $\boldsymbol{\Gamma}$ defined becomes a $(7\times 7)$ matrix of covariances. For this setup, Figure \ref{fig:selection_example} shows the corresponding $(7\times 7)$ sample covariance matrix $\widehat{\mathbf{\Gamma}}$. Additionally, Figure \ref{fig:selection_example} indicates the submatrices $\widehat{\mathbf{\Gamma}}_{\underline{m}}$
for the choices $m_1=m_2=0$ (green), $m_1=m_2=1$ (green and red) and $m_1=m_2=2$ (green, red and blue) that are obtained by applying the selection matrices shown in Figure \ref{fig:P_example}. In these cases, the green-and-red-shaded entries ($m_1=m_2=1$) determine a $(3 \times 3)$ submatrix centered around the zero lag $(0,0)$, capturing all lag combinations $(h_1, h_2)$ where $|h_1| \leq m_1=1$ and $|h_2| \leq m_2=1$, while the green-and-red-and-blue-shaded region ($m_1=m_2=2$) represents the corresponding $(5 \times 5)$ submatrix, which includes the lag combinations where $|h_1| \leq m_1=2$ and $|h_2| \leq m_2=2$. Finally, Figure \ref{fig:W_n_star_comparison} (right panel) shows a heatmap of the $(25\times 25)$ inverted weight matrix $\mathbf{W}_{\underline{n},\underline{m}}^{*}$ for the case $m_1=m_2=2$.
\end{ex}

\begin{figure}[t]
\begin{center}
\renewcommand{\arraystretch}{1.0}
\footnotesize
\begin{minipage}[t]{0.31\textwidth}
\centering
\textbf{$\mathbf{P}_{1,m}=\mathbf{P}_{2,m}$ for $m = 0$}
\vspace{2mm}
\resizebox{\linewidth}{!}{%
\begin{NiceTabular}{c|*{7}{c}}[cell-space-limits=0.5pt]
$p \backslash h$ & $-3$ & $-2$ & $-1$ & $0$ & $1$ & $2$ & $3$ \\ \hline
1 & 0 & 0 & 0 & \text{1} & 0 & 0 & 0 \\
\end{NiceTabular}
}
\end{minipage}
\hfill
\begin{minipage}[t]{0.31\textwidth}
\centering
\textbf{$\mathbf{P}_{1,m}=\mathbf{P}_{2,m}$ for $m = 1$}
\vspace{2mm}
\resizebox{\linewidth}{!}{%
\begin{NiceTabular}{c|*{7}{c}}[cell-space-limits=0.5pt]
$p \backslash h$ & $-3$ & $-2$ & $-1$ & $0$ & $1$ & $2$ & $3$ \\ \hline
1 & 0 & 0 & \text{1} & 0 & 0 & 0 & 0 \\
2 & 0 & 0 & 0 & \text{1} & 0 & 0 & 0 \\
3 & 0 & 0 & 0 & 0 & \text{1} & 0 & 0 \\
\end{NiceTabular}
}
\end{minipage}
\hfill
\begin{minipage}[t]{0.31\textwidth}
\centering
\textbf{$\mathbf{P}_{1,m}=\mathbf{P}_{2,m}$ for $m = 2$}
\vspace{2mm}
\resizebox{\linewidth}{!}{%
\begin{NiceTabular}{c|*{7}{c}}[cell-space-limits=0.5pt]
$p \backslash h$ & $-3$ & $-2$ & $-1$ & $0$ & $1$ & $2$ & $3$ \\ \hline
1 & 0 & \text{1} & 0 & 0 & 0 & 0 & 0 \\
2 & 0 & 0 & \text{1} & 0 & 0 & 0 & 0 \\
3 & 0 & 0 & 0 & \text{1} & 0 & 0 & 0 \\
4 & 0 & 0 & 0 & 0 & \text{1} & 0 & 0 \\
5 & 0 & 0 & 0 & 0 & 0 & \text{1} & 0 \\
\end{NiceTabular}
}
\end{minipage}
\end{center}
\caption{The selection matrices $\mathbf{P}_{1,m}=\mathbf{P}_{2,m}$ for a grid of size $n_1=n_2=4$ ($L_1=L_2=7$). Left: $m_1=m_2=0$, selecting only lag $h_i=0$. Middle: $m_1=m_2=1$, selecting lags $|h_i|\leq1$. Right: $m_1=m_2=2$, selecting lags $|h_i|\leq2$.}
\label{fig:P_example}
\end{figure}

\begin{figure}[ht]
    \centering
    \begin{subfigure}[t]{0.48\textwidth}
        \centering
        \includegraphics[width=\textwidth]{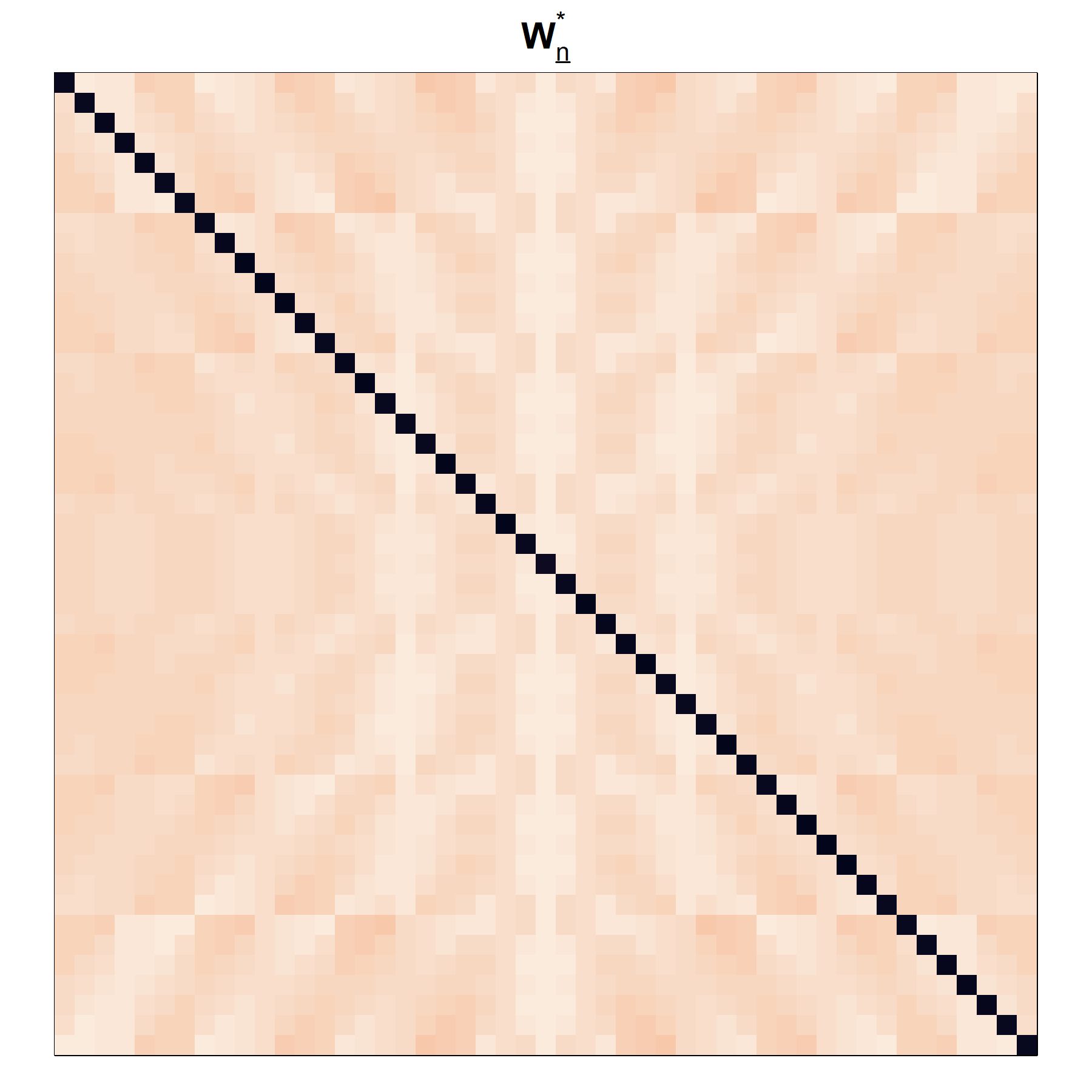}
    \end{subfigure}
    \hfill
    \begin{subfigure}[t]{0.48\textwidth}
        \centering
        \includegraphics[width=\textwidth]{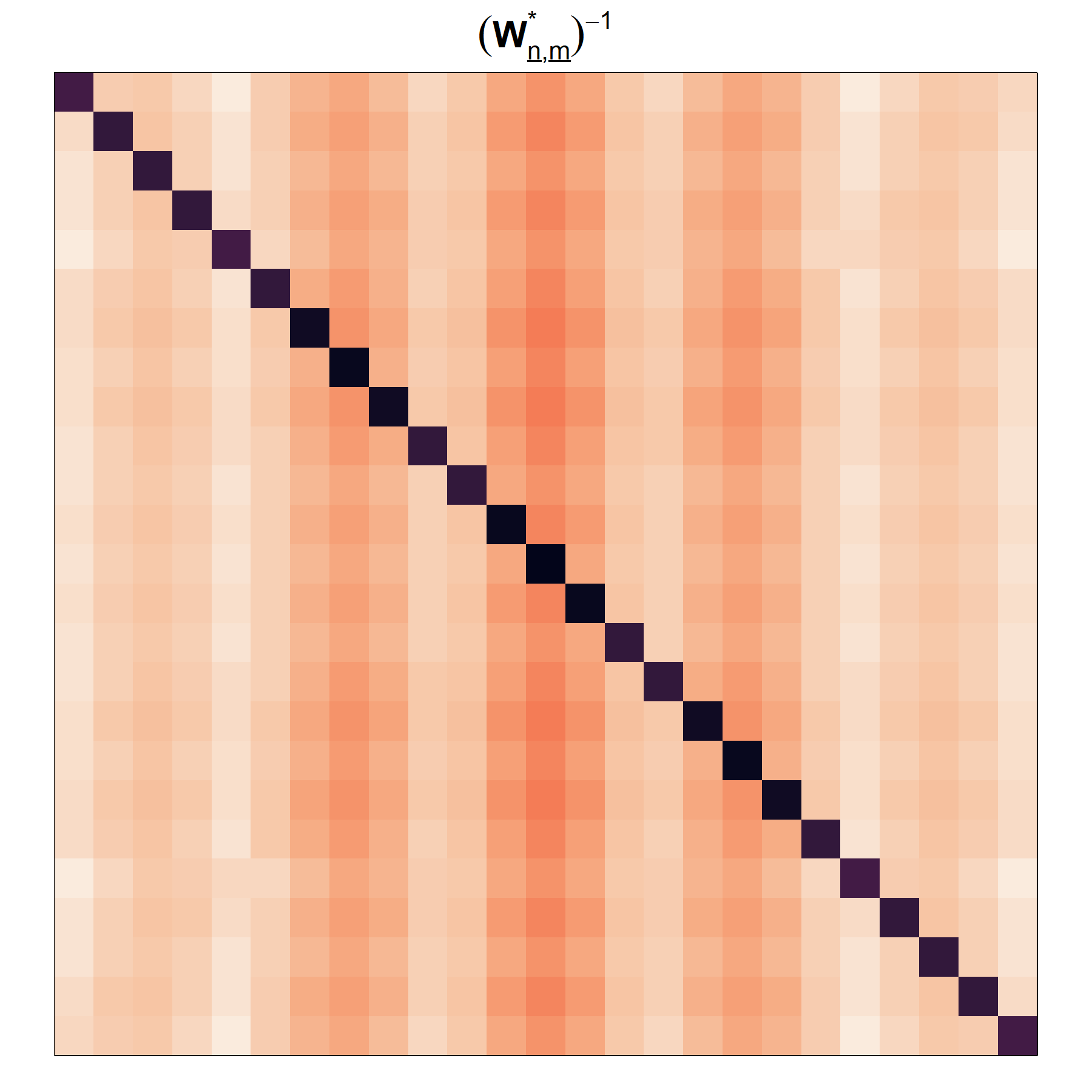}
    \end{subfigure}
\caption{Heatmaps of the (singular) $(49\times 49)$ weight matrix $\mathbf{W}_{\underline{n}}^*$ (left) and the inverse of its central $(25\times 25)$ submatrix $\mathbf{W}_{\underline{n},m}^{*}$, i.e., $(\mathbf{W}_{\underline{n},m}^{*})^{-1}$, (right) for a grid of size $n_1=n_2=4$ and cutoff parameters $m_1=m_2=2$.}
\label{fig:W_n_star_comparison}
\end{figure}

\subsection{On the invertibility of the weight matrix}\label{subsec:invertibility_W}
In the following lemma, we give sufficient conditions on the cutoff parameters $\un m=(m_1,m_2)$ that guarantee invertibility of the resulting weight matrix $\mathbf{W}_{\underline{n},m}^{*}$.

\begin{lemma}[Invertibility and bound of the inverse weight matrix $\mathbf{W}_{\un n, \un m}$]\label{lemma:invertibility_W_nm} Let $\un n = (n_1, n_2)$ be the sizes of the observed lattice $\mathbf{X}_{\un n}$ and $\un m= (m_1, m_2)$ are the cutoff parameters leading to the weight matrix $\mathbf{W}_{\un n, \un m}^*$ from \eqref{eq:reducedWeightMatrix}, with $L_{\un n} = \prod_{i=1}^2 (2n_i - 1)$ and $M_{\un m} = \prod_{i=1}^2 (2m_i + 1)$. Then, we have
\begin{equation}\label{eq:bounds_lemma_W_I}
    \lVert \mathbf{W}_{\un n, \un m}^*-\mathbf{I}_{M_{\un m}}\rVert_\infty \leq \frac{M_{\un m}}{N} \quad \text{and}\quad \lVert \mathbf{W}_{\un n, \un m}^*-\mathbf{I}_{M_{\un m}}\rVert_1 \leq \frac{M_{\un m}}{N}.
\end{equation}
Consequently, whenever $M_{\un m}/N<1$,  the matrix $\mathbf{W}_{\un n, \un m}^*$ is nonsingular with
\begin{equation}\label{eq:invbound}
\bigl\|(\mathbf{W}_{\underline{n},\underline{m}}^{*})^{-1}\bigr\|_{\infty}
\leq\frac{N}{N-M_{\un m}}\quad \text{and}
\quad
\bigl\|\bigl(\mathbf{W}_{\underline{n},\underline{m}}^{*}\bigr)^{-1}-\mathbf{I}_{M_{\underline{m}}}\bigr\|_{\infty}
\leq\frac{M_{\un m}}{N-M_{\un m}}.
\end{equation}
\end{lemma}

A few remarks on the condition $M_{\un m}/ N <1$ of Lemma~\ref{lemma:invertibility_W_nm} are in order.

\begin{remark}[On the invertibility condition]\label{remark:invertibility_W_nm} \
\begin{enumerate}[label=(\roman*)]
\item The condition is stricter than $m_i\leq n_i-2$ for $i \in \{1,2\}$. In fact, since it is based on a product of the two cutoff parameters, the two axes trade against each other. For axis 1, we have that that
    \[
    m_1 < \frac{1}{2}\left(\frac{n_1n_2}{2m_2+1}-1\right),
    \] 
so that a larger cutoff on one axis must be compensated by a smaller one on the other axis. For a
square lattice with $n_1=n_2=n$ and $m_1=m_2=m$ it reduces to $2m+1<n$, that is
$m\leq\lfloor (n-1)/2\rfloor$, roughly half of the admissible range $m\leq n-2$. For example for a $20\times 20$ grid, we would only accept cutoff values $m \leq 9$. 

\item The condition is sufficient but not necessary. We have evaluated numerically lattices with dimensions $3\leq n_i\leq 40$ and their respective cutoffs with the more relaxed condition $m_i \leq n_i -2$ and we see that the reduced weight matrix turned out to be nonsingular in all these cases. The spectral radius of $\mathbf{I}_{M_{\un m}}-\mathbf{W}^{*}_{\un n,\un m}$ increases with the lattice size but stays below 1, so that the Neumann series still converges and remains invertible throughout the admissible range, by Lemma 2.3.3 from \citet{golub2013matrix}. Hence, we conjecture that $\mathbf{W}^{*}_{\un n,\un m}$ is nonsingular, whenever $m_i\leq n_i-2$, $i \in \{1,2\}$ holds.
\end{enumerate}
\end{remark}

\subsection{Equivalence of the bias correction for $\widetilde{C}(\un h)$}\label{subsec:equivalence_bias_correction}
The same construction discussed above applies also to the joint estimation using the constant-prefactor estimator 
$\widetilde{C}(\un h)$ from \eqref{eq:sampleAutocov_const}. We apply the joint characterization from
Corollary~\ref{cor:joint_bias_const} with the same partitioning matrices from \eqref{eq:partition_matrix} to get
\begin{equation}\label{eq:const_Gamma}
\widetilde{\boldsymbol{\Gamma}}_{\underline m}
=
\mathbf{P}_{1,m_1}
\widetilde{\boldsymbol{\Gamma}}
\mathbf{P}_{2,m_2}^{\top}
\end{equation}
which denotes the central submatrix of $\widetilde{\boldsymbol{\Gamma}}$, whose
entries are obtained from $\widetilde C(\un h)$ from \eqref{eq:sampleAutocov_const} instead of $\widehat{C}(\un h)$ from \eqref{eq:sampleAutocov}. Recall from \eqref{eq:W_n_tilde} that the full
constant prefactor weight matrix satisfies
$\widetilde{\mathbf W}_{\underline n}^{*}=(\boldsymbol{\Delta}_2\otimes\boldsymbol{\Delta}_1)
\mathbf W_{\underline n}^{*}$.
We now apply the selection matrix $\mathbf S_{\underline m}$ to obtain the
corresponding reduced weight matrix. For $i\in \{1,2\}$, we define the central Bartlett window
matrices by
\[
\boldsymbol{\Delta}_{i,m_i}
:=
\mathbf{P}_{i,m_i}
\boldsymbol{\Delta}_i
\mathbf{P}_{i,m_i}^{\top}.
\]
Since $\mathbf{P}_{i,m_i}$ from \eqref{eq:partition_matrix} is a coordinate selection matrix,
$\mathbf{P}_{i,m_i}^{\top}\mathbf{P}_{i,m_i}$ is the diagonal projection onto the selected coordinates, but also since  $\boldsymbol{\Delta}_i$ is already diagonal we have that $\mathbf{P}_{i,m_i}\boldsymbol{\Delta}_i=\boldsymbol{\Delta}_{i,m_i}\mathbf{P}_{i,m_i}$. This implies that 
\begin{equation}\label{eq:bartlett_commuting}
\mathbf{S}_{\underline m}
(\boldsymbol{\Delta}_2\otimes\boldsymbol{\Delta}_1)
=
(\boldsymbol{\Delta}_{2,m_2}\otimes\boldsymbol{\Delta}_{1,m_1})
\mathbf{S}_{\underline m}.
\end{equation}
Therefore, the weight matrix for the constant prefactor bias correction is obtained by 
\begin{align}
\widetilde{\mathbf{W}}_{\underline n,\underline m}^{*}
&:=
\mathbf{S}_{\underline m}
\widetilde{\mathbf{W}}_{\underline n}^{*}
\mathbf{S}_{\underline m}^{\top}=
\mathbf{S}_{\underline m}
(\boldsymbol{\Delta}_2\otimes\boldsymbol{\Delta}_1)
\mathbf{W}_{\underline n}^{*}
\mathbf{S}_{\underline m}^{\top}\notag\\
&=(\boldsymbol{\Delta}_{2,m_2}\otimes\boldsymbol{\Delta}_{1,m_1})
\mathbf{S}_{\underline m}
\mathbf{W}_{\underline n}^{*}
\mathbf{S}_{\underline m}^{\top}=
(\boldsymbol{\Delta}_{2,m_2}\otimes\boldsymbol{\Delta}_{1,m_1})
\mathbf{W}_{\underline n,\underline m}^{*}.
\label{eq:Wnm_tilde}
\end{align}
Whenever the reduced weight matrix $\mathbf{W}_{\underline n,\underline m}^{*}$ is invertible, the matrix $\widetilde{\mathbf{W}}_{\underline n,\underline m}^{*}$ is invertible as well, since it is the
product of
$\mathbf{W}_{\underline n,\underline m}^{*}$ and the invertible diagonal
matrix
$\boldsymbol{\Delta}_{2,m_2}\otimes\boldsymbol{\Delta}_{1,m_1}$.
Analogously to \eqref{eq:bias_corrected_estimator}, we define the
bias-corrected vector of central sample covariance estimators with constant
prefactor by
\begin{equation}
\widetilde{\underline{C}}_{\underline m}^{*}
=
\bigl(
\mathbf{S}_{\underline m}
\widetilde{\mathbf{W}}_{\underline n}^{*}
\mathbf{S}_{\underline m}^{\top}
\bigr)^{-1}
\mathbf{S}_{\underline m}
\vect\bigl(\widetilde{\boldsymbol{\Gamma}}\bigr)
=
\bigl(
\widetilde{\mathbf{W}}_{\underline n,\underline m}^{*}
\bigr)^{-1}
\vect\bigl(\widetilde{\boldsymbol{\Gamma}}_{\underline m}\bigr).
\label{eq:bias_corrected_estimator_const}
\end{equation}
Interestingly, although the two bias corrections start from estimators with different prefactors and different (biased) expectations, the correction removes this discrepancy entirely: it turns out that both bias-corrected estimators coincide.
 
\begin{corollary}[Both bias-corrected estimators coincide]\label{cor:coincide}
For every cutoff $\underline{m}=(m_1,m_2)$ that implies invertibility of  $\mathbf{W}_{\underline n,\underline m}^{*}$,
the bias-corrected estimators $\widehat{\underline{C}}_{\underline m}^{*}$ (standard) from \eqref{eq:bias_corrected_estimator} 
and $\widetilde{\underline{C}}_{\underline m}^{*}$ (constant prefactor) from \eqref{eq:bias_corrected_estimator_const} coincide. That is, we have
\begin{equation}
\widetilde{\underline{C}}_{\underline{m}}^{*}=\widehat{\underline{C}}_{\underline{m}}^{*}.
\end{equation}
\end{corollary}

The choice of the prefactor $1/N_h$ or $1/N$, which matters considerably for the standard and constant prefactor estimators \eqref{eq:sampleAutocov} and\eqref{eq:sampleAutocov_const}, as discussed in Section~\ref{subsec:sample_covariances}, is thus irrelevant after the bias correction has been applied and the Bartlett taper is removed as a by-product. However, the main advantage of the constant prefactor estimator is also lost: the bias-corrected estimator does not necessarily lead to a positive semidefinite covariance. We address this issue in Section~\ref{subsec:plug_in_Sigma}.

\subsection{Mean-squared error analysis}\label{subsec:MSE_analysis}
    
For assessing the overall estimation performance of the bias-corrected covariance estimator $\widehat{\underline{C}}_{\underline{m}}^{*}$ beyond its bias properties, its mean-squared error (MSE) is typically considered.
For this purpose, first, we define the covariance matrix of all sample covariances
\begin{equation}\label{eq:Omega_full_def}
\boldsymbol{\Omega}:= \Cov\bigl(\vect(\widehat{\boldsymbol{\Gamma}})\bigr)\in \mathbb{R}^{L_{\un n} \times L_{\un n}}
\end{equation}
and using $\un{\widehat{C}}_{\un m}= \mathbf{S}_{\un m}\vect(\widehat{\boldsymbol{\Gamma}})$, we have 
\begin{align}\label{eq:Omega_m_def}
\mathbf{\Omega}^*_{\underline{m}} := \Cov\left(\widehat{\underline{C}}_{\underline{m}}^{*}\right) = (\mathbf{W}_{\underline{n}, \underline{m}}^*)^{-1}\mathbf{S}_{\underline{m}}\,\boldsymbol{\Omega}\,\mathbf{S}_{\underline{m}}^\top(\mathbf{W}_{\underline{n}, \underline{m}}^{*\top})^{-1}\in \mathbb{R}^{M_{\un m}\times M_{\un m}},    
\end{align}
where each entry is given by
\[
\bigl[\boldsymbol{\Omega}_{\un m}^*\bigr]_{k,l}
=\Cov\bigl(\widehat{C}^*(\un{\ell}_{W_{\un m}}(k)),\,\widehat{C}^*(\un{\ell}_{W_{\un m}}(l))\bigr),
\qquad k,l = 1,\dots, M_{\un m},
\]
with $\un{\ell}_{W_{\un m}}$ the enumeration \eqref{eq:mapping_vec} of the within-window
lag set $W_{\un m}=I(-m_1:m_1)\times I(-m_2:m_2)$, for which
$a_i=-m_i$, $c_i=2m_i+1$ and $|W_{\un m}|=M_{\un m}$. That is, we have 
\[
\un{\ell}_{W_{\un m}}(k)=\bigl(((k-1)\bmod(2m_1+1))-m_1,\;
\lfloor(k-1)/(2m_1+1)\rfloor-m_2\bigr).
\]
Hence, in the setup considered here, for estimating $\underline{C}_{\underline{m}}$, the MSE of the bias-corrected vector of (central) sample covariance estimators $\widehat{\underline{C}}_{\underline{m}}^{*}$ becomes
\begin{align}\label{eq:mse_vector_bc}
\textup{MSE}(\widehat{\underline{C}}_{\underline{m}}^{*}):=\E\left(\left\|\widehat{\underline{C}}_{\underline{m}}^{*}-\underline{C}_{\underline{m}}\right\|_2^2\right) = \E\left(\left(\widehat{\underline{C}}_{\underline{m}}^{*}-\underline{C}_{\underline{m}}\right)^\top\left(\widehat{\underline{C}}_{
    \underline{m}}^{*}-\underline{C}_{\underline{m}}\right)\right) = \trace\left(\boldsymbol{\Omega}_{\un m}^*\right)+\|\underline{b}^*_{\underline{m}}\|_2^2, 
    \end{align}
where $\|\cdot\|_2$ denotes the Euclidean (vector) norm, $\trace(\cdot)$ the trace of a (square) matrix. That is, in the above, $\underline{b}^*_{\underline{m}}$ is
the $M_{\un m}$-dimensional bias vector of $\widehat{\underline{C}}_{\underline{m}}^{*}$ from Theorem~\ref{prop:bias_of_bce}, and 
\begin{align*}
\trace(\mathbf{\Omega}_{\un m}^*)=\trace\left((\mathbf{W}_{\underline{n}, \underline{m}}^*)^{-1}\mathbf{S}_{\underline{m}}\,\boldsymbol{\Omega}\,\mathbf{S}_{\underline{m}}^\top(\mathbf{W}_{\underline{n}, \underline{m}}^{*\top})^{-1}\right)=\sum_{i=1}^{M_{\un m}} \Var(\widehat{C}_{\underline{m},i}^*).    
\end{align*}

\subsection{Plug-in estimator of the variance covariance matrix and its mean squared error}\label{subsec:plug_in_Sigma}
The MSE in \eqref{eq:mse_vector_bc} holds for the vector of bias-corrected estimators $\widehat{C}^*_{\un m}$ only. In practical application, one is often interested in working with the proper variance covariance matrix of the vectorized spatial lattice
\begin{equation}\label{eq:real_sigma}
    \boldsymbol{\Sigma}:= \Cov\bigl(\vect(\mathbf{X}_{\un n})\bigr)\in \mathbb{R}^{N \times N}.
\end{equation}
We use the column-wise enumeration $G_{\un n}=\{\un{\ell}_{G_{\un n}}(1),\ldots,\un{\ell}_{G_{\un n}}(N)\}$
given by \eqref{eq:mapping_vec} applied to the grid
$G_{\un n}=I(1:n_1)\times I(1:n_2)$, for which $a_i=1$, $c_i=n_i$ and
$|G_{\un n}|=N$, 
\[
\un{\ell}_{G_{\un n}}(k)
=\Biggl(\bigl((k-1)\bmod n_1\bigr)+1,\;
\Biggl\lfloor\frac{k-1}{n_1}\Biggr\rfloor+1\Biggr),
\qquad k=1,\ldots,N,
\]
to define the entries as
\[
[\boldsymbol{\Sigma}]_{k,l}=C\bigl(\un{\ell}_{G_{\un n}}(k)-\un{\ell}_{G_{\un n}}(l)\bigr),
\qquad k,l=1,\ldots,N .
\]
Then, we build a simple plug-in estimator by replacing $C$ by $\widehat{C}^{*}_{\underline{m}}$ and tapering to zero the entries where the lag lies outside the within-window lag
\begin{equation}
\label{eq:Sigma_hat_m_corr}
\bigl[\widehat{\boldsymbol{\Sigma}}^*_{\un m}\bigr]_{k,l}
=
\begin{cases}
\widehat{C}^*_{\un m}\bigl(\un{\ell}_{G_{\un n}}(k)-\un{\ell}_{G_{\un n}}(l)\bigr),
& \bigl(\un{\ell}_{G_{\un n}}(k)-\un{\ell}_{G_{\un n}}(l)\bigr)\in W_{\un m},\\
0, & \text{otherwise}.
\end{cases}
\qquad k,l=1,\ldots,N.
\end{equation}
We define the full matrix containing the effective sample sizes for all the lags $\un h \in H_{\un n}$ using the Bartlett weight matrices from \eqref{eq:bartlett_matrices}
\begin{equation*}
    \mathbf{N}_{\un n} = N (\boldsymbol{\Delta}_2 \otimes \boldsymbol{\Delta}_1) = \textup{diag}(N_h) \in \mathbb R^{L_{\un n}\times L_{\un n}}
\end{equation*}
Then, under Assumption \eqref{ass:stationarity}, the exact mean squared error of \eqref{eq:Sigma_hat_m_corr} is given by the MSE($\widehat{\un{C}}^*_{\un m}(\un h)$) plus the bias produced by the sum of the squared tapered covariances
\begin{align}\label{eq:mse_sigma_m_hat}
\textup{MSE}\left(\widehat{\boldsymbol{\Sigma}}^*_{\underline{m}}\right)
:=\E\left(\left\|\widehat{\boldsymbol{\Sigma}}^*_{\underline{m}}-\boldsymbol{\Sigma}\right\|_F^2\right)
=\trace\left(\mathbf{S}_{\underline{m}}\mathbf{N}_{\underline{n}}\mathbf{S}_{\underline{m}}^{\top}
\left(\boldsymbol{\Omega}_{\underline{m}}^*+\underline{b}_{\underline{m}}^*\underline{b}_{\underline{m}}^{*\top}\right)\right)
+\trace\left(\mathbf{S}_{\underline{m}}^{c}\mathbf{N}_{\underline{n}}\mathbf{S}_{\underline{m}}^{c\top}
\mathbf{S}_{\underline{m}}^{c}\un C\un C^{\top}\mathbf{S}_{\underline{m}}^{c\top}\right),
\end{align}
where $\|\cdot\|_F$ denotes the Frobenius (matrix) norm.

\begin{remark}[On the positive semidefiniteness of  $\widehat{\boldsymbol{\Sigma}}_{\underline{m}}^*$]\ 
\begin{enumerate}[label=(\roman*)]
    \item The tapering does not necessarily produce a positive semidefinite matrix, thus not a valid covariance matrix. However, by construction, it is a real symmetric matrix, which always admits the spectral decomposition $\widehat{\boldsymbol{\Sigma}}_{\underline{m}}^*=\mathbf{Q}
    \boldsymbol{\Lambda}\mathbf{Q}^{\top}$ with orthogonal $\mathbf{Q}$ and
    $\mathbf{\Lambda}=\operatorname{diag}(\lambda_1,\ldots,\lambda_N)$. We apply the spectral projections used, e.g., in \citet{jentsch2015covariance}, and set all negative eigenvalues to zero and define the projected estimator
\begin{equation}
\label{eq:psd_proj}
\widehat{\boldsymbol{\Sigma}}_{\underline{m}}^{*+}
=\mathbf{Q}\mathbf{\Lambda}^{+}\mathbf{Q}^{\top},\quad \text{where}\quad
\mathbf{\Lambda}^{+}=\operatorname{diag}\bigl(\max(\lambda_1,0),\ldots,\max(\lambda_N,0)\bigr),
\end{equation}
which is assured to be positive semidefinite. Similarly, by setting
\begin{equation*}
\mathbf{\widetilde \Lambda}^{+}=\operatorname{diag}\bigl(\max(\lambda_1,\epsilon),\ldots,\max(\lambda_N,\epsilon)\bigr),    \quad \text{where}  \quad   \epsilon>0
\end{equation*}
and using $\mathbf{\widetilde \Lambda}^{+}$ instead of $\mathbf{\Lambda}^{+}$, we get a positive definite estimator.
\item For fixed $\un m$, When $n_1,n_2\rightarrow \infty$, both estimators $\widehat{\boldsymbol{\Sigma}}_{\underline{m}}^{*}$ and $\widehat{\boldsymbol{\Sigma}}_{\underline{m}}^{*+}$ converge to the same limit that is always positive semidefinite.
\end{enumerate}
\end{remark}

\subsection{Asymptotic theory}\label{subsec:asymptotics}

In this section, we provide asymptotic theory for the sample covariance estimator with and without bias correction. For this purpose, following \cite{vogelsang2016exactly} and \cite{li2024correcting} for the time series case, we suppose that the spatial process $\mathbf{X}$ is linear.
\begin{assumption}[Linear spatial process]\label{ass:linear} The spatial lattice process $\mathbf{X}=\{X_{\tu}, \tu \in \mathbb{Z}^2\}$ is linear,
\begin{equation}
    X_{\tu} = \mu + \sum_{\un j\in \mathbb{Z}^2} \psi_{\un j} \epsilon_{\tu - \un j}, \quad \tu \in \mathbb{Z}^2,
\end{equation}
where $\{\epsilon_{\tu}, \tu\in\mathbb{Z}^2\}$ is an i.i.d. white noise process with $\E(\epsilon_{\tu})=0$, $\Var(\epsilon_{\tu})=\sigma_\epsilon^2\in(0,\infty)$ and with finite fourth moments $\E(\epsilon_{\tu}^4)=\eta\sigma_\epsilon^4<\infty$. Further, suppose that the sequence $\{\psi_{\un j}, \un j\in\mathbb{Z}^2\}$ is absolutely summable, $\sum_{\un j \in \mathbb{Z}^2 }|\psi_{\un j}|<\infty$.    
\end{assumption}

However, other types of dependence assumptions may be employed such as mixing-type conditions \citep{Bolthausen1982, Politis1999, Bradley2007intro}, weak dependence concepts \citep{DedeckerDoukhan2007}, or physical/functional dependence concepts \citep{ElMachkouri2013}.
Irrespective of the dependency concept used, it turns out that the bias-corrected sample covariance estimator (for both considered prefactors) is asymptotically equivalent to the uncorrected standard one.

\begin{proposition}[Asymptotics for $\widehat{\underline{C}}_{\un m}$ and $\widehat{\underline{C}}_{\un m}^*$]\label{prop:asymptotics_without_and_with_bias_correction}
Let $\un{m}\in\mathbb{N}^2$ be fixed and let
Assumptions~\ref{ass:stationarity} and \ref{ass:linear} hold. Then, as $n_1,n_2\rightarrow \infty$, we have
\begin{align*}
\sqrt{N}\left(\widehat{\underline{C}}_{\un m}-\underline{C}_{\un m}\right)\overset{d}{\rightarrow}\mathcal{N}\left(\underline{0},\mathbf{V}_{\un m}\right)    \quad\text{and} \quad   \sqrt{N}\left(\widehat{\underline{C}}_{\un m}^*-\underline{C}_{\un m}\right)\overset{d}{\rightarrow}\mathcal{N}\left(\underline{0},\mathbf{V}_{\un m}\right),   
\end{align*}
with
\begin{equation}
    \mathbf{V}_{\un m}:= \lim_{N\to\infty} N\,\boldsymbol{\Omega}^*_{\un m},
    \qquad
    \bigl[\mathbf{V}_{\un m}\bigr]_{k,l}
    = V\bigl(\un{\ell}_{W_{\un m}}(k),\,\un{\ell}_{W_{\un m}}(l)\bigr),
    \quad k,l=1,\ldots,M_{\un m},
\end{equation}
with $\un{\ell}_{W_{\un m}}$ the enumeration \eqref{eq:mapping_vec} of $W_{\un m}$ and where the entries of $\mathbf{V}_{\un m}$ are of the form
\begin{equation}
    V(\hu,\gu)=(\eta-3)C(\hu)C(\gu)
+\sum_{\du\in\Z^{2}}\Bigl[C(\du)C(\du+\gu-\hu)
+C(\du+\gu)C(\du-\hu)\Bigr],\quad \hu,\gu\in W_{\un m}.
\end{equation}
\end{proposition}

Some remarks on 
the limiting results in Proposition \ref{prop:asymptotics_without_and_with_bias_correction} are in order.

\begin{remark}[On Proposition \ref{prop:asymptotics_without_and_with_bias_correction}  and possible extensions]\ 
\begin{itemize}
    \item[(i)] The result in Proposition \ref{prop:asymptotics_without_and_with_bias_correction} extends Theorem 1 in \citet{vogelsang2016exactly} to linear lattice processes in the plane.
    \item[(ii)] The same limiting distribution does also hold for $\widetilde{\underline{C}}_m=\textup{vec}\bigl(\widetilde{\mathbf{\Gamma}}_m\bigr)$, which relies on the constant prefactor covariance estimator. However, according to Corollary \ref{cor:coincide}, its bias corrected version $\widetilde{\underline{C}}_m=\textup{vec}\bigl(\widetilde{\mathbf{\Gamma}}_m\bigr)$ coincides with $\widehat{\underline{C}}_m=\textup{vec}\bigl(\widehat{\mathbf{\Gamma}}_m\bigr)$ and is trivially also covered by Proposition \ref{prop:asymptotics_without_and_with_bias_correction}.
    \item[(iii)] Without relying on the Assumption~\ref{ass:linear} of linearity of the underlying process $\mathbf{X}$, under suitable mixing conditions (that guarantee sufficient summability conditions for the cumulants), asymptotic normality of $\sqrt{N}(\widehat{\underline{C}}_m-\underline{C}_m)$ can still be obtained. However, by dropping Assumption~\ref{ass:linear} and using summability conditions for cumulants, the entries of the limiting variance change to become
\begin{align*}
V_{general}(\widehat{\underline{C}}_m)\big(\un h,\un g\big)  =& \sum_{\un  \delta\in \mathbb{Z}^2} \textup{cum}\left(X_{\hu}-\mu,X_{\un 0}-\mu, X_{\du+\gu}-\mu,X_{\du}-\mu\right)   \\
& +\sum_{\du\in\Z^{2}}\Bigl[C(\du)C(\du+\gu-\hu)
+C(\du+\gu)C(\du-\hu)\Bigr],
\end{align*}
where $\textup{cum}(W,X,Y,Z)$ denotes the joint cumulant of random variables $W,X,Y$ and $Z$.
\end{itemize}    
\end{remark}

\subsection{Selection of the cutoff parameter}\label{subsec:cutoff_selection}
The cutoff $\un m$ plays here the role that commonly is taken by the bandwidth in classical nonparametric estimation methods. In particular, it is not identified by the model, it governs a bias-variance trade-off and it has to be fixed before the estimator can be computed. Hyperparameter selection is a vast topic of interest in nonparametric methods, which have a need of defining the optimal hyperparameter to minimize the bias-variance trade-off. In our method, on the one hand, we have that a small $\un m$ corrects the bias minimally at the cost of setting to $0$ all non selected covariances with lags $\un h\notin W_{\un m}$, on the other hand, a higher $\un m$ results in a larger bias correction, but it inflates the variance of the estimator, resulting in poorer results according to MSE criteria. As usual for hyperparameter selection of this kind, practitioners often rely on data-driven methods to find parameters that perform well given the task at hand. We propose here two methods for the selection of the cutoff parameters $\un m$.

First, we suggest a \emph{plug-in rule} from a fitted correlogram.
\begin{enumerate}
\item  For $i =1,2$, with $\un e_1 =(1,0)$ and $\un e_2 =(0,1)$ and evaluate the correlogram by $\widehat\rho_i(h):=\widehat C(h\,\un{e}_i)/\widehat C(\un{0})$ for $h = 1, \dots, K_i$, where one can choose $K_i=\lfloor\min(n_1,n_2)/2\rfloor$. 
    \item For a correlogram model given by $\rho(h)=\exp(-h/\lambda)$, we fit a weighted least squares model
through the origin on the log scale, with weights given by the lags similar to those in \citet{Cressie1985}, to obtain an estimate of the inverse of the dependence parameter,
\begin{equation}\label{eq:wls_range}
\widehat\lambda_i^{-1}
=-\frac{\sum_{h=1}^{K_i}w_{i,h}h\log\widehat\rho_i(h)}
        {\sum_{h=1}^{K_i}w_{i,h}h^{2}},\quad w_{i,h}=n_i-h.
\end{equation}
\item We scale the fitted dependence parameter by the axis length and we select the cutoff parameter by $\widehat{m}^{\text{rule}}_i := \lfloor \kappa\,\bigl(n_i\widehat\lambda_i\bigr)^{\eta}\rceil$, with $\lfloor\cdot\rceil$ denoting rounding to the nearest integer and with $\kappa=0.445$ and $\eta=0.5$, which are not implied by the model and were calibrated by numerical simulations.
\end{enumerate}

Second, we suggest applying a \emph{block cross-validation} method. 
\begin{enumerate}
\item Fix a grid of cutoffs, where we set $1\leq m\leq\lfloor\min(n_1,n_2)/2\rfloor$, for $m_1=m_2=m$.
\item For each candidate $\un{m}$, compute $\widehat{\boldsymbol{\Sigma}}_{\underline{m}}^{*}$ on the whole lattice. Then, compute its inverse $(\boldsymbol{\Sigma}_{\un m}^{*+})^{-1}$ with the spectral projection method from \eqref{eq:psd_proj}.
\item Determine $B$ square blocks $P_1,\ldots,P_B$ of size $f\times f$, all kept inside the lattice, so that the criterion does not depend on one particular position of the block.
\item For each block, predict the held-out block from the
remaining sites $O_b$ by ordinary kriging, 
\begin{equation}\label{eq:ok_predictor}
\widehat{y}_{P_b}
=\widehat\mu\eones
+\widehat{\boldsymbol\Sigma}_{P_bO_b}(\widehat{\boldsymbol\Sigma}^{*+}_{O_bO_b})^{-1}
 \bigl(y_{O_b}-\widehat\mu\eones\bigr),
\quad
\widehat{\mu}=(\eones^{\top}(\widehat{\boldsymbol{\Sigma}}_{OO}^{*+})^{-1}\eones)^{-1}\eones^{\top}(\widehat{\boldsymbol{\Sigma}}_{OO}^{*+})^{-1}\underline{y}_{O},
\end{equation}
\item Average the root mean squared prediction error over the blocks and select the $\un m$, which minimizes the criterion,
\begin{equation}\label{eq:cv_score}
\mathrm{CV}(\un{m})=\frac{1}{B}\sum_{b=1}^{B}
\Bigl(\frac{1}{|P_b|}\sum_{s\in P_b}
\bigl(\widehat{y}_s-y_s\bigr)^{2}\Bigr)^{1/2},
\qquad
\widehat{\un{m}}^{\text{cv}}=\arg\min_{\un{m}}\mathrm{CV}(\un{m}).
\end{equation}
\end{enumerate}

Unlike the plug-in rule, the block cross-validation does not assume a shape of the correlogram, however it is evidently more expensive computationally. 
\FloatBarrier
\section{Simulation Study}\label{sec:simulations}

In this section, we perform a Monte Carlo simulation study to assess the finite-sample performances of the naive (uncorrected) and the bias-corrected covariance estimators under varying degrees of spatial dependence. We consider spatial Gaussian lattice processes as introduced in Example \ref{example_Gaussian} with covariance function defined in \eqref{eq:trueautocov}. That is, we resume the same simulation setup as described in Example~\ref{example_Gaussian} and generate a true variance covariance matrix $\boldsymbol{\Sigma}_{\lambda,N} \in \mathbb{R}^{N \times N}$ by plugging the entries

\begin{equation}\label{eq:Sigma_true}
\bigl[\boldsymbol{\Sigma}_{\lambda,N}\bigr]_{k,l}
=C_{\lambda,N}\bigl(\un \ell_{G_{\un n}}(k)-\un \ell_{G_{\un n}}(l)\bigr),
\quad k,l=1,\ldots,N.
\end{equation}
with $
\un \ell_{G_{\un n}}(k)=\bigl(((k-1)\bmod n_1)+1,\; \lceil k/n_1\rceil\bigr)$.

\begin{figure}[htb]
\centering
\captionsetup[subfigure]{skip=6pt}
\begin{minipage}[c]{0.06\textwidth}
  \centering
  \rotatebox{90}{$\beta = 0$}
\end{minipage}\hfill
\begin{minipage}[c]{0.92\textwidth}
  \centering
  \begin{subfigure}[t]{0.24\textwidth}
    \centering
    \includegraphics[width=\linewidth,height=0.18\textheight,keepaspectratio]{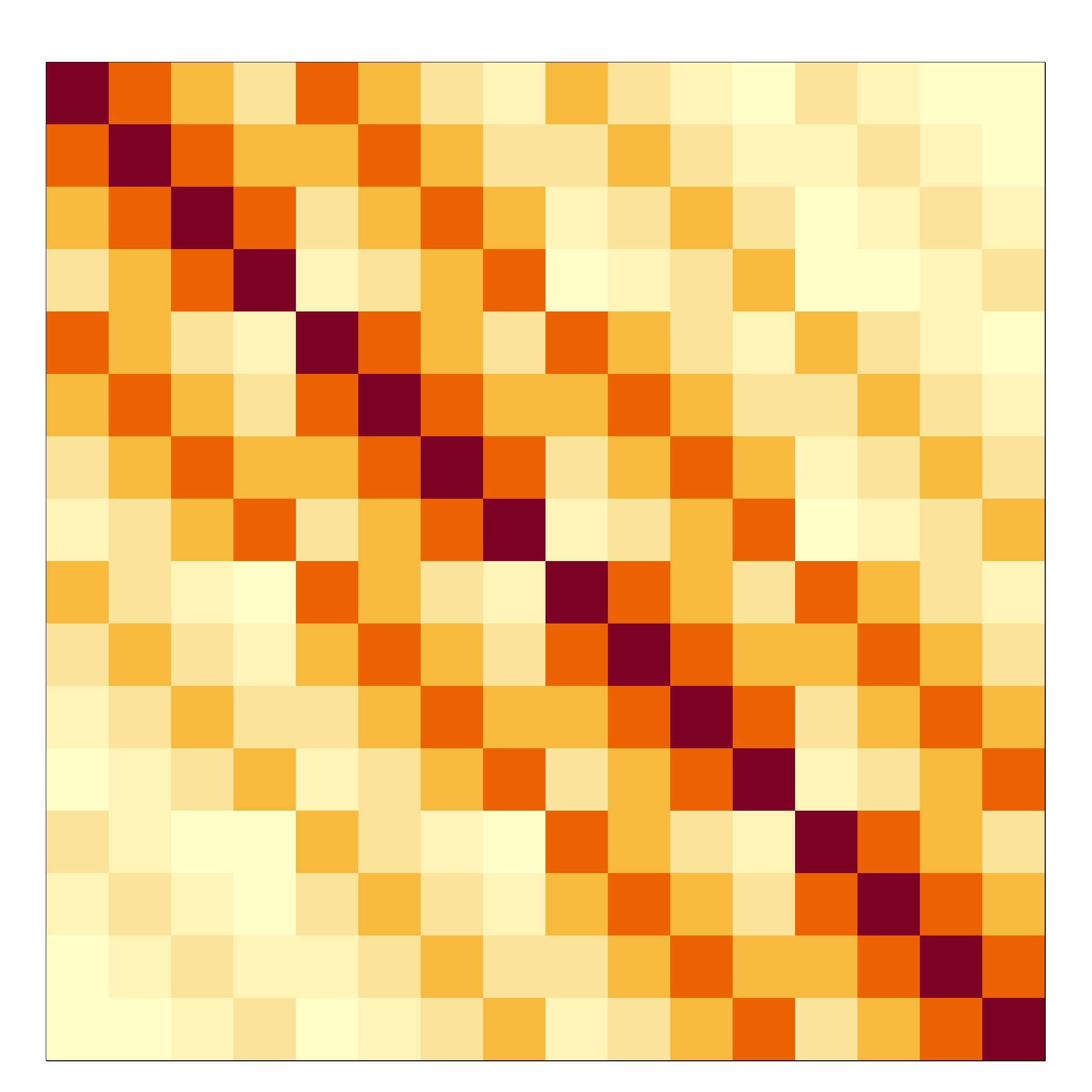}
    \caption{$\lambda=2$}
  \end{subfigure}\hfill
  \begin{subfigure}[t]{0.24\textwidth}
    \centering
    \includegraphics[width=\linewidth,height=0.18\textheight,keepaspectratio]{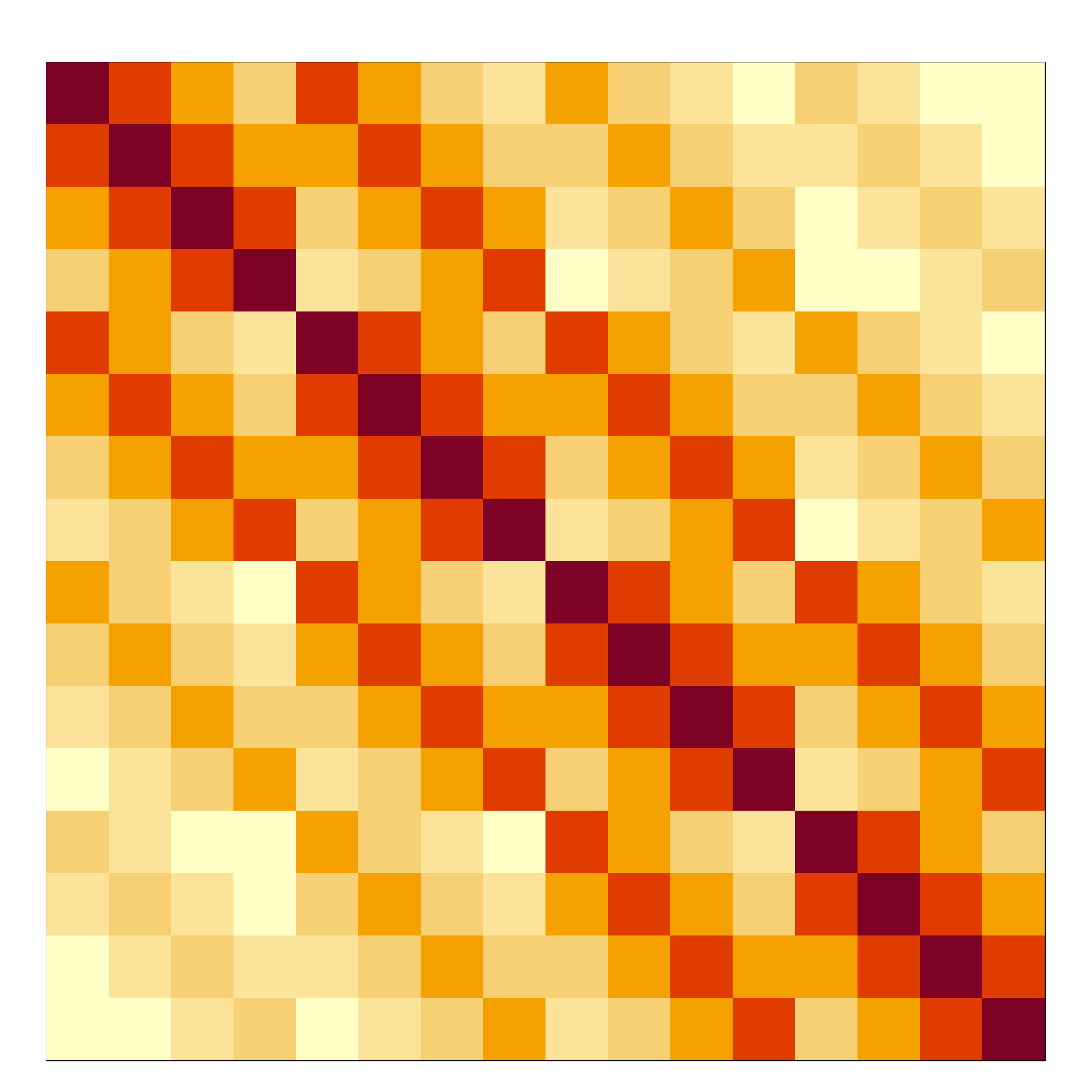}
    \caption{$\lambda=4$}
  \end{subfigure}\hfill
  \begin{subfigure}[t]{0.24\textwidth}
    \centering
    \includegraphics[width=\linewidth,height=0.18\textheight,keepaspectratio]{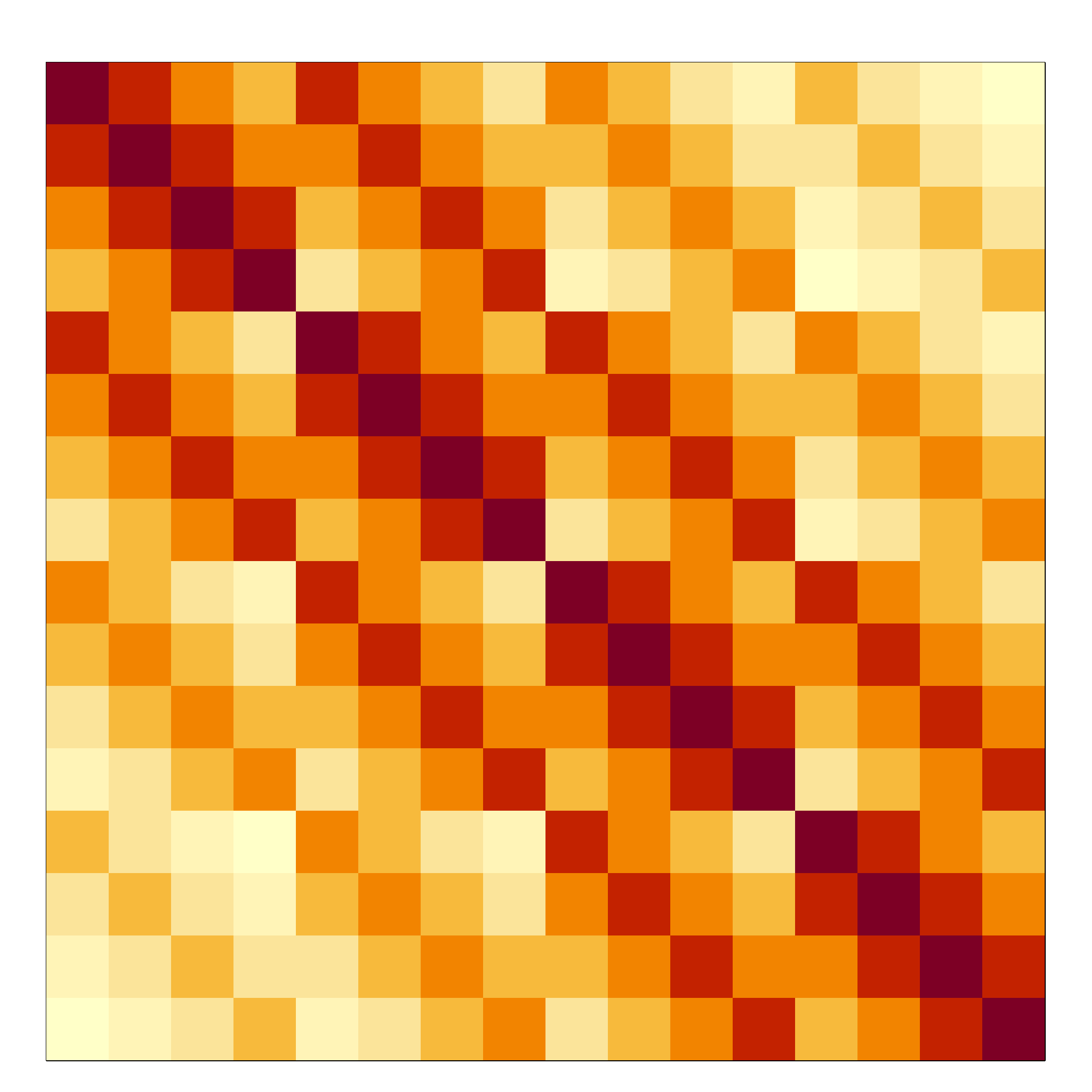}
    \caption{$\lambda=6$}
  \end{subfigure}\hfill
  \begin{subfigure}[t]{0.24\textwidth}
    \centering
    \includegraphics[width=\linewidth,height=0.18\textheight,keepaspectratio]{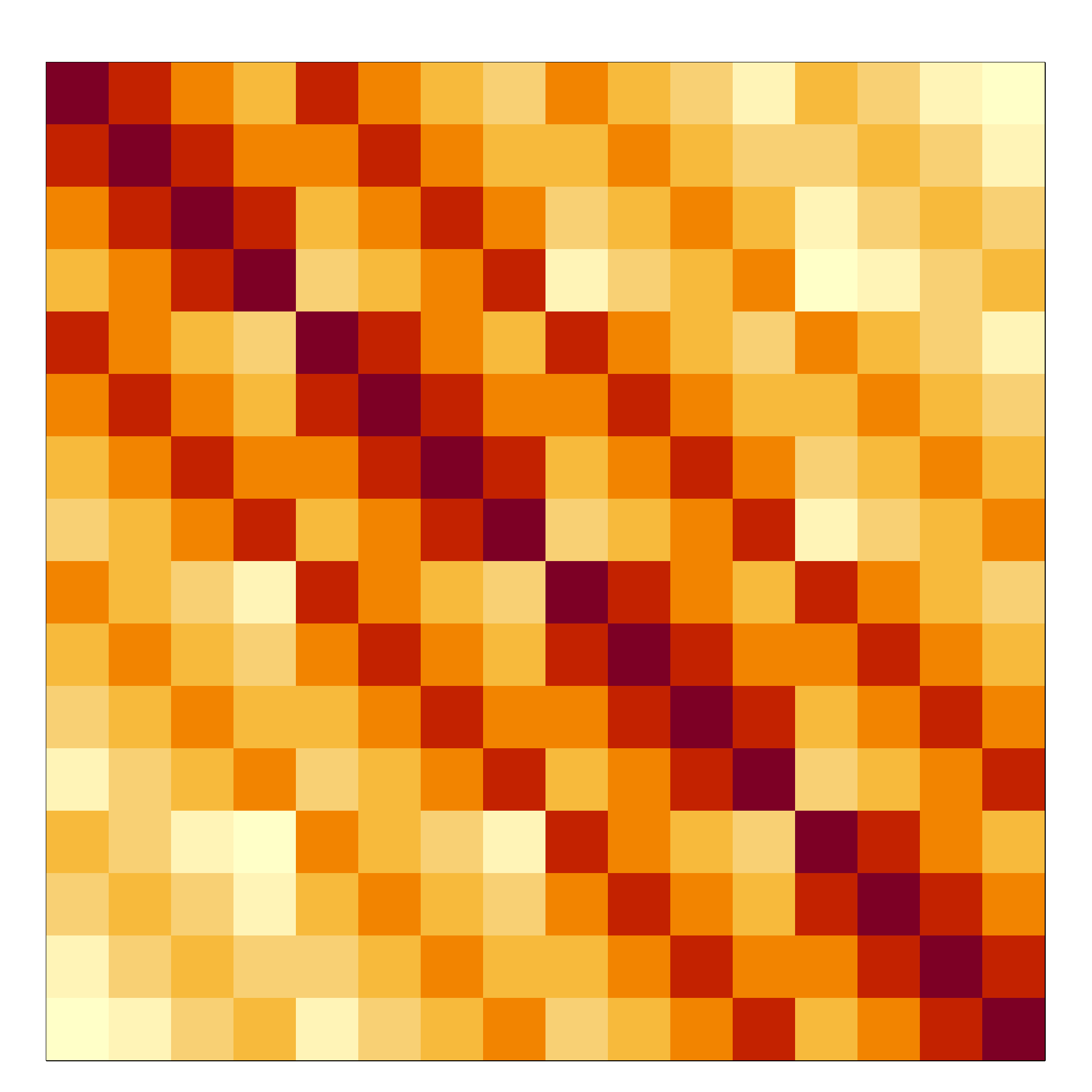}
    \caption{$\lambda=8$}
  \end{subfigure}
\end{minipage}
\vspace{1em}
\begin{minipage}[c]{0.06\textwidth}
  \centering
  \rotatebox{90}{$\beta = 1$}
\end{minipage}\hfill
\begin{minipage}[c]{0.92\textwidth}
  \centering
  \begin{subfigure}[t]{0.24\textwidth}
    \centering
    \includegraphics[width=\linewidth,height=0.18\textheight,keepaspectratio]{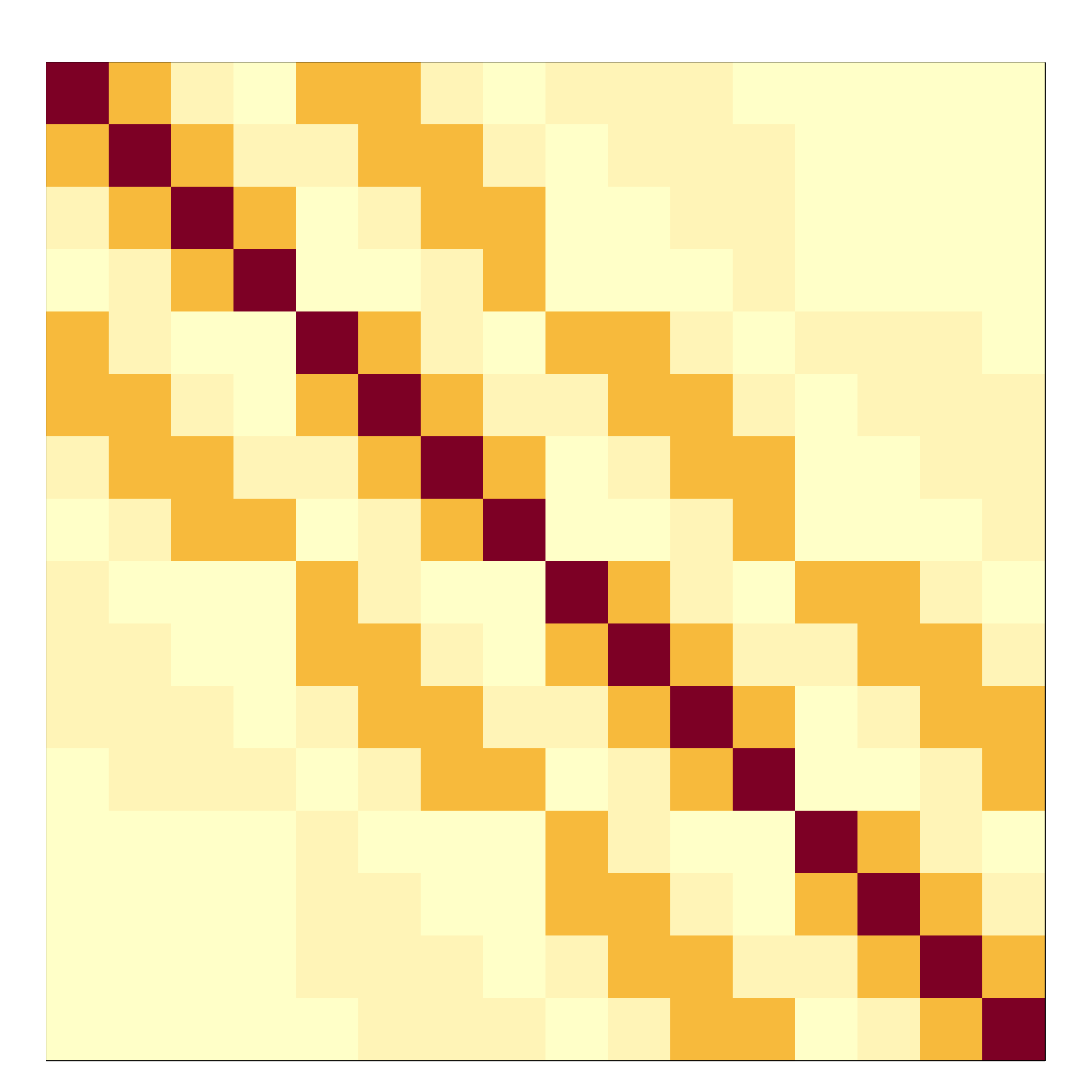}
    \caption{$\lambda=2$}
  \end{subfigure}\hfill
  \begin{subfigure}[t]{0.24\textwidth}
    \centering
    \includegraphics[width=\linewidth,height=0.18\textheight,keepaspectratio]{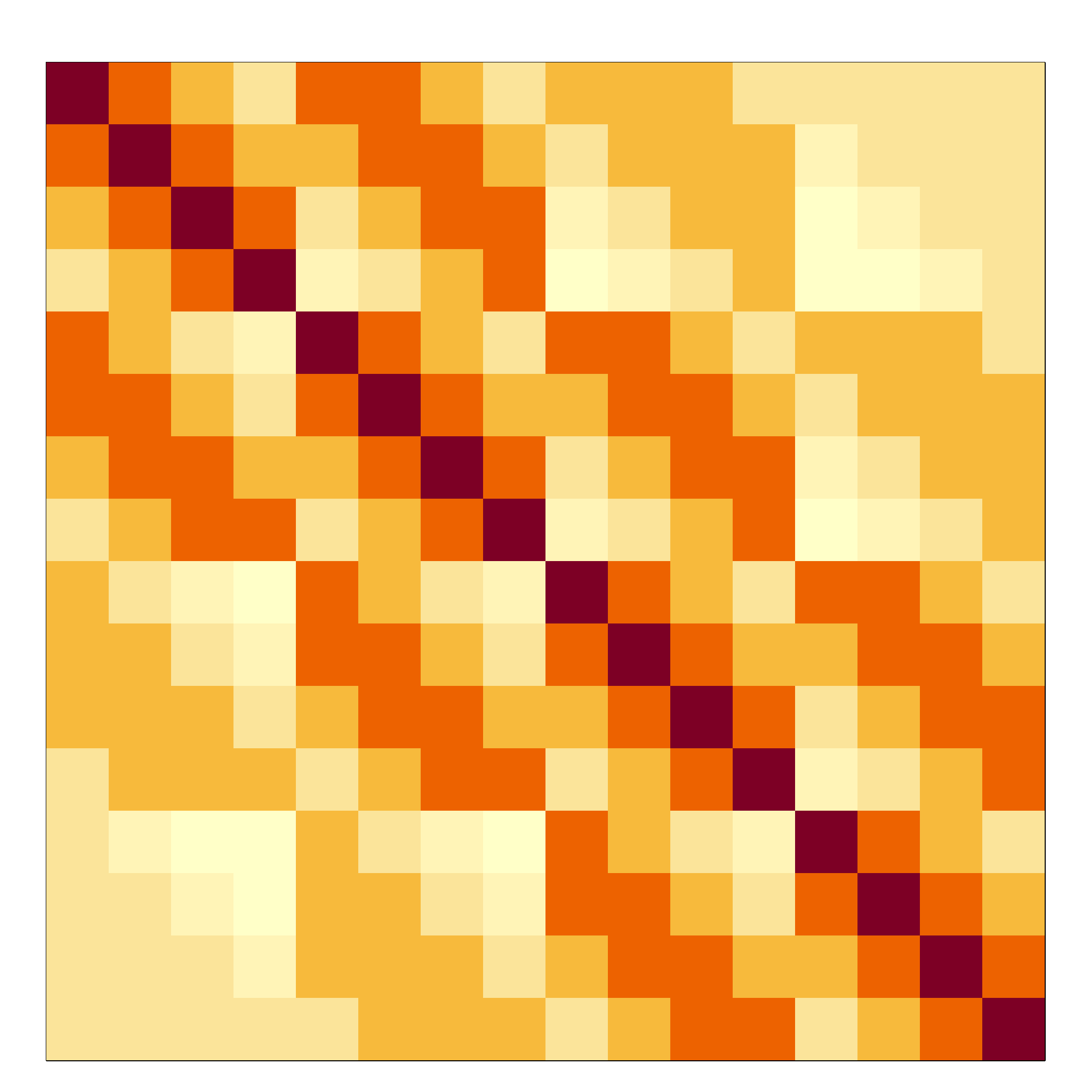}
    \caption{$\lambda=4$}
  \end{subfigure}\hfill
  \begin{subfigure}[t]{0.24\textwidth}
    \centering
    \includegraphics[width=\linewidth,height=0.18\textheight,keepaspectratio]{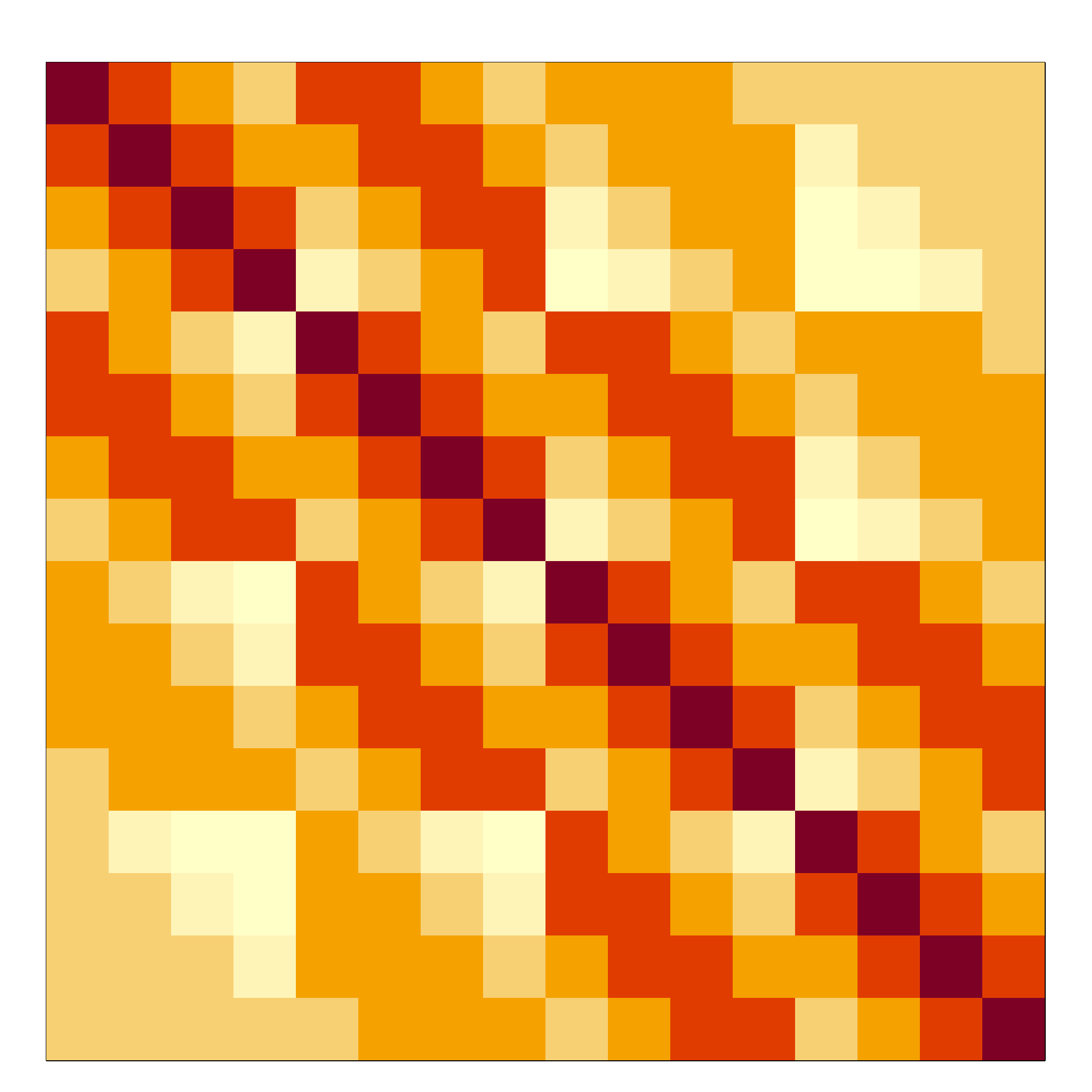}
    \caption{$\lambda=6$}
  \end{subfigure}\hfill
  \begin{subfigure}[t]{0.24\textwidth}
    \centering
    \includegraphics[width=\linewidth,height=0.18\textheight,keepaspectratio]{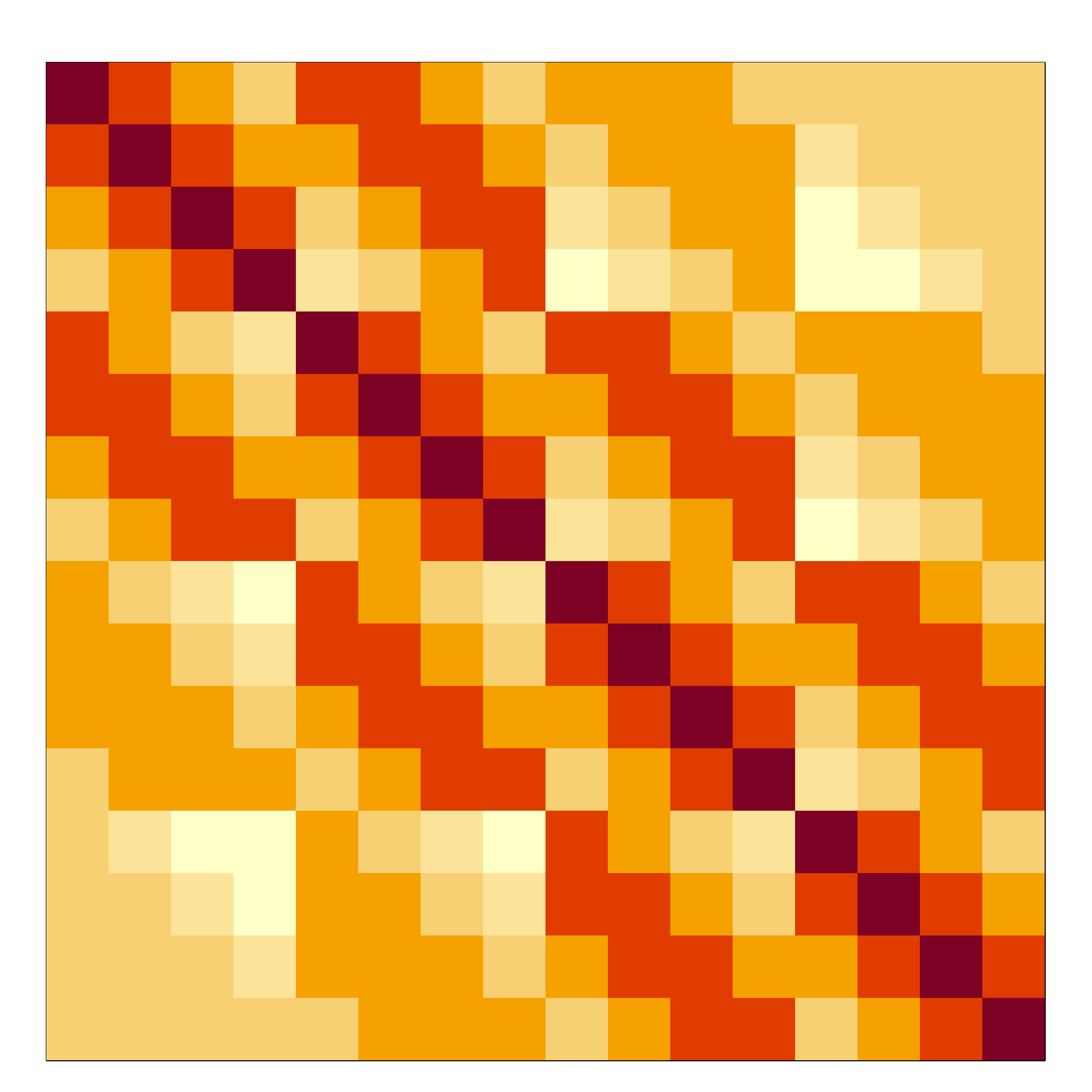}
    \caption{$\lambda=8$}
  \end{subfigure}
\end{minipage}
\caption{True $(16\times 16))$ covariance matrices $\mathbf{\Sigma}_{\lambda,16}$ for separable ($\beta=0$; first row) and non-separable ($\beta=1$; second row) dependence across increasing spatial dependence $\lambda$. Grid size is $(4\times 4)$.}
\label{fig:truecov_all16}
\end{figure}

We then draw the realizations of the lattice process from a multivariate normal distribution
$\vect(\mathbf{X}_{\underline{n}}) \sim \mathcal{N}\bigl(0,\, \boldsymbol{\Sigma}_{\lambda,N}\bigr)$
and we use these realizations to estimate the covariances. We draw \(10,000\) Monte Carlo replications for each setting that we examine, in order to evaluate the bias, variance and MSE of the estimators under each combinations of $\lambda\in \{ 2,4,6,8\}$ and sample sizes $N=n_1n_2 \in \{16,100,400,900\} $, for convenience we have restricted the simulations to squared grids with $n_1=n_2$. In Figures \ref{fig:truecov_all16}, we show the true covariance matrices for $N=16$ and for $\beta=0$ and $\beta=1$, respectively.  

Additionally, we compare the naive (uncorrected) and the bias-corrected estimates with a competitor variogram-based covariance estimator for the covariance function that is given by
\[
\widehat{C}_{CH}(\underline{h})=\widehat{C}(\underline{0})-\widehat{\gamma}(\underline{h}),
\]
where we use the popular variogram estimator proposed by \citet{cressie1980robust}.
\begin{equation}
2 \widehat{\gamma}(\underline{h})
:=\left(\frac{1}{N_h} 
\sum_{i=1}^{N_h}
\left|X\left(\underline{t}_{i+\underline{h}}\right)-X\left(\underline{t}_i\right)\right|^{1/2}\right)^4 
\bigg/ (0.457+0.494 / N_h)
\label{eq:cressiehawkins}
\end{equation}

\subsection{Separable case} 
 We analyse the results of the case in which the true covariance $C_{\lambda, N}(\un h)$ from Equation~(\ref{eq:trueautocov}) has a separable parametrization $(\beta = 0)$. In order to aggregate the results of two dimensional lattices, we often use the Euclidean norm applied to the lag vector
 $\|\underline{h}\|_2$.
We compare the performance of the vector of standard (or as we call below, naive) estimator $\widehat{\un C}$, with entries given by \eqref{eq:sampleAutocov}, and the vector of the Cressie-Hawkins $\widehat{\un C}_{CH}$, with entries given by \eqref{eq:cressiehawkins}, with the proposed bias-corrected estimator $\un{C}_{\un m}^*$ from \eqref{eq:bias_corrected_estimator}.
In Figure~\ref{fig:10x10_bias_MSE_sep}, in the top panel, we show the bias of the three estimators, where the bias-corrected estimator is evaluated with different cutoff parameters $\un m$, where (for simplicity) we take the same cutoff value for each axis dimension, i.e., $m:=m_1=m_2$. The plot shows Monte Carlo averages of the bias as a function of the norm of the lag vector, for the grid of size $n_1=n_2=10$ and for the four levels of spatial dependence determined by $\lambda\in\{2,4,6,8\}$.

\begin{figure}[t]
  \centering
\includegraphics[width=1\linewidth]{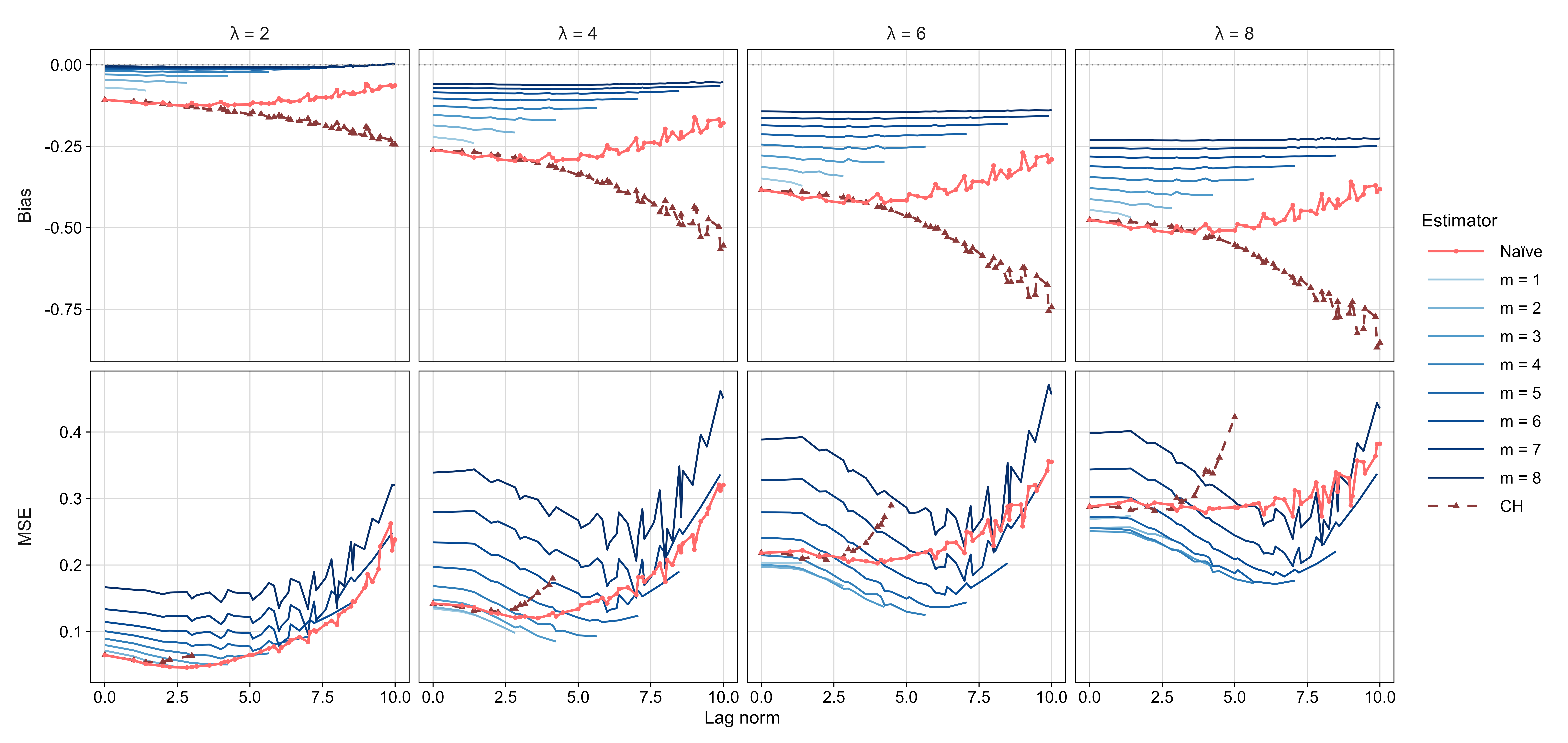}
  \caption{Bias and MSE performance with increasing lag norm $\|\un h\|_2$ for the naive estimator $\widehat{\underline{C}}$ (solid light red line), Cressie-Hawkins $\widehat{\underline{C}}_{CH}$ estimator (dashed dark red line) and the bias-corrected estimator $\widehat{\underline{C}}_{m}^{*}$ for different $m$ values (several blue shades) for four spatial dependence decay parameters $\lambda \in \{2,4,6,8\}$ and for grid size $n_1=n_2=10$. Separable model $\beta =0$.}
  \label{fig:10x10_bias_MSE_sep}
\end{figure}

Across all settings, a reduction in the bias level is already achieved for $m=1$, with increasing improvements as the cutoff parameter $m$ increases.
The naive estimator, shown as the light red solid curve, exhibits larger bias than all bias-corrected estimators over the entire range of lag norms.
The Cressie--Hawkins estimator, shown as the dashed dark red line,  performs slightly better than the naive estimator at small lags, but its bias increases rapidly at higher lags. Similar results are obtained for the other not reported grid sizes.

However, due to the common bias-variance trade-off effect, correcting the bias in the estimation may come at the cost of an increased variance. Hence, a more comprehensive measure for the quality of the estimator is given by the Mean Squared Error (MSE) studied in Section \ref{subsec:MSE_analysis}, which sums the squared biases and the variances of the estimator. When looking at this criterion, the benefits of the bias correction are less immediate. In Figure~\ref{fig:10x10_bias_MSE_sep}, in the bottom panel, we report the MSE results of the same setting. We see that at short range lags, all estimators exhibit similar MSE levels, but (as expected) higher-order bias corrections (corresponding to large values of $ m$) lead to worse performances than simply using the naive estimation. At moderate to high lags, the MSE of the Cressie-Hawkins estimator blows up, reflecting the increasing inappropriateness of the bias correction it introduces. Nevertheless, for strong spatial dependence ($\lambda$ large), the bias-corrected estimator shows to be not only competitive, but superior in comparison to the naive and the Cressie-Hawkins benchmarks also in terms of MSE for suitable choices of $\un m$, reporting lower MSE-levels in the simulations.

We also examine the MSE of each lag pair $\un h = (h_1,h_2)$ separately. For this purpose, Table~\ref{tab:mse_ratio} reports the ratios of the MSE for the bias-corrected estimator $\un C_{\un m}^*$ and the Cressie-Hawkins $\un C_{CH}$ estimator over the standard (bias-uncorrected) one $\un{\widehat{C}}$. Values that are smaller than 1 indicate those instances, where the competitor estimator works better than the standard one. We show that indeed the gain in MSE becomes quite substantial in the strong dependency case, while still some valuable gains are already achieved for short lags with cutoff value $m\in\{1,2\}$.

\begin{table}[t]
\centering
\caption{MSE Ratio $\widehat{C}_m^* / \widehat{C}$ and $\widehat{C}_{CH} / \widehat{C}$ for selected $(h_1, h_2)$ and cutoff parameter $m$.}
\label{tab:mse_ratio}
\tiny
\begin{tabular}{l|l|l|rrrrrrrrrr|r}
\hline
$\lambda$ & $N$ & $(h_1, h_2)$ & $m=1$ & $m=2$ & $m=3$ & $m=4$ & $m=5$ & $m=6$ & $m=7$ & $m=8$ & $m=9$ & $m=10$ & CH \\
\hline
2 & 4 & (0,0) & \cellcolor{gray!20}1.041 & \cellcolor{gray!20}1.605 & -- & -- & -- & -- & -- & -- & -- & -- & 1.000 \\
 &  & (1,0) & 0.956 & \cellcolor{gray!20}1.619 & -- & -- & -- & -- & -- & -- & -- & -- & \cellcolor{gray!20}1.130 \\
 &  & (1,1) & 0.821 & \cellcolor{gray!20}1.538 & -- & -- & -- & -- & -- & -- & -- & -- & \cellcolor{gray!20}1.285 \\
 &  & (2,0) & -- & \cellcolor{gray!20}1.140 & -- & -- & -- & -- & -- & -- & -- & -- & \cellcolor{gray!20}1.663 \\
 &  & (2,2) & -- & 0.900 & -- & -- & -- & -- & -- & -- & -- & -- & \cellcolor{gray!20}3.541 \\
\hline
 & 10 & (0,0) & \cellcolor{gray!20}1.011 & \cellcolor{gray!20}1.105 & \cellcolor{gray!20}1.235 & \cellcolor{gray!20}1.384 & \cellcolor{gray!20}1.558 & \cellcolor{gray!20}1.775 & \cellcolor{gray!20}2.072 & \cellcolor{gray!20}2.583 & -- & -- & 1.000 \\
 &  & (1,0) & 0.994 & \cellcolor{gray!20}1.105 & \cellcolor{gray!20}1.264 & \cellcolor{gray!20}1.448 & \cellcolor{gray!20}1.659 & \cellcolor{gray!20}1.923 & \cellcolor{gray!20}2.280 & \cellcolor{gray!20}2.875 & -- & -- & \cellcolor{gray!20}1.007 \\
 &  & (1,1) & 0.976 & \cellcolor{gray!20}1.104 & \cellcolor{gray!20}1.295 & \cellcolor{gray!20}1.517 & \cellcolor{gray!20}1.772 & \cellcolor{gray!20}2.083 & \cellcolor{gray!20}2.499 & \cellcolor{gray!20}3.174 & -- & -- & \cellcolor{gray!20}1.043 \\
 &  & (2,0) & -- & \cellcolor{gray!20}1.058 & \cellcolor{gray!20}1.249 & \cellcolor{gray!20}1.479 & \cellcolor{gray!20}1.744 & \cellcolor{gray!20}2.071 & \cellcolor{gray!20}2.511 & \cellcolor{gray!20}3.223 & -- & -- & \cellcolor{gray!20}1.130 \\
 &  & (2,2) & -- & 0.966 & \cellcolor{gray!20}1.185 & \cellcolor{gray!20}1.470 & \cellcolor{gray!20}1.807 & \cellcolor{gray!20}2.203 & \cellcolor{gray!20}2.718 & \cellcolor{gray!20}3.498 & -- & -- & \cellcolor{gray!20}1.526 \\
\hline
 & 20 & (0,0) & \cellcolor{gray!20}1.005 & \cellcolor{gray!20}1.034 & \cellcolor{gray!20}1.070 & \cellcolor{gray!20}1.105 & \cellcolor{gray!20}1.136 & \cellcolor{gray!20}1.163 & \cellcolor{gray!20}1.188 & \cellcolor{gray!20}1.213 & \cellcolor{gray!20}1.240 & \cellcolor{gray!20}1.270 & 1.000 \\
 &  & (1,0) & \cellcolor{gray!20}1.003 & \cellcolor{gray!20}1.037 & \cellcolor{gray!20}1.081 & \cellcolor{gray!20}1.124 & \cellcolor{gray!20}1.163 & \cellcolor{gray!20}1.197 & \cellcolor{gray!20}1.228 & \cellcolor{gray!20}1.259 & \cellcolor{gray!20}1.290 & \cellcolor{gray!20}1.327 & 0.996 \\
 &  & (1,1) & 0.999 & \cellcolor{gray!20}1.035 & \cellcolor{gray!20}1.087 & \cellcolor{gray!20}1.138 & \cellcolor{gray!20}1.184 & \cellcolor{gray!20}1.224 & \cellcolor{gray!20}1.261 & \cellcolor{gray!20}1.297 & \cellcolor{gray!20}1.336 & \cellcolor{gray!20}1.380 & \cellcolor{gray!20}1.009 \\
 &  & (2,0) & -- & \cellcolor{gray!20}1.029 & \cellcolor{gray!20}1.084 & \cellcolor{gray!20}1.140 & \cellcolor{gray!20}1.192 & \cellcolor{gray!20}1.238 & \cellcolor{gray!20}1.280 & \cellcolor{gray!20}1.320 & \cellcolor{gray!20}1.360 & \cellcolor{gray!20}1.407 & \cellcolor{gray!20}1.058 \\
 &  & (2,2) & -- & \cellcolor{gray!20}1.008 & \cellcolor{gray!20}1.071 & \cellcolor{gray!20}1.142 & \cellcolor{gray!20}1.210 & \cellcolor{gray!20}1.270 & \cellcolor{gray!20}1.325 & \cellcolor{gray!20}1.377 & \cellcolor{gray!20}1.431 & \cellcolor{gray!20}1.492 & \cellcolor{gray!20}1.208 \\
\hline
 & 30 & (0,0) & \cellcolor{gray!20}1.003 & \cellcolor{gray!20}1.016 & \cellcolor{gray!20}1.033 & \cellcolor{gray!20}1.049 & \cellcolor{gray!20}1.063 & \cellcolor{gray!20}1.073 & \cellcolor{gray!20}1.081 & \cellcolor{gray!20}1.088 & \cellcolor{gray!20}1.094 & \cellcolor{gray!20}1.100 & 1.000 \\
 &  & (1,0) & \cellcolor{gray!20}1.001 & \cellcolor{gray!20}1.016 & \cellcolor{gray!20}1.037 & \cellcolor{gray!20}1.056 & \cellcolor{gray!20}1.072 & \cellcolor{gray!20}1.085 & \cellcolor{gray!20}1.096 & \cellcolor{gray!20}1.104 & \cellcolor{gray!20}1.111 & \cellcolor{gray!20}1.118 & 1.000 \\
 &  & (1,1) & 1.000 & \cellcolor{gray!20}1.017 & \cellcolor{gray!20}1.040 & \cellcolor{gray!20}1.063 & \cellcolor{gray!20}1.082 & \cellcolor{gray!20}1.097 & \cellcolor{gray!20}1.110 & \cellcolor{gray!20}1.120 & \cellcolor{gray!20}1.129 & \cellcolor{gray!20}1.137 & \cellcolor{gray!20}1.014 \\
 &  & (2,0) & -- & \cellcolor{gray!20}1.014 & \cellcolor{gray!20}1.039 & \cellcolor{gray!20}1.064 & \cellcolor{gray!20}1.086 & \cellcolor{gray!20}1.103 & \cellcolor{gray!20}1.117 & \cellcolor{gray!20}1.129 & \cellcolor{gray!20}1.139 & \cellcolor{gray!20}1.148 & \cellcolor{gray!20}1.042 \\
 &  & (2,2) & -- & \cellcolor{gray!20}1.007 & \cellcolor{gray!20}1.036 & \cellcolor{gray!20}1.067 & \cellcolor{gray!20}1.096 & \cellcolor{gray!20}1.120 & \cellcolor{gray!20}1.139 & \cellcolor{gray!20}1.156 & \cellcolor{gray!20}1.170 & \cellcolor{gray!20}1.182 & \cellcolor{gray!20}1.165 \\
\hline\hline
4 & 4 & (0,0) & 0.862 & 0.993 & -- & -- & -- & -- & -- & -- & -- & -- & 1.000 \\
 &  & (1,0) & 0.807 & 0.921 & -- & -- & -- & -- & -- & -- & -- & -- & \cellcolor{gray!20}1.034 \\
 &  & (1,1) & 0.743 & 0.843 & -- & -- & -- & -- & -- & -- & -- & -- & \cellcolor{gray!20}1.031 \\
 &  & (2,0) & -- & 0.766 & -- & -- & -- & -- & -- & -- & -- & -- & \cellcolor{gray!20}1.474 \\
 &  & (2,2) & -- & 0.658 & -- & -- & -- & -- & -- & -- & -- & -- & \cellcolor{gray!20}2.583 \\
\hline
 & 10 & (0,0) & 0.947 & 0.964 & \cellcolor{gray!20}1.044 & \cellcolor{gray!20}1.185 & \cellcolor{gray!20}1.387 & \cellcolor{gray!20}1.647 & \cellcolor{gray!20}1.967 & \cellcolor{gray!20}2.385 & -- & -- & 1.000 \\
 &  & (1,0) & 0.932 & 0.943 & \cellcolor{gray!20}1.027 & \cellcolor{gray!20}1.179 & \cellcolor{gray!20}1.398 & \cellcolor{gray!20}1.679 & \cellcolor{gray!20}2.021 & \cellcolor{gray!20}2.459 & -- & -- & 0.983 \\
 &  & (1,1) & 0.916 & 0.920 & \cellcolor{gray!20}1.006 & \cellcolor{gray!20}1.168 & \cellcolor{gray!20}1.403 & \cellcolor{gray!20}1.701 & \cellcolor{gray!20}2.064 & \cellcolor{gray!20}2.520 & -- & -- & 0.959 \\
 &  & (2,0) & -- & 0.896 & 0.975 & \cellcolor{gray!20}1.136 & \cellcolor{gray!20}1.376 & \cellcolor{gray!20}1.685 & \cellcolor{gray!20}2.063 & \cellcolor{gray!20}2.537 & -- & -- & \cellcolor{gray!20}1.034 \\
 &  & (2,2) & -- & 0.812 & 0.879 & \cellcolor{gray!20}1.052 & \cellcolor{gray!20}1.323 & \cellcolor{gray!20}1.675 & \cellcolor{gray!20}2.103 & \cellcolor{gray!20}2.626 & -- & -- & \cellcolor{gray!20}1.105 \\
\hline
 & 20 & (0,0) & 0.987 & 0.995 & \cellcolor{gray!20}1.026 & \cellcolor{gray!20}1.074 & \cellcolor{gray!20}1.134 & \cellcolor{gray!20}1.200 & \cellcolor{gray!20}1.271 & \cellcolor{gray!20}1.347 & \cellcolor{gray!20}1.429 & \cellcolor{gray!20}1.516 & 1.000 \\
 &  & (1,0) & 0.983 & 0.991 & \cellcolor{gray!20}1.023 & \cellcolor{gray!20}1.074 & \cellcolor{gray!20}1.138 & \cellcolor{gray!20}1.210 & \cellcolor{gray!20}1.286 & \cellcolor{gray!20}1.368 & \cellcolor{gray!20}1.455 & \cellcolor{gray!20}1.548 & 0.976 \\
 &  & (1,1) & 0.980 & 0.986 & \cellcolor{gray!20}1.019 & \cellcolor{gray!20}1.075 & \cellcolor{gray!20}1.144 & \cellcolor{gray!20}1.222 & \cellcolor{gray!20}1.305 & \cellcolor{gray!20}1.394 & \cellcolor{gray!20}1.487 & \cellcolor{gray!20}1.587 & 0.958 \\
 &  & (2,0) & -- & 0.979 & \cellcolor{gray!20}1.011 & \cellcolor{gray!20}1.066 & \cellcolor{gray!20}1.137 & \cellcolor{gray!20}1.217 & \cellcolor{gray!20}1.304 & \cellcolor{gray!20}1.396 & \cellcolor{gray!20}1.494 & \cellcolor{gray!20}1.598 & 0.975 \\
 &  & (2,2) & -- & 0.960 & 0.991 & \cellcolor{gray!20}1.052 & \cellcolor{gray!20}1.135 & \cellcolor{gray!20}1.231 & \cellcolor{gray!20}1.335 & \cellcolor{gray!20}1.446 & \cellcolor{gray!20}1.564 & \cellcolor{gray!20}1.689 & 0.978 \\
\hline
 & 30 & (0,0) & 0.995 & 1.000 & \cellcolor{gray!20}1.016 & \cellcolor{gray!20}1.040 & \cellcolor{gray!20}1.070 & \cellcolor{gray!20}1.103 & \cellcolor{gray!20}1.136 & \cellcolor{gray!20}1.169 & \cellcolor{gray!20}1.201 & \cellcolor{gray!20}1.232 & 1.000 \\
 &  & (1,0) & 0.994 & 0.998 & \cellcolor{gray!20}1.015 & \cellcolor{gray!20}1.041 & \cellcolor{gray!20}1.073 & \cellcolor{gray!20}1.108 & \cellcolor{gray!20}1.144 & \cellcolor{gray!20}1.179 & \cellcolor{gray!20}1.214 & \cellcolor{gray!20}1.247 & 0.982 \\
 &  & (1,1) & 0.992 & 0.996 & \cellcolor{gray!20}1.013 & \cellcolor{gray!20}1.041 & \cellcolor{gray!20}1.075 & \cellcolor{gray!20}1.113 & \cellcolor{gray!20}1.151 & \cellcolor{gray!20}1.189 & \cellcolor{gray!20}1.226 & \cellcolor{gray!20}1.262 & 0.965 \\
 &  & (2,0) & -- & 0.994 & \cellcolor{gray!20}1.011 & \cellcolor{gray!20}1.040 & \cellcolor{gray!20}1.076 & \cellcolor{gray!20}1.115 & \cellcolor{gray!20}1.155 & \cellcolor{gray!20}1.196 & \cellcolor{gray!20}1.235 & \cellcolor{gray!20}1.272 & 0.977 \\
 &  & (2,2) & -- & 0.987 & \cellcolor{gray!20}1.004 & \cellcolor{gray!20}1.035 & \cellcolor{gray!20}1.076 & \cellcolor{gray!20}1.123 & \cellcolor{gray!20}1.171 & \cellcolor{gray!20}1.219 & \cellcolor{gray!20}1.266 & \cellcolor{gray!20}1.312 & 0.974 \\
\hline\hline
6 & 4 & (0,0) & 0.855 & 0.852 & -- & -- & -- & -- & -- & -- & -- & -- & 1.000 \\
 &  & (1,0) & 0.822 & 0.803 & -- & -- & -- & -- & -- & -- & -- & -- & \cellcolor{gray!20}1.021 \\
 &  & (1,1) & 0.784 & 0.750 & -- & -- & -- & -- & -- & -- & -- & -- & 0.997 \\
 &  & (2,0) & -- & 0.734 & -- & -- & -- & -- & -- & -- & -- & -- & \cellcolor{gray!20}1.337 \\
 &  & (2,2) & -- & 0.670 & -- & -- & -- & -- & -- & -- & -- & -- & \cellcolor{gray!20}1.982 \\
\hline
 & 10 & (0,0) & 0.933 & 0.904 & 0.919 & 0.984 & \cellcolor{gray!20}1.105 & \cellcolor{gray!20}1.280 & \cellcolor{gray!20}1.500 & \cellcolor{gray!20}1.781 & -- & -- & 1.000 \\
 &  & (1,0) & 0.923 & 0.887 & 0.898 & 0.964 & \cellcolor{gray!20}1.088 & \cellcolor{gray!20}1.268 & \cellcolor{gray!20}1.494 & \cellcolor{gray!20}1.776 & -- & -- & 0.983 \\
 &  & (1,1) & 0.912 & 0.868 & 0.875 & 0.940 & \cellcolor{gray!20}1.067 & \cellcolor{gray!20}1.252 & \cellcolor{gray!20}1.484 & \cellcolor{gray!20}1.769 & -- & -- & 0.946 \\
 &  & (2,0) & -- & 0.858 & 0.860 & 0.922 & \cellcolor{gray!20}1.048 & \cellcolor{gray!20}1.234 & \cellcolor{gray!20}1.468 & \cellcolor{gray!20}1.761 & -- & -- & \cellcolor{gray!20}1.008 \\
 &  & (2,2) & -- & 0.802 & 0.785 & 0.838 & 0.965 & \cellcolor{gray!20}1.160 & \cellcolor{gray!20}1.407 & \cellcolor{gray!20}1.710 & -- & -- & \cellcolor{gray!20}1.006 \\
\hline
 & 20 & (0,0) & 0.978 & 0.967 & 0.975 & \cellcolor{gray!20}1.001 & \cellcolor{gray!20}1.046 & \cellcolor{gray!20}1.105 & \cellcolor{gray!20}1.177 & \cellcolor{gray!20}1.262 & \cellcolor{gray!20}1.358 & \cellcolor{gray!20}1.466 & 1.000 \\
 &  & (1,0) & 0.975 & 0.962 & 0.969 & 0.995 & \cellcolor{gray!20}1.041 & \cellcolor{gray!20}1.103 & \cellcolor{gray!20}1.178 & \cellcolor{gray!20}1.267 & \cellcolor{gray!20}1.367 & \cellcolor{gray!20}1.480 & 0.979 \\
 &  & (1,1) & 0.972 & 0.957 & 0.962 & 0.989 & \cellcolor{gray!20}1.037 & \cellcolor{gray!20}1.101 & \cellcolor{gray!20}1.180 & \cellcolor{gray!20}1.272 & \cellcolor{gray!20}1.377 & \cellcolor{gray!20}1.495 & 0.953 \\
 &  & (2,0) & -- & 0.953 & 0.956 & 0.982 & \cellcolor{gray!20}1.030 & \cellcolor{gray!20}1.095 & \cellcolor{gray!20}1.175 & \cellcolor{gray!20}1.269 & \cellcolor{gray!20}1.377 & \cellcolor{gray!20}1.497 & 0.972 \\
 &  & (2,2) & -- & 0.936 & 0.934 & 0.959 & \cellcolor{gray!20}1.008 & \cellcolor{gray!20}1.079 & \cellcolor{gray!20}1.168 & \cellcolor{gray!20}1.272 & \cellcolor{gray!20}1.392 & \cellcolor{gray!20}1.526 & 0.945 \\
\hline
 & 30 & (0,0) & 0.990 & 0.987 & 0.992 & \cellcolor{gray!20}1.007 & \cellcolor{gray!20}1.032 & \cellcolor{gray!20}1.064 & \cellcolor{gray!20}1.102 & \cellcolor{gray!20}1.144 & \cellcolor{gray!20}1.189 & \cellcolor{gray!20}1.236 & 1.000 \\
 &  & (1,0) & 0.989 & 0.985 & 0.990 & \cellcolor{gray!20}1.005 & \cellcolor{gray!20}1.031 & \cellcolor{gray!20}1.064 & \cellcolor{gray!20}1.103 & \cellcolor{gray!20}1.147 & \cellcolor{gray!20}1.194 & \cellcolor{gray!20}1.243 & 0.983 \\
 &  & (1,1) & 0.988 & 0.983 & 0.987 & \cellcolor{gray!20}1.003 & \cellcolor{gray!20}1.030 & \cellcolor{gray!20}1.064 & \cellcolor{gray!20}1.105 & \cellcolor{gray!20}1.151 & \cellcolor{gray!20}1.200 & \cellcolor{gray!20}1.251 & 0.961 \\
 &  & (2,0) & -- & 0.981 & 0.985 & \cellcolor{gray!20}1.001 & \cellcolor{gray!20}1.027 & \cellcolor{gray!20}1.063 & \cellcolor{gray!20}1.105 & \cellcolor{gray!20}1.151 & \cellcolor{gray!20}1.202 & \cellcolor{gray!20}1.255 & 0.974 \\
 &  & (2,2) & -- & 0.975 & 0.976 & 0.992 & \cellcolor{gray!20}1.020 & \cellcolor{gray!20}1.058 & \cellcolor{gray!20}1.105 & \cellcolor{gray!20}1.157 & \cellcolor{gray!20}1.213 & \cellcolor{gray!20}1.272 & 0.949 \\
\hline\hline
8 & 4 & (0,0) & 0.870 & 0.826 & -- & -- & -- & -- & -- & -- & -- & -- & 1.000 \\
 &  & (1,0) & 0.847 & 0.790 & -- & -- & -- & -- & -- & -- & -- & -- & \cellcolor{gray!20}1.019 \\
 &  & (1,1) & 0.821 & 0.749 & -- & -- & -- & -- & -- & -- & -- & -- & 0.984 \\
 &  & (2,0) & -- & 0.755 & -- & -- & -- & -- & -- & -- & -- & -- & \cellcolor{gray!20}1.255 \\
 &  & (2,2) & -- & 0.712 & -- & -- & -- & -- & -- & -- & -- & -- & \cellcolor{gray!20}1.649 \\
\hline
 & 10 & (0,0) & 0.935 & 0.890 & 0.871 & 0.888 & 0.947 & \cellcolor{gray!20}1.050 & \cellcolor{gray!20}1.194 & \cellcolor{gray!20}1.384 & -- & -- & 1.000 \\
 &  & (1,0) & 0.927 & 0.876 & 0.853 & 0.868 & 0.927 & \cellcolor{gray!20}1.032 & \cellcolor{gray!20}1.179 & \cellcolor{gray!20}1.370 & -- & -- & 0.983 \\
 &  & (1,1) & 0.920 & 0.861 & 0.834 & 0.846 & 0.905 & \cellcolor{gray!20}1.011 & \cellcolor{gray!20}1.158 & \cellcolor{gray!20}1.347 & -- & -- & 0.948 \\
 &  & (2,0) & -- & 0.856 & 0.825 & 0.833 & 0.889 & 0.993 & \cellcolor{gray!20}1.140 & \cellcolor{gray!20}1.331 & -- & -- & 0.999 \\
 &  & (2,2) & -- & 0.818 & 0.770 & 0.767 & 0.819 & 0.923 & \cellcolor{gray!20}1.072 & \cellcolor{gray!20}1.264 & -- & -- & 0.980 \\
\hline
 & 20 & (0,0) & 0.974 & 0.954 & 0.945 & 0.951 & 0.974 & \cellcolor{gray!20}1.012 & \cellcolor{gray!20}1.065 & \cellcolor{gray!20}1.133 & \cellcolor{gray!20}1.217 & \cellcolor{gray!20}1.315 & 1.000 \\
 &  & (1,0) & 0.972 & 0.949 & 0.939 & 0.944 & 0.966 & \cellcolor{gray!20}1.005 & \cellcolor{gray!20}1.060 & \cellcolor{gray!20}1.130 & \cellcolor{gray!20}1.216 & \cellcolor{gray!20}1.317 & 0.980 \\
 &  & (1,1) & 0.970 & 0.945 & 0.932 & 0.937 & 0.959 & 0.999 & \cellcolor{gray!20}1.055 & \cellcolor{gray!20}1.127 & \cellcolor{gray!20}1.215 & \cellcolor{gray!20}1.319 & 0.956 \\
 &  & (2,0) & -- & 0.943 & 0.928 & 0.932 & 0.953 & 0.993 & \cellcolor{gray!20}1.049 & \cellcolor{gray!20}1.121 & \cellcolor{gray!20}1.210 & \cellcolor{gray!20}1.316 & 0.975 \\
 &  & (2,2) & -- & 0.930 & 0.909 & 0.908 & 0.928 & 0.968 & \cellcolor{gray!20}1.027 & \cellcolor{gray!20}1.104 & \cellcolor{gray!20}1.199 & \cellcolor{gray!20}1.311 & 0.945 \\
\hline
 & 30 & (0,0) & 0.988 & 0.978 & 0.975 & 0.980 & 0.994 & \cellcolor{gray!20}1.016 & \cellcolor{gray!20}1.047 & \cellcolor{gray!20}1.085 & \cellcolor{gray!20}1.128 & \cellcolor{gray!20}1.177 & 1.000 \\
 &  & (1,0) & 0.987 & 0.976 & 0.972 & 0.976 & 0.990 & \cellcolor{gray!20}1.014 & \cellcolor{gray!20}1.045 & \cellcolor{gray!20}1.084 & \cellcolor{gray!20}1.128 & \cellcolor{gray!20}1.178 & 0.982 \\
 &  & (1,1) & 0.986 & 0.975 & 0.969 & 0.973 & 0.987 & \cellcolor{gray!20}1.011 & \cellcolor{gray!20}1.043 & \cellcolor{gray!20}1.083 & \cellcolor{gray!20}1.129 & \cellcolor{gray!20}1.180 & 0.963 \\
 &  & (2,0) & -- & 0.973 & 0.967 & 0.971 & 0.984 & \cellcolor{gray!20}1.008 & \cellcolor{gray!20}1.041 & \cellcolor{gray!20}1.081 & \cellcolor{gray!20}1.127 & \cellcolor{gray!20}1.180 & 0.974 \\
 &  & (2,2) & -- & 0.968 & 0.959 & 0.961 & 0.975 & 0.999 & \cellcolor{gray!20}1.033 & \cellcolor{gray!20}1.076 & \cellcolor{gray!20}1.126 & \cellcolor{gray!20}1.182 & 0.944 \\
\hline
\end{tabular}
\end{table}

Additionally, to visualize the MSE results from Table \ref{tab:mse_ratio}, we show heatmap plots in Figures \ref{fig:masks10} (for $n_1=n_2=10$) and \ref{fig:masks20} (for $n_1=n_2=20$) in the appendix. There, for each lag pair $(h_1,h_2)$ shown in each plot axis, it is indicated whether the bias-corrected estimator $\un C_{\un m}^*$ can achieve a better MSE than the uncorrected estimator $\widehat{\un C}$ and whether simple tapering the covariance estimator to zero leads to the best MSE. We refer to \ref{sec:additional_plots} for details.

\begin{figure}[t]
  \centering
\includegraphics[width=1\linewidth]{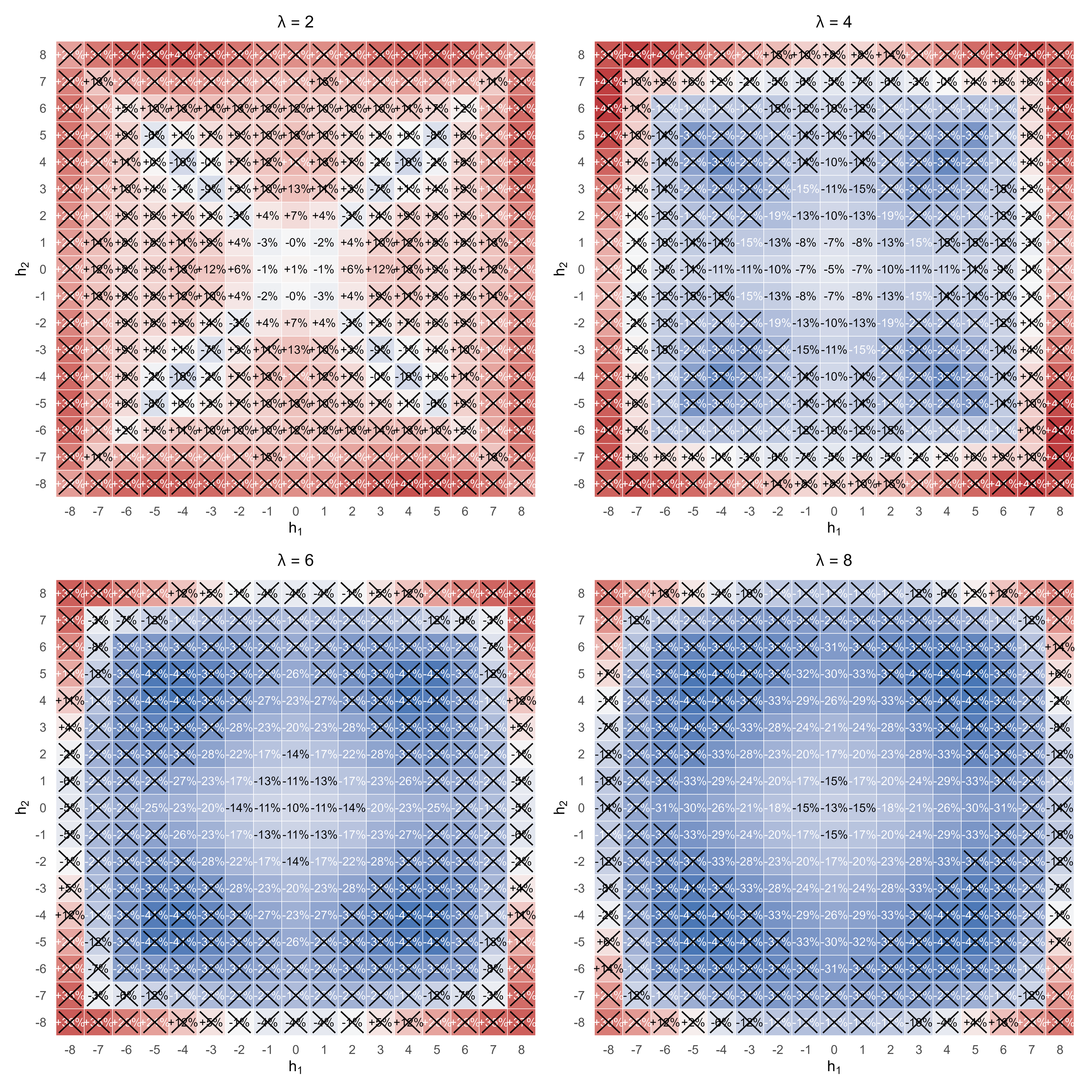}
  \caption{MSE ratio of the estimator with the optimal cutoff $m$ for each lag pair $(h_1,h_2)$ and naive estimator. Negative percentages indicate a reduction in MSE while positive percentages indicate an increase in MSE for the specific lag pair. Equivalently, blue shades indicate a reduction in MSE while red shades indicate a worsening of the MSE when compared to the naive estimator. Grid size $(10\times 10)$. Separable case $\beta = 0$.}
  \label{fig:mse_heatmap_10}
\end{figure}

We now investigate the maximal MSE reduction attainable by selecting, for each lag pair, the optimal cutoff parameter $m$, namely we report the oracle choice for the cutoff parameter that minimizes the MSE of the bias-corrected estimator and compare this with the MSE of the naive estimator. This is indicated in Figure~\ref{fig:mse_heatmap_10}, which shows the relative increase or decrease in MSE compared to the naive estimator, when each lag pair $(h_1,h_2)$, appearing in the plots' axes, is evaluated at its optimal $m$ for the $(10\times 10)$ grid and across the four dependence levels $\lambda$. Negative percentages indicate a reduction in MSE while positive percentages indicate an increase in MSE.
We notice that there are some slight gains in MSE estimation at short lags like $\underline h=(0,1)$ for the weak dependence case ($\lambda$ small), while the gains in MSE also become more substantial  as the dependency structure increases ($\lambda$ large). In these plots, a crossed-out value indicates that tapering the covariance estimator to zero results in a better MSE than using the naive or the bias-corrected estimator. That is, in these cases, one should not (try to) estimate them at all from the data.

\FloatBarrier
Finally, we are interested in reporting the properties of the estimations of the variance covariance matrix $\boldsymbol{\Sigma}$. We consider the naive estimator $\widehat{\boldsymbol{\Sigma}}$, whose entries are given by the standard sample estimator $\widehat{C}(\un h)$ from \eqref{eq:sampleAutocov}, the bias-corrected estimator $\widehat{\boldsymbol{\Sigma}}^*_{\un m}$ defined in \eqref{eq:Sigma_hat_m_corr} and a tapered estimator $\widehat{\boldsymbol{\Sigma}}_{\un m}$, where the entries not selected by the $m$ parameter are tapered by zero, just like in \eqref{eq:Sigma_hat_m_corr}, but the selected entries are given by with the standard estimator $\widehat{C}(\un h)$ from \eqref{eq:sampleAutocov}, such that only the tapering is applied but not the bias correction. Following \citet{Cai2010}, we evaluate in
Figure~\ref{fig:frob2_p_sep} the normalized squared Frobenius norm
\begin{equation*}
    \frac{1}{N}\E\left(\left\|\widehat{\boldsymbol{\Sigma}}^*_{\underline{m}}-\boldsymbol{\Sigma}\right\|_F^2\right).
\end{equation*}
We observe that for suitably selected cutoff values both the bias-corrected and the tapered estimator obtain substantially lower errors than the simple naive estimator. This behavior is well-known and is due to the regularization induced by the tapering present in both $\widehat{\boldsymbol{\Sigma}}^*_{\un m}$ and $\widehat{\boldsymbol{\Sigma}}_{\un m}$, see among others \citep{Ledoit2004, Furrer2006, Bickel2008, Bickel2008thresh}. The additional benefit of the bias-correction occurs for small grids and strong dependence, where $\widehat{\boldsymbol{\Sigma}}^*_{\un m}$ attains lower values than tapering alone done by $\widehat{\boldsymbol{\Sigma}}_{\un m}$. In the other settings, we see that they attain similar levels.

\begin{figure}[ht]
\centering    \includegraphics[width=1\linewidth]{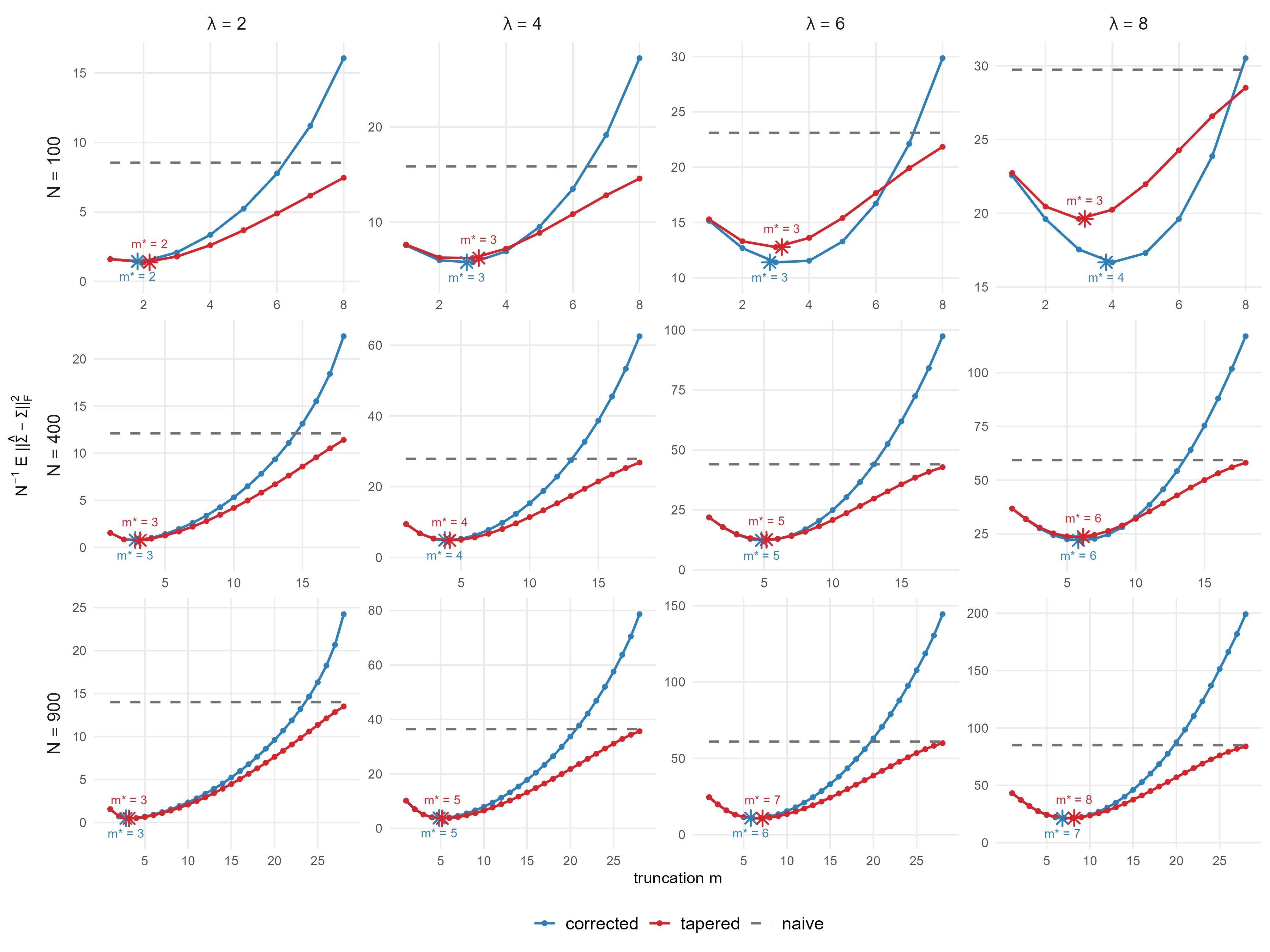}
  \caption{Average Normalized squared Frobenius norm difference of 10,000 Monte Carlo simulations vs different cutoff parameter $m$ comparing the naive estimator  $\boldsymbol{\widehat{\Sigma}}$ (gray), the bias-corrected estimator $\widehat{\boldsymbol{\Sigma}}^*_{\un m}$ (blue) and the tapered estimator $\widehat{\boldsymbol{\Sigma}}_{\un m}$ (red)}
  \label{fig:frob2_p_sep}
\end{figure}

To complement the aggregate error measure given by the Frobenius criterion, we also consider the spectral or operator norm, which for symmetric matrices, it corresponds to the absolute largest eigenvalue. Spectral norm consistency is often investigated in the literature, see e.g. \citet{Bickel2008}, thus we report in Figure~\ref{fig:spectral_sep} the spectral norm
\begin{equation*}
\E\left(\left\|\widehat{\boldsymbol{\Sigma}}^*_{\underline{m}}-\boldsymbol{\Sigma}\right\|_2\right).
\end{equation*}
The results of the simulations broadly confirm that the evidence presented in Figure~\ref{fig:frob2_p_sep}, where the bias-correction estimator achieves visible improvements in settings with small grid sizes and strong dependence. However, we also see improvements from simple tapering also in the moderate and large grid sizes and when the dependence is moderate or strong. 
\begin{figure}[ht]
\centering    \includegraphics[width=1\linewidth]{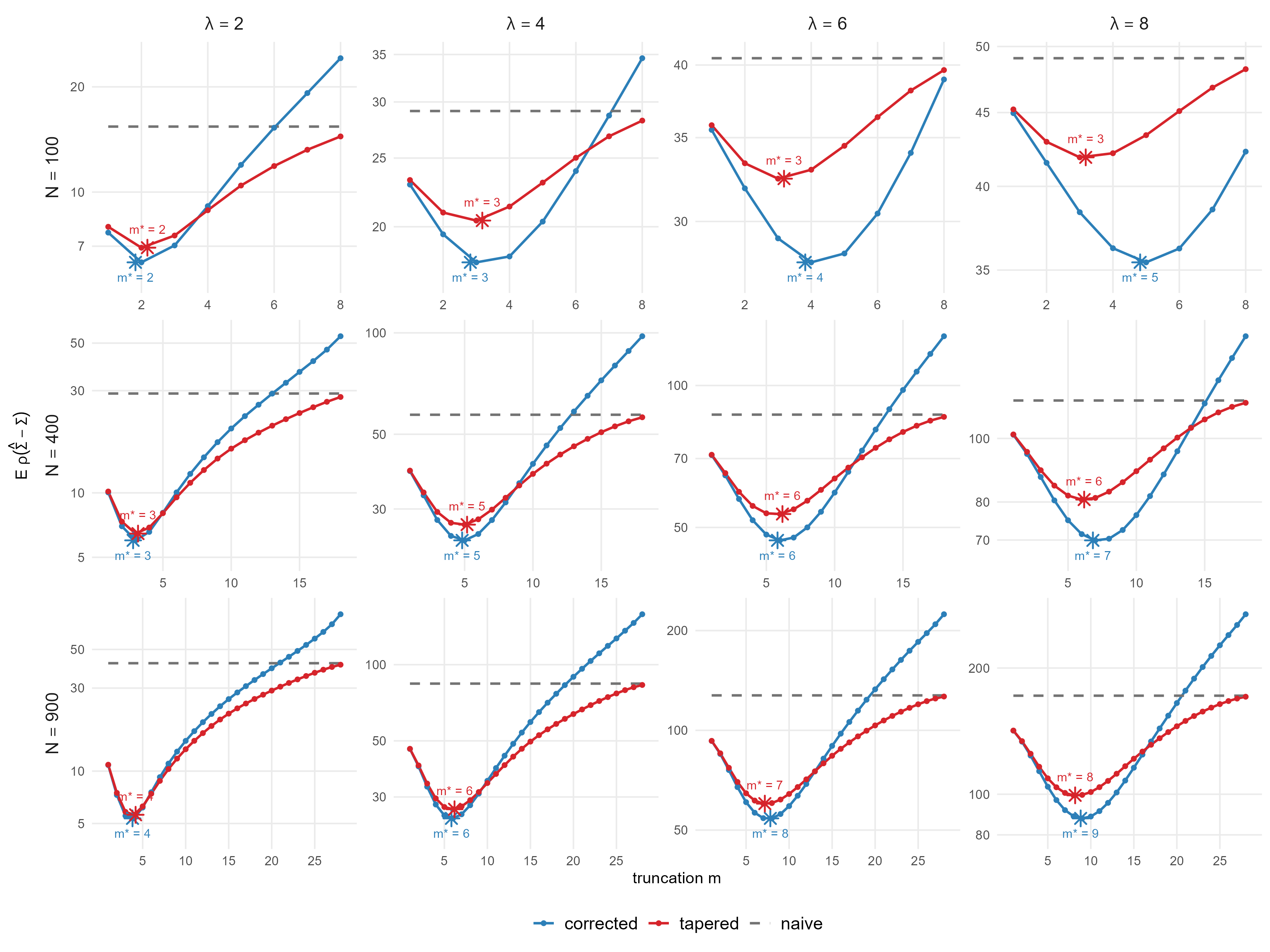}
  \caption{Average spectral norm difference of 10,000 Monte Carlo simulations vs different cutoff parameter $m$ comparing the naive estimator  $\boldsymbol{\widehat{\Sigma}}$ (gray), the bias-corrected estimator $\widehat{\boldsymbol{\Sigma}}^*_{\un m}$ (blue) and the tapered estimator $\widehat{\boldsymbol{\Sigma}}_{\un m}$ (red)}
  \label{fig:spectral_sep}
\end{figure}

\FloatBarrier
\subsection{Nonseparable case} 

Now, we examine the case where our true covariance from Equation~\ref{eq:trueautocov}, where we choose the parameter $\beta=1$ leading to nonseparability of the covariance. 

\begin{figure}[ht]
\centering    \includegraphics[width=1\linewidth]{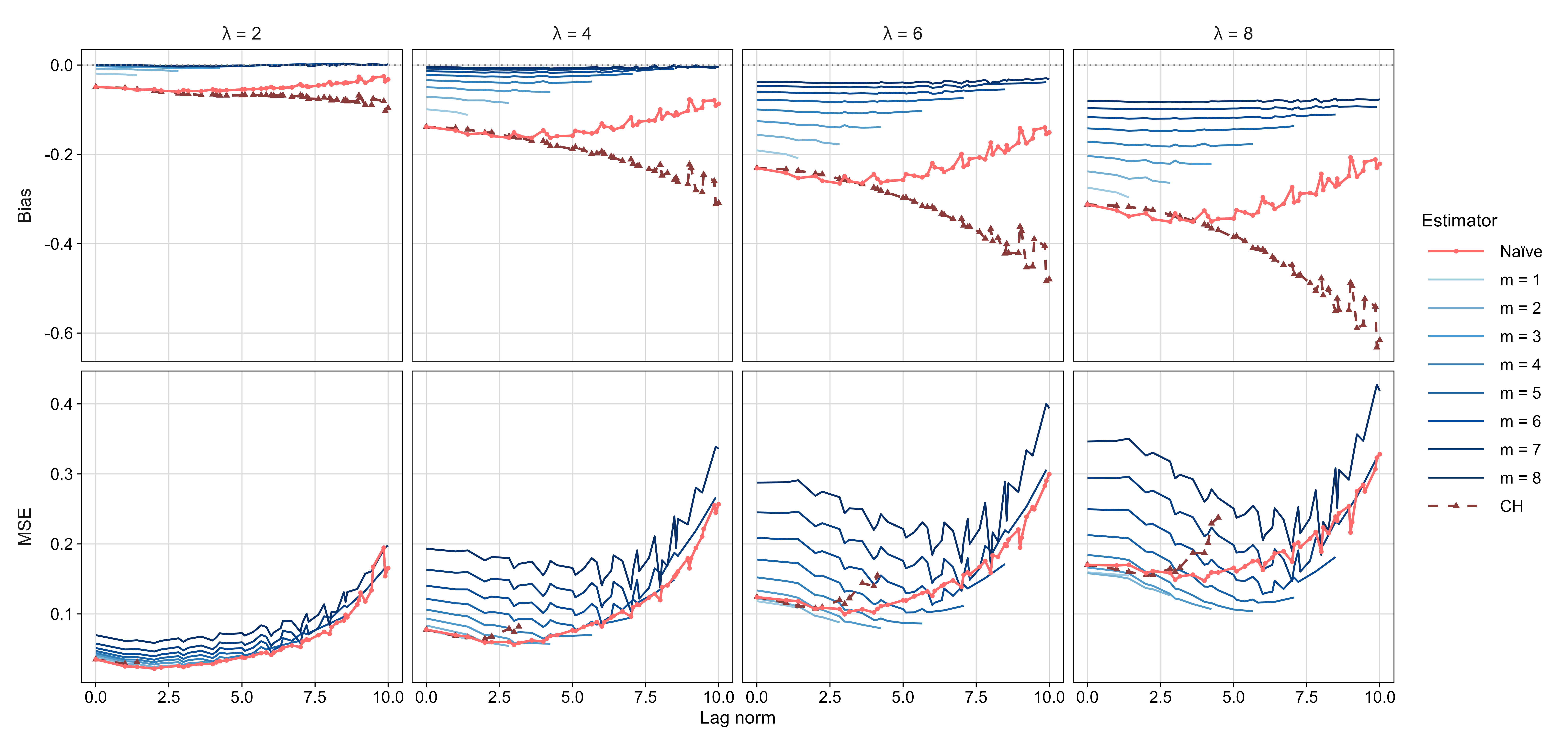}
  \caption{Bias and MSE performance with increasing lag norm $\|\un h\|_2$ for the naive estimator $\widehat{\underline{C}}$ (solid light red line), Cressie-Hawkins $\widehat{\underline{C}}_{CH}$ estimator (dashed dark red line) and the bias-corrected estimator $\widehat{\underline{C}}_{m}^{*}$ for different $m$ values (several blue shades) for four spatial dependence decay parameters $\lambda \in \{2,4,6,8\}$ and for grid size $n_1=n_2=10$. Nonseparable model $\beta =1$}
  \label{fig:Nh_MSE_v_norm}
\end{figure}

Figure~\ref{fig:Nh_MSE_v_norm} depicts the bias and the MSE for the bias-corrected estimator with standard prefactor with the competitor Cressie-Hawkins estimator and the naive estimator. Also here the bias correction is confirmed for almost all $m$ cutoff values. The MSE level has a similar behavior to the separable case: at small levels of spatial dependence, already small values of $m$ cause an inflation in MSE, which increases together with $m$, while as the spatial dependence increases low values of $m$ obtain both a reduction of the bias and in the MSE. The effect of nonseparability becomes visible in Figure~\ref{fig:mse_heatmap_10_nonsep}, where we lose the symmetry from the respective results in Figure~\ref{fig:mse_heatmap_10}, and a selection of lag pairs achieves a better MSE performance if we taper the estimation for that lag to zero. The bias correction method leads to better results with increasing strength of the spatial dependence. Note that, by the nature of the true covariance model \eqref{eq:trueautocov}, the nonseparable case generates overall a weaker dependence structure, because, for $\beta>0$, we introduce a term that additionally supports the decay of the exponential function for increasing lags. This is also seen in Figure~\ref{fig:truecov_all16}, where the nonseparable covariances (second row) appear in lighter shades compared to the respective separable ones (first row). Therefore, bearing this in mind, the results from the separable and non-separable case should be compared carefully. Nonetheless, we can see that the bias-correction achieves valuable results in the MSE sense, as reported in Figure~\ref{fig:mse_heatmap_10_nonsep}. There, for the weakest dependence ($\lambda=2$), the naive estimator is actually better than the bias-corrected one, but in the remaining three settings we see that the bias correction achieves MSE improvements for short lags, while covariances at large lags are indeed better tapered to zero. The non-separable nature of the covariance is manifested by the non-symmetric diagonal slant in the plot, where lag pairs with both positive or both negative lags can be estimated well with either the naive or bias corrected methods, while the cross-lags are better tapered early to zero as the lags increase.

\begin{figure}[t]
  \centering
\includegraphics[width=1\linewidth]{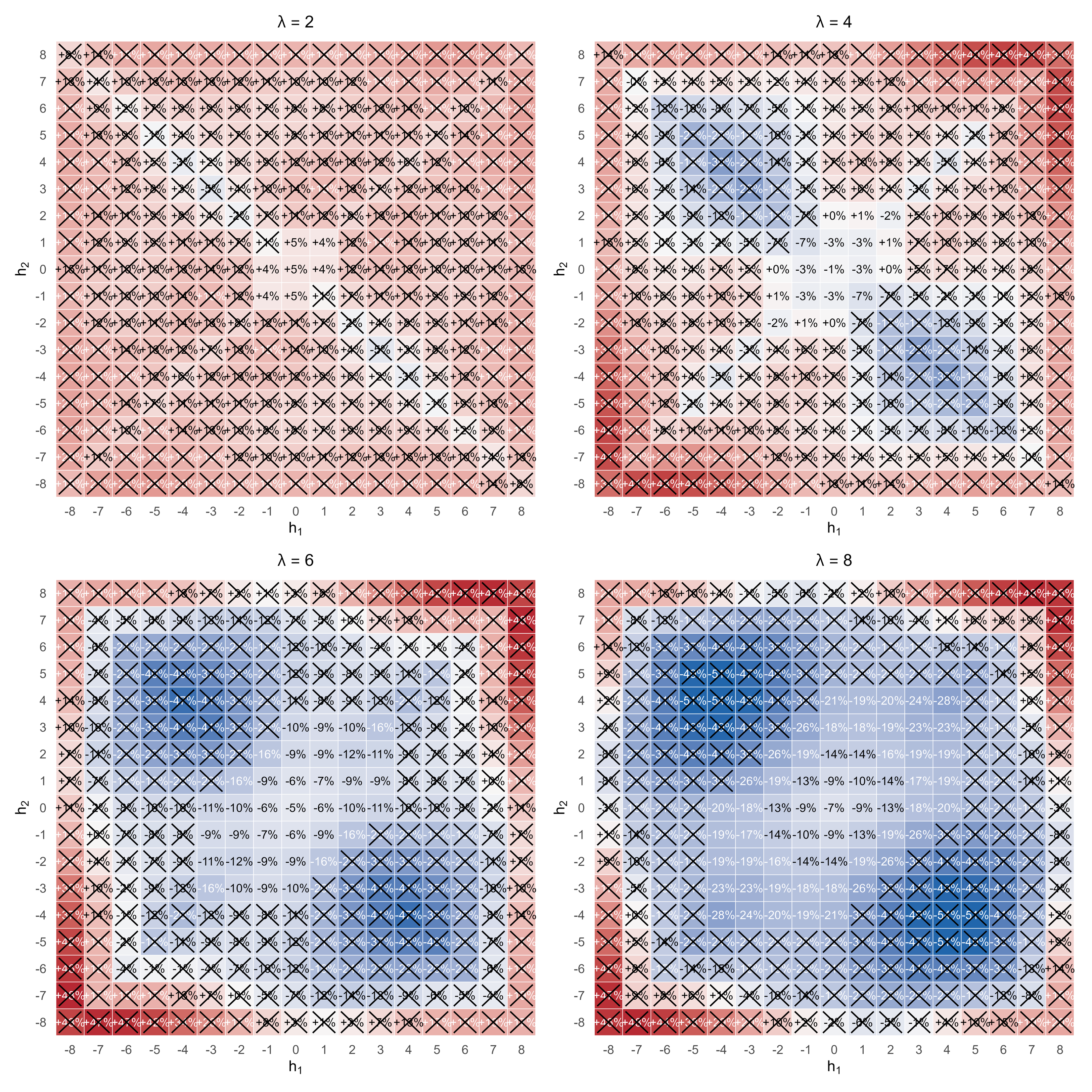}
  \caption{MSE ratio of the estimator with the optimal cutoff $m$ for each lag pair $(h_1,h_2)$ and naive estimator. Negative percentages indicate a reduction in MSE while positive percentages indicate an increase in MSE for the specific lag pair. Equivalently, blue shades indicate a reduction in MSE while red shades indicate a worsening of the MSE when compared to the naive estimator. Grid size $(10\times 10)$. Nonseparable case $\beta=1$.}
  \label{fig:mse_heatmap_10_nonsep}
\end{figure}

\FloatBarrier
\section{Real data illustration: kriging on the topography of the Maungawhau volcano}\label{sec:dataapp}
We illustrate the proposed bias corrected covariance estimation from Equation~\eqref{eq:bias_corrected_estimator} on a real data example, using it to drive a kriging-based prediction task on image data, a common type of lattice realization. The point of the exercise is not to claim a state-of-the-art interpolation
method, but to show that correcting the finite-sample bias of the sample
covariance on a lattice translates into a measurable improvement in a
downstream task that depends entirely on the estimated covariance structure.

\begin{figure}[H]
\centering
\includegraphics[width=\textwidth]{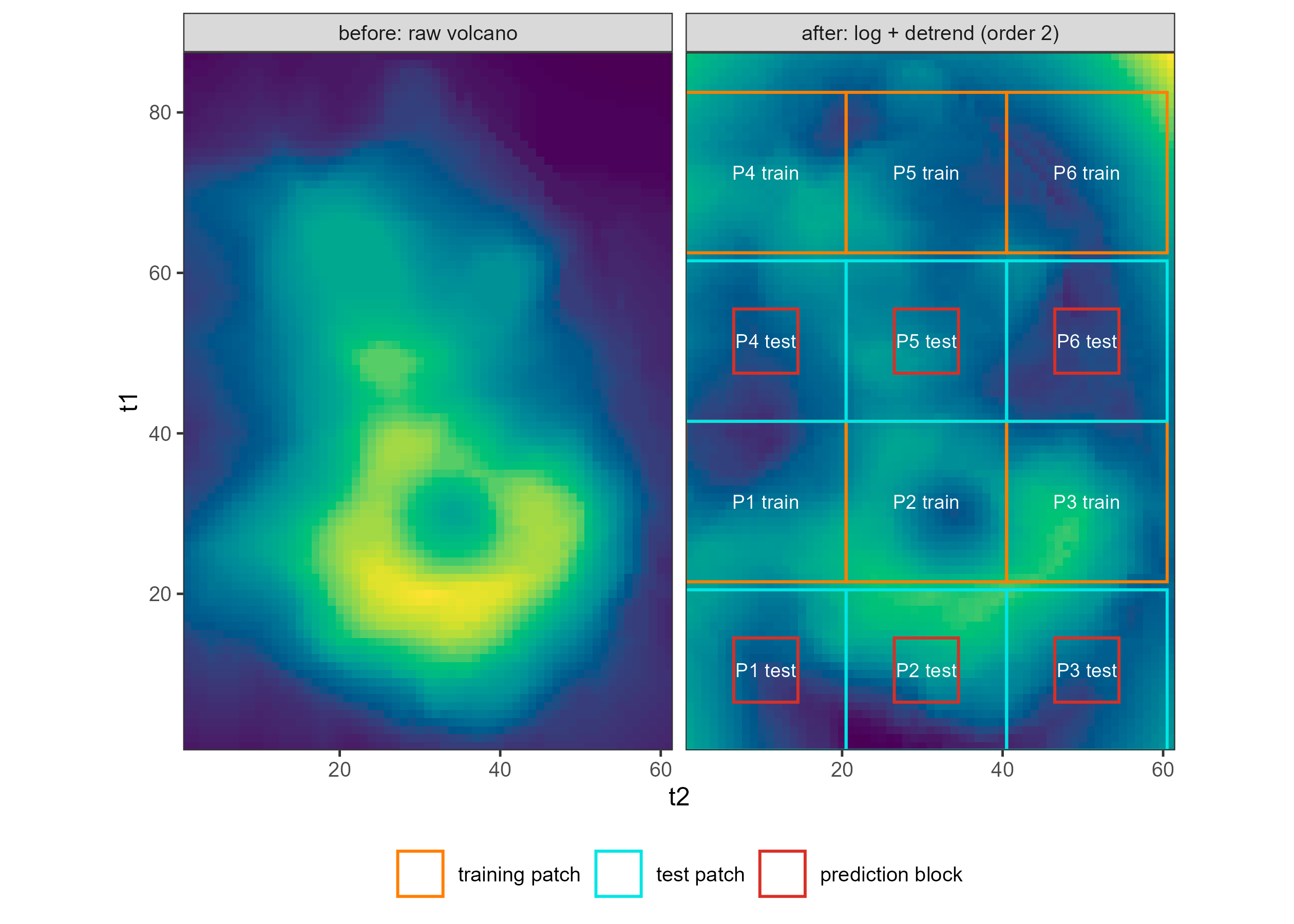}
    \caption{The \texttt{volcano} lattice before (left panel) and after (right panel) the log transform and the subtraction of the fitted order two polynomial surface.}
    \label{fig:volcano_field}
\end{figure}

We consider the \texttt{volcano} dataset, distributed with the \texttt{datasets} package in \textsf{R} \citep{dataset}. It records the topography of Maungawhau (Mount Eden), in Auckland, on a $87\times61$ lattice grid. As can be seen in the left panel of Figure~\ref{fig:volcano_field}, the elevation of the volcano is a cone, such that it does not meet the weak stationarity imposed by Assumption~\ref{ass:stationarity}. Hence, we apply two transformations on the field, namely a log-transform and the detrending described in the following.

With the coordinates $\un t=(t_1,t_2)$, $t_i =1,\dots, n_i$ of the lattice $G_{\un n}$, being the process a record of the elevation, it is always positive $X_{\un t}>0$, thus it admits a log-transform $Y_{\un t}=\log(X_{\tu})$. This transformation is needed because the variability of the elevations grows with their level. Then, an order two polynomial with the coordinates as variables, that is,
\[
\mu_{\un{\beta}}(\un{t}):=\beta_0+\beta_1t_1+\beta_2t_2
+\beta_3t_1^{2}+\beta_4t_2^{2}+\beta_5t_1t_2,
\]
is used and fitted by ordinary least squares for detrending the log-transformed data, such that we work with the residuals $\{y_{\un{t}}\}$ obtained by 
\[\widehat{\un{\beta}}:=\arg\min_{\un{\beta}}
\sum_{\un{t}\in G_{\un{n}}}\bigl(Y_{\un{t}}-\mu_{\un{\beta}}(\un{t})\bigr)^{2}, \qquad y_{\un{t}}:=Y_{\un{t}}-\mu_{\widehat{\un{\beta}}}(\un{t})\]
in the following.

Figure~\ref{fig:volcano_field} shows the raw lattice data before (left panel) and the residuals after the two transformations (right panel). On the right panel, the lattice has been split into 6 training and 6 respective testing patches. We perform ordinary kriging for each patch pair. First, we estimate the variance covariance matrix $\boldsymbol{\Sigma}$ on each of the $(20\times 20)$ training patches. We build four different estimators: two bias-corrected estimators $\widehat{\boldsymbol{\Sigma}}_{\un m}$ from \eqref{eq:Sigma_hat_m_corr}, one with cutoff value $\widehat{m}^{\text{rule}}$ obtained from the plug-in rule procedure and one with $\widehat{m}^{\text{cv}}$ obtained from the cross-validation procedure, both described in Section~\ref{subsec:cutoff_selection}. Further, a naive estimator $\widehat{\boldsymbol{\Sigma}}_N$ obtained by plugging-in the standard estimator $\widehat{C}(\un h)$ from \eqref{eq:sampleAutocov}  in its entries, and $\widehat{\boldsymbol{\Sigma}}_{CH}$, which is obtained by plugging-in the Cressie-Hawkins estimator \eqref{eq:cressiehawkins}. All four estimators above undergo the spectral projection technique from \eqref{eq:psd_proj} to secure a positive semidefinite covariance matrix, defined generally with $\widehat{\boldsymbol{\Sigma}}^+$ for the description of the ordinary kriging below.

In each testing patch, we split two disjoint sets: the set $O$ of observed locations and the set $P$ of locations to be predicted, indicated by the red square in Figure~\ref{fig:volcano_field}. The observed values are collected in the vector $\underline{y}_{O}$, while the central block values are collected in vector $\underline{y}_{P}$, and used only for the evaluation of the prediction. We reconfigure each estimator by the index sets $O$ and $P$. Concretely, for all four estimators, $\widehat{\boldsymbol{\Sigma}}_{OO}^+$ represents the submatrix of rows and columns indexed by $O$ and $\widehat{\boldsymbol{\Sigma}}_{PO}^+$ the submatrix of rows indexed by $P$ and columns indexed by $O$.
Ordinary kriging treats the process mean as an unknown constant, estimated by generalized least squares
\begin{equation}
\label{eq:ok_mean}
\widehat{\mu}=(\eones^{\top}(\widehat{\boldsymbol{\Sigma}}_{OO}^+)^{-1}\eones)^{-1}\eones^{\top}(\widehat{\boldsymbol{\Sigma}}_{OO}^+)^{-1}\underline{y}_{O}
\end{equation}
and then used in the prediction of the held-out central block by ordinary kriging
\begin{equation}\label{eq:ok_pred}
\widehat{\underline{y}}_{P}=\widehat{\mu}\eones+\widehat{\boldsymbol{\Sigma}}^+_{PO}(\widehat{\boldsymbol{\Sigma}}_{OO}^+)^{-1}
\bigl(\underline{y}_{O}-\widehat{\mu}\eones\bigr),
\end{equation}
for further details on the ordinary kriging methodology, see for example \citet{sacks1989designs, Cressie1993, sherman2011spatial}.
The prediction accuracy is measured for each estimator $\widehat{\boldsymbol{\Sigma}}^+$ by the root mean squared prediction error (RMSPE)
\begin{equation}
\label{eq:rmspe}
\text{RMSPE}(\widehat{\boldsymbol{\Sigma}}^+) = \sqrt{\frac{1}{|P|}\sum_{t\in P}(\hat{y}_t-y_t)^2}.
\end{equation}
The results are shown in Table~\ref{tab:PRMSE_results}. It appears that, for most patches, the bias-corrected estimator $\widehat{\boldsymbol{\Sigma}}_{\un m}^*$ with cutoff value selected by the plug-in rule obtains the lowest prediction error, while in general both bias-corrected estimators compete well against the naive benchmark.

\begin{table}[H]
\centering
\begin{tabular}{lcccccc}
\toprule
& \multicolumn{2}{c}{full lag grid}
& \multicolumn{2}{c}{cross-validated}
& \multicolumn{2}{c}{plug-in rule}\\
\cmidrule(lr){2-3}\cmidrule(lr){4-5}\cmidrule(lr){6-7}
patch & $\widehat{\boldsymbol{\Sigma}}^+_N$ & $\widehat{\boldsymbol{\Sigma}}^+_{\mathrm{CH}}$
& $\widehat{m}^{\mathrm{cv}}$ & $\widehat{\boldsymbol{\Sigma}}^{*+}_{\un m}$
& $\widehat{m}^{\mathrm{rule}}$ & $\widehat{\boldsymbol{\Sigma}}^{*+}_{\un m}$\\
\midrule
$\mathrm{P}_1$ & \cellcolor{gray!20}0.031 & \cellcolor{gray!20}0.030
& \cellcolor{gray!20}8 & \cellcolor{gray!20}0.020 & 5 & 0.019\\
$\mathrm{P}_2$ & \cellcolor{gray!20}0.044 & \cellcolor{gray!20}0.064
& 8 & 0.025 & \cellcolor{gray!20}4 & \cellcolor{gray!20}0.030\\
$\mathrm{P}_3$ & \cellcolor{gray!20}0.044 & 0.011
& \cellcolor{gray!20}9 & \cellcolor{gray!20}0.028
& \cellcolor{gray!20}5 & \cellcolor{gray!20}0.018\\
$\mathrm{P}_4$ & \cellcolor{gray!20}0.031 & \cellcolor{gray!20}0.045
& \cellcolor{gray!20}5 & \cellcolor{gray!20}0.023 & 4 & 0.022\\
$\mathrm{P}_5$ & \cellcolor{gray!20}0.072 & \cellcolor{gray!20}0.034
& \cellcolor{gray!20}9 & \cellcolor{gray!20}0.031 & 4 & 0.020\\
$\mathrm{P}_6$ & \cellcolor{gray!20}0.125 & \cellcolor{gray!20}0.081
& \cellcolor{gray!20}8 & \cellcolor{gray!20}0.037 & 5 & 0.017\\
\midrule
avg & 0.058 & 0.044
& --- & 0.027
& --- & 0.021\\
\bottomrule
\end{tabular}
\caption{RMSPE results of the ordinary kriging for each estimator considered. $\widehat{\boldsymbol{\Sigma}}_N^+$ standard estimator in its entries; $\widehat{\boldsymbol{\Sigma}}_{\mathrm{CH}}^+$ Cressie-Hawkins estimator in its entries; $\widehat{\boldsymbol{\Sigma}}^{*,+}_{\un m}$ bias-corrected estimator with cutoff selected by either plug-in rule or cross-validation.}
\label{tab:PRMSE_results}
\end{table}
The bias-corrected estimator is implemented in the \texttt{bclc} package that is available on Github \citep{Fioravanti2026bclc}.

\FloatBarrier
\section{Conclusion}\label{sec:conclusion}
In this paper, we developed a framework for bias correction of sample covariance estimators for stationary lattice processes indexed in $\mathbb{Z}^2$. The starting point was the derivation of the exact finite-sample bias for the covariance estimator $\widehat{C}(h_1,h_2)$ at any single spatial lag $(h_1,h_2)$ computed on a rectangular $(n_1\times n_2)$ grid. When considering a subset of spatial sample covariances jointly, we derived how to construct a deterministic matrix-correction factor $\mathbf{W}_{\underline{n}}^*$ that is known and depends only on the grid dimensions. This weight matrix $\mathbf{W}_{\underline{n}}^*$ is shown to be singular, but the reduced weight matrix $\mathbf{W}_{\underline{n}, m}^*$ obtained by extracting the central submatrix is argued to be invertible and the resulting bias-corrected estimator $\underline{\widehat{C}}_m^* = (\mathbf{W}_{\underline{n},m}^*)^{-1}\vect(\mathbf{\widehat{\Gamma}}_m)$ becomes nearly unbiased for jointly estimating the corresponding covariances. We show that the correction applied to the standard estimator with prefactor $1/N_h$ and to the constant prefactor estimator with $1/N$ leads to exactly the same bias-corrected estimator. We also derive the finite sample formulas for the covariances of the estimators and establish asymptotic normality of our bias-corrected estimator, which is asymptotically equivalent to its uncorrected version. Also, we propose two data-driven methods to select the cutoff parameter $m$.

In Monte Carlo simulations, we study the bias-corrected estimator and compare it to the naive one and a benchmark variogram-based estimator proposed by \citet{cressie1980robust}, which also attempts to mitigate the bias introduced by the sample covariance estimation. Our simulation focuses on the study of Gaussian lattice processes with an exponential covariance function with varying degrees of spatial dependence, covering both separable and non-separable structures. The bias-corrected estimator effectively manages to reduce the bias in many setups, especially when the spatial dependence is strong, which may be the practically more important case. In such setups, we also see that the reduction of the bias results also in a reduction of the MSE, when a conservative choice of the parameter $m$ made. Finally, we applied the bias-corrected estimator in a ordinary kriging task, which showed good results in prediction compared to the naive counterpart.

\section*{Acknowledgments} 

This research was funded in the course of TRR 391 Spatio-temporal Statistics for the Transition of Energy and Transport (520388526) by the Deutsche Forschungsgemeinschaft (DFG,
German Research Foundation). In addition, the authors gratefully acknowledge the computing time provided on the Linux HPC cluster at TU Dortmund University (LiDO3), partially funded in the course of the Large-Scale Equipment Initiative by the German Research Foundation (DFG) as project 271512359. 

\section*{Declaration of generative AI and AI-assisted technologies in the writing process}
During the preparation of this work the authors used Claude (Anthropic) and ChatGPT (OpenAI) in order to improve language and readability, check consistency of notation, assist with the review and debugging of the simulation code. After using these tools, the authors reviewed, verified and edited all content as needed and take full responsibility for the content of the published article.

\section*{Declaration of competing interest}
The authors declare that they have no known competing financial interests or personal relationships that could have appeared to influence the work reported in this paper.

\bibliographystyle{elsarticle-harv}
\bibliography{pipbias}

@article{vogelsang2016exactly,
  title = {Exactly/Nearly Unbiased Estimation of Autocovariances of a Univariate Time Series With Unknown Mean},
  volume = {37},
  ISSN = {1467-9892},
  url = {http://dx.doi.org/10.1111/jtsa.12184},
  DOI = {10.1111/jtsa.12184},
  number = {6},
  journal = {Journal of Time Series Analysis},
  publisher = {Wiley},
  author = {Vogelsang,  Timothy J. and Yang,  Jingjing},
  year = {2016},
  month = Feb,
  pages = {723–740}
}

@article{li2024correcting,
  title = {Correcting the bias of the sample cross‐covariance estimator},
  volume = {45},
  ISSN = {1467-9892},
  url = {http://dx.doi.org/10.1111/jtsa.12701},
  DOI = {10.1111/jtsa.12701},
  number = {2},
  journal = {Journal of Time Series Analysis},
  publisher = {Wiley},
  author = {Li,  Yifan},
  year = {2023},
  month = May,
  pages = {214–247}
}

@article{guillaumin2022debiased,
  title = {The Debiased Spatial Whittle Likelihood},
  volume = {84},
  ISSN = {1467-9868},
  url = {http://dx.doi.org/10.1111/rssb.12539},
  DOI = {10.1111/rssb.12539},
  number = {4},
  journal = {Journal of the Royal Statistical Society Series B: Statistical Methodology},
  publisher = {Oxford University Press (OUP)},
  author = {Guillaumin,  Arthur P. and Sykulski,  Adam M. and Olhede,  Sofia C. and Simons,  Frederik J.},
  year = {2022},
  month = July,
  pages = {1526–1557}
}

@article{Bilchouris2025,
  title = {On Nonparametric Estimation of Covariogram},
  volume = {54},
  ISSN = {1026-597X},
  url = {http://dx.doi.org/10.17713/ajs.v54i1.1975},
  DOI = {10.17713/ajs.v54i1.1975},
  number = {1},
  journal = {Austrian Journal of Statistics},
  publisher = {Austrian Statistical Society},
  author = {Bilchouris,  Adam and Olenko,  Andriy},
  year = {2025},
  month = Jan,
  pages = {112–137}
}

@article{cressie1980robust,
  title = {Robust estimation of the variogram: I},
  volume = {12},
  ISSN = {1573-8868},
  url = {http://dx.doi.org/10.1007/BF01035243},
  DOI = {10.1007/bf01035243},
  number = {2},
  journal = {Journal of the International Association for Mathematical Geology},
  publisher = {Springer Science and Business Media LLC},
  author = {Cressie,  Noel and Hawkins,  Douglas M.},
  year = {1980},
  month = Apr,
  pages = {115–125}
}

@article{Bartlett1946,
  title = {On the Theoretical Specification and Sampling Properties of Autocorrelated Time-Series},
  volume = {8},
  ISSN = {1467-9868},
  url = {http://dx.doi.org/10.2307/2983611},
  DOI = {10.2307/2983611},
  number = {1},
  journal = {Journal of the Royal Statistical Society Series B: Statistical Methodology},
  publisher = {Oxford University Press (OUP)},
  author = {Bartlett,  M. S.},
  year = {1946},
  month = Jan,
  pages = {27–41}
}

@article{bartlett1950periodogram,
  title={Periodogram analysis and continuous spectra},
  volume = {37},
  ISSN = {1464-3510},
  url = {http://dx.doi.org/10.1093/biomet/37.1-2.1},
  DOI = {10.1093/biomet/37.1-2.1},
  number = {1-2},
  journal = {Biometrika},
  publisher = {Oxford University Press (OUP)},
  author = {Bartlett,  M. S.},
  year = {1950},
  pages = {1–16}
}

@article{guyon1982parameter,
  title = {Parameter estimation for a stationary process on a <i>d</i>-dimensional lattice},
  volume = {69},
  ISSN = {1464-3510},
  url = {http://dx.doi.org/10.1093/biomet/69.1.95},
  DOI = {10.1093/biomet/69.1.95},
  number = {1},
  journal = {Biometrika},
  publisher = {Oxford University Press (OUP)},
  author = {Guyon,  Xavier},
  year = {1982},
  pages = {95–105}
}

@article{dahlhaus1987edge,
  title={Edge effects and efficient parameter estimation for stationary random fields},
  author={Dahlhaus, Rainer and K{\"u}nsch, Hansruedi},
  journal={Biometrika},
  volume={74},
  number={4},
  pages={877--882},
  year={1987},
  publisher={Oxford University Press},
    ISSN = {1464-3510},
    url = {http://dx.doi.org/10.1093/biomet/74.4.877},
  DOI = {10.1093/biomet/74.4.877}
}

@article{ripley1984spatial,
 ISSN = {03067734, 17515823},
 URL = {http://www.jstor.org/stable/1403097},
 author = {B. D. Ripley},
 journal = {International Statistical Review / Revue Internationale de Statistique},
 number = {2},
 pages = {141--150},
 publisher = {[Wiley, International Statistical Institute (ISI)]},
 title = {Spatial Statistics: Developments 1980-3, Correspondent Paper},
 urldate = {2026-10-02},
 volume = {52},
 year = {1984}
}

@article{Isserlis1918,
  title = {On a Formula for the Product-Moment Coefficient of any Order of a Normal Frequency Distribution in any Number of Variables},
  volume = {12},
  ISSN = {0006-3444},
  url = {http://dx.doi.org/10.2307/2331932},
  DOI = {10.2307/2331932},
  number = {1/2},
  journal = {Biometrika},
  publisher = {JSTOR},
  author = {Isserlis,  L.},
  year = {1918},
  month = Nov,
  pages = {134}
}

@article{ElMachkouri2013,
  title = {A central limit theorem for stationary random fields},
  volume = {123},
  ISSN = {0304-4149},
  url = {http://dx.doi.org/10.1016/j.spa.2012.08.014},
  DOI = {10.1016/j.spa.2012.08.014},
  number = {1},
  journal = {Stochastic Processes and their Applications},
  publisher = {Elsevier BV},
  author = {El Machkouri,  Mohamed and Volný,  Dalibor and Wu,  Wei Biao},
  year = {2013},
  month = Jan,
  pages = {1–14}
}

@article{hall1994nonparametric,
  ISSN = {00905364, 21688966},
 URL = {http://www.jstor.org/stable/2242501},
 author = {Peter Hall and Nicholas I. Fisher and Branka Hoffmann},
 journal = {The Annals of Statistics},
 number = {4},
 pages = {2115--2134},
 publisher = {Institute of Mathematical Statistics},
 title = {On the Nonparametric Estimation of Covariance Functions},
 urldate = {2026-10-02},
 volume = {22},
 year = {1994}
}

@article{wang2023nonparametric,
author = {Yiming Wang and Sujit K. Ghosh},
title = {Nonparametric estimation of isotropic covariance function},
journal = {Journal of Nonparametric Statistics},
volume = {35},
number = {1},
pages = {198--237},
year = {2023},
publisher = {Taylor \& Francis},
doi = {10.1080/10485252.2022.2146111},
URL = {https://doi.org/10.1080/10485252.2022.2146111}
}

@article{Bjrnstad2001,
  title = {Nonparametric spatial covariance functions: Estimation and testing},
  volume = {8},
  ISSN = {1573-3009},
  url = {http://dx.doi.org/10.1023/A:1009601932481},
  DOI = {10.1023/a:1009601932481},
  number = {1},
  journal = {Environmental and Ecological Statistics},
  publisher = {Springer Science and Business Media LLC},
  author = {BjØrnstad,  Ottar N. and Falck,  Wilhelm},
  year = {2001},
  month = Mar,
  pages = {53–70}
}

@article{jentsch2015covariance,
  title = {Covariance matrix estimation and linear process bootstrap for multivariate time series of possibly increasing dimension},
  volume = {43},
  ISSN = {0090-5364},
  url = {http://dx.doi.org/10.1214/14-AOS1301},
  DOI = {10.1214/14-aos1301},
  number = {3},
  journal = {The Annals of Statistics},
  publisher = {Institute of Mathematical Statistics},
  author = {Jentsch,  Carsten and Politis,  Dimitris N.},
  year = {2015},
  month = June 
}

@article{mardia1984,
title = {Maximum likelihood estimation of models for residual covariance in spatial regression},
  volume = {71},
  ISSN = {1464-3510},
  url = {http://dx.doi.org/10.1093/biomet/71.1.135},
  DOI = {10.1093/biomet/71.1.135},
  number = {1},
  journal = {Biometrika},
  publisher = {Oxford University Press (OUP)},
  author = {Mardia,  K. V. and Marshall,  R. J.},
  year = {1984},
  pages = {135–146}
}

@article{Cai2010,
  title = {Optimal rates of convergence for covariance matrix estimation},
  volume = {38},
  ISSN = {0090-5364},
  url = {https://www.jstor.org/stable/20744485},
  DOI = {10.1214/09-aos752},
  number = {4},
  journal = {The Annals of Statistics},
  publisher = {Institute of Mathematical Statistics},
  author = {Cai,  T. Tony and Zhang,  Cun-Hui and Zhou,  Harrison H.},
  year = {2010},
  month = Aug 
}

@article{Bickel2008,
  title = {Regularized estimation of large covariance matrices},
  volume = {36},
  ISSN = {0090-5364},
  url = {http://dx.doi.org/10.1214/009053607000000758},
  DOI = {10.1214/009053607000000758},
  number = {1},
  journal = {The Annals of Statistics},
  publisher = {Institute of Mathematical Statistics},
  author = {Bickel,  Peter J. and Levina,  Elizaveta},
  year = {2008},
  month = Feb 
}

@article{Bickel2008thresh,
  title = {Covariance regularization by thresholding},
  volume = {36},
  ISSN = {0090-5364},
  url = {http://dx.doi.org/10.1214/08-AOS600},
  DOI = {10.1214/08-aos600},
  number = {6},
  journal = {The Annals of Statistics},
  publisher = {Institute of Mathematical Statistics},
  author = {Bickel,  Peter J. and Levina,  Elizaveta},
  year = {2008},
  month = Dec 
}

@article{Furrer2006,
  title = {Covariance Tapering for Interpolation of Large Spatial Datasets},
  volume = {15},
  ISSN = {1537-2715},
  url = {http://dx.doi.org/10.1198/106186006X132178},
  DOI = {10.1198/106186006x132178},
  number = {3},
  journal = {Journal of Computational and Graphical Statistics},
  publisher = {Informa UK Limited},
  author = {Furrer,  Reinhard and Genton,  Marc G and Nychka,  Douglas},
  year = {2006},
  month = Sept,
  pages = {502–523}
}

@article{Ledoit2004,
  title = {Honey,  {I} Shrunk the Sample Covariance Matrix},
  volume = {30},
  ISSN = {2168-8656},
  url = {http://dx.doi.org/10.3905/jpm.2004.110},
  DOI = {10.3905/jpm.2004.110},
  number = {4},
  journal = {The Journal of Portfolio Management},
  publisher = {With Intelligence LLC},
  author = {Ledoit,  Olivier and Wolf,  Michael},
  year = {2004},
  month = July,
  pages = {110–119}
}

@article{Zhang2004,
   title = {Inconsistent Estimation and Asymptotically Equal Interpolations in Model-Based Geostatistics},
  volume = {99},
  ISSN = {1537-274X},
  url = {http://dx.doi.org/10.1198/016214504000000241},
  DOI = {10.1198/016214504000000241},
  number = {465},
  journal = {Journal of the American Statistical Association},
  publisher = {Informa UK Limited},
  author = {Zhang,  Hao},
  year = {2004},
  month = Mar,
  pages = {250–261}
}

@article{Stein1989,
  title = {Some asymptotic properties of kriging when the covariance function is misspecified},
  volume = {21},
  ISSN = {1573-8868},
  url = {http://dx.doi.org/10.1007/BF00893213},
  DOI = {10.1007/bf00893213},
  number = {2},
  journal = {Mathematical Geology},
  publisher = {Springer Science and Business Media LLC},
  author = {Stein,  Michael L. and Handcock,  Mark S.},
  year = {1989},
  month = Feb,
  pages = {171–190}
}

@article{Cressie1985,
  title = {Fitting variogram models by weighted least squares},
  volume = {17},
  ISSN = {1573-8868},
  url = {http://dx.doi.org/10.1007/BF01032109},
  DOI = {10.1007/bf01032109},
  number = {5},
  journal = {Journal of the International Association for Mathematical Geology},
  publisher = {Springer Science and Business Media LLC},
  author = {Cressie,  Noel},
  year = {1985},
  month = July,
  pages = {563–586}
}

@article{Guttorp2006,
  title = {Studies in the history of probability and statistics XLIX On the Matérn correlation family},
  volume = {93},
  ISSN = {0006-3444},
  url = {http://dx.doi.org/10.1093/biomet/93.4.989},
  DOI = {10.1093/biomet/93.4.989},
  number = {4},
  journal = {Biometrika},
  publisher = {Oxford University Press (OUP)},
  author = {Guttorp,  Peter and Gneiting,  Tilmann},
  year = {2006},
  month = Dec,
  pages = {989–995}
}

@book{Cressie1993,
  title = {Statistics for Spatial Data},
  ISBN = {9781119115151},
  ISSN = {1940-6347},
  url = {http://dx.doi.org/10.1002/9781119115151},
  DOI = {10.1002/9781119115151},
  journal = {Wiley Series in Probability and Statistics},
  publisher = {Wiley},
  author = {Cressie,  Noel A. C.},
  year = {1993},
  month = Sept 
}

@book{schabenberger2017statistical,
  title = {Statistical Methods for Spatial Data Analysis: Texts in Statistical Science},
  ISBN = {9781315275086},
  url = {http://dx.doi.org/10.1201/9781315275086},
  DOI = {10.1201/9781315275086},
  publisher = {Chapman and Hall/CRC},
  author = {Schabenberger,  Oliver and Gotway,  Carol A.},
  year = {2017},
  month = Jan 
}

@article{sacks1989designs,
  title = {Designs for Computer Experiments},
  volume = {31},
  ISSN = {1537-2723},
  url = {http://dx.doi.org/10.1080/00401706.1989.10488474},
  DOI = {10.1080/00401706.1989.10488474},
  number = {1},
  journal = {Technometrics},
  publisher = {Informa UK Limited},
  author = {Sacks,  Jerome and Schiller,  Susannah B. and Welch,  William J.},
  year = {1989},
  month = Feb,
  pages = {41–47}
}

@book{sherman2011spatial,
  title = {Spatial Statistics and Spatio‐Temporal Data: Covariance Functions and Directional Properties},
  ISBN = {9780470974391},
  ISSN = {1940-6347},
  url = {http://dx.doi.org/10.1002/9780470974391},
  DOI = {10.1002/9780470974391},
  journal = {Wiley Series in Probability and Statistics},
  publisher = {Wiley},
  author = {Sherman,  Michael},
  year = {2010},
  month = Nov 
}

@book{matheron1962traite,
  title={Trait{\'e} de g{\'e}ostatistique appliqu{\'e}e},
  author={Matheron, Georges},
  volume={14},
  year={1962},
  publisher={Editions Technip}
}

@book{golub2013matrix,
  author    = {Golub, Gene H. and {Van Loan}, Charles F.},
  title     = {Matrix Computations},
  edition   = {4},
  series    = {Johns Hopkins Studies in the Mathematical Sciences},
  publisher = {Johns Hopkins University Press},
  address   = {Baltimore, MD},
  year      = {2013},
  isbn      = {978-1-4214-0794-4},
  doi       = {10.56021/9781421407944},
  url       = {https://doi.org/10.56021/9781421407944}
}

@book{brockwell1991,
  title = {Time Series: Theory and Methods},
  ISBN = {9781441903204},
  ISSN = {0172-7397},
  url = {http://dx.doi.org/10.1007/978-1-4419-0320-4},
  DOI = {10.1007/978-1-4419-0320-4},
  journal = {Springer Series in Statistics},
  publisher = {Springer New York},
  author = {Brockwell,  Peter J. and Davis,  Richard A.},
  year = {1991}
}

@book{kreiss2006einfuhrung,
  title={Einf{\"u}hrung in die Zeitreihenanalyse},
  author={Krei{\ss}, Jens-Peter and Neuhaus, Georg},
  year={2006},
  publisher={Springer-Verlag},
  url = {http://dx.doi.org/10.1007/3-540-33571-4},
  DOI = {10.1007/3-540-33571-4},
}

@book{Politis1999,
  title = {Subsampling},
  ISBN = {9781461215547},
  ISSN = {0172-7397},
  url = {http://dx.doi.org/10.1007/978-1-4612-1554-7},
  DOI = {10.1007/978-1-4612-1554-7},
  journal = {Springer Series in Statistics},
  publisher = {Springer New York},
  author = {Politis,  Dimitris N. and Romano,  Joseph P. and Wolf,  Michael},
  year = {1999}
}

@book{Bradley2007intro,
  author    = {Richard C. Bradley},
  title     = {Introduction to Strong Mixing Conditions},
  year      = {2007},
  publisher = {Kendrick Press},
  address   = {Heber City, UT},
  note      = {Volumes 1--3},
}

@article{Bolthausen1982,
  title = {On the Central Limit Theorem for Stationary Mixing Random Fields},
  volume = {10},
  ISSN = {0091-1798},
  url = {http://dx.doi.org/10.1214/aop/1176993726},
  DOI = {10.1214/aop/1176993726},
  number = {4},
  journal = {The Annals of Probability},
  publisher = {Institute of Mathematical Statistics},
  author = {Bolthausen,  E.},
  year = {1982},
  month = Nov 
}

@book{DedeckerDoukhan2007,
  author    = {J\'{e}r\^{o}me Dedecker and
               Paul Doukhan and
               Gabriel Lang and
               Sana Louhichi and
               Cl{\'e}mentine Prieur},
  title     = {Weak Dependence: With Examples and Applications},
  series    = {Lecture Notes in Statistics},
  volume    = {190},
  publisher = {Springer},
  address   = {New York},
  year      = {2007},
  isbn      = {978-0-387-69951-6},
    url = {http://dx.doi.org/10.1007/978-0-387-69952-3},
  DOI = {10.1007/978-0-387-69952-3}
}

@Manual{dataset,
    title = {R: A Language and Environment for Statistical Computing},
    author = {{R Core Team}},
    organization = {R Foundation for Statistical Computing},
    address = {Vienna, Austria},
    year = {2024},
    url = {https://www.R-project.org/},
  }

@software{Fioravanti2026bclc,
  author  = {Fioravanti, Antonio},
  title   = {bclc: Bias-Corrected Lattice Covariance Estimation},
  year    = {2026},
  version = {0.5.0},
  url     = {https://github.com/antfioravanti/bclc},
  note    = {{R} code}
}

\newpage
\appendix

\section{Proof of Theorem~\ref{prop:single_bias}}
We adopt the following set notation for arbitrary integers $a,b\in\mathbb{Z}$ with $a\leq b$, we write
\begin{align}\label{eq:I_set_definition_proof}
I(a:b):= \{a,\, a+1,\, a+2,\, \ldots,\,b-2,\,b-1,\,b\}
\end{align}
for the set of all integers from $a$ to $b$. We let $I(a:b)= \emptyset $ when $a>b$. Hence, for all $a,b\in\mathbb{Z}$, the cardinality $|I(a:b)|$ of $I(a:b)$ becomes
\begin{equation*}\label{eq:intersetion_set_def}
    \bigl|I(a:b)\bigr| = (b-a+1)^+, \quad \textup{with} \quad x^+:=\max(x,0), \quad \forall x\in \mathbb Z.
\end{equation*}
For integers $a_j\leq b_j$, $j = 1, \dots J$  and for $x\in \mathbb{Z}$, we have that
\begin{align}\label{eq:intersection_id}
    x \in \bigcap_{j = 1}^J I(a_j:b_j)
    \iff \max_{1\leq j\leq J}a_j\leq x\leq \min_{1\leq j\leq J}b_j
    \iff x \in I\Bigl(\max_{j} a_j: \min_{j} b_j\Bigr) .
\end{align}
Then, the cardinality of the intersection in \eqref{eq:intersection_id} is given by
\begin{equation}    \label{eq:cardinality_intervals_joint}
\Bigl| \bigcap_{j = 1}^J I(a_j:b_j)\Bigr|
=
\begin{cases}
\min_{j} b_j-\max_{j} a_j+1, & \max_{j} a_j\leq \min_{j} b_j,\\[2pt]
0, & \max_{j} a_j > \min_{j} b_j,
\end{cases}
= \Bigl(\min_{j} b_j-\max_{j} a_j+1\Bigr)^+ .
\end{equation}

\begin{lemma}[Window count]\label{lemma:count_nu} Let $n \geq 1$, let $J\ge1$ and let $d_1,\dots,d_J\in\mathbb{Z}$ be arbitrary shifts. Then, we have
\begin{equation}
\label{eq:nu_count_lemma}
\nu_n\{d_1, \dots, d_J\}:=\Big|\bigcap_{j =1}^J\bigl\{t\in\mathbb{Z}:\ t+d_j\in I(1:n)\bigr\}\Big|
=\Bigl(n-\max_{j,j'\le J}\bigl|d_j-d_{j'}\bigr|\Bigr)^{+}.
\end{equation}
\end{lemma}

\begin{proof}
    For the fixed $j$-th constraint given by the $j$th shift $d_j$, we have that 
    \[
t+d_j \in I(1:n)\iff1\leq t+d_j \leq n \iff 1-d_j\leq t \leq n-d_j \iff t \in I(1-d_j: n-d_j)
    \]
so that the $j$-th constraint set, is itself a set of the form \eqref{eq:I_set_definition_proof}. The admissible $t$ are therefore the elements that satisfy all $J$ constraints simultaneously, thus we have an intersection of sets such as in \eqref{eq:intersection_id} with $a_j = 1-d_j$ and $b_j = n-d_j$,

\begin{align}\label{eq:window_is_interval}
\bigcap_{j=1}^{J}\bigl\{t\in\mathbb{Z}: t+d_j\in I(1:n)\bigr\}
&=\bigcap_{j=1}^{J}I(1-d_j:n-d_j)=I\Bigl(\max_{j\leq J}(1-d_j): \min_{j\leq J}(n-d_j)\Bigr)\notag\\
&=I\Bigl(1-\min_{j\leq J}d_j: n-\max_{j\leq J}d_j\Bigr).
\end{align}

Then, by \eqref{eq:cardinality_intervals_joint} the cardinality of \eqref{eq:window_is_interval} is
\begin{equation}
\label{eq:window_cardinality}
\begin{aligned}
\Bigl|\bigcap_{j=1}^{J}\bigl\{t\in\mathbb{Z}:t+d_j\in I(1:n)\bigr\}\Bigr|&=\Bigl(\bigl(n-\max_{j\leq J}d_j\bigr)-\bigl(1-\min_{j\leq J}d_j\bigr)+1\Bigr)^{+}\\
&=\Bigl(n-\bigl(\max_{j\leq J}d_j-\min_{j\leq J}d_j\bigr)\Bigr)^{+}= \Bigl(n-\max_{j,j'\leq J}\bigl|d_j-d_{j'}\bigr|\Bigr)^+
\end{aligned}
\end{equation}
\end{proof}

\begin{remark}[Window count on a $2$-dimensional lattice]\label{remark:2dim_count}

For shifts $\un d_1,\dots,\un d_J\in\mathbb{Z}^2$, with $\un d_j=(d_{j,1},d_{j,2})$, let
\begin{equation}
\label{eq:nu_lattice}
\nu_n\bigl\{\un d_1,\dots,\un d_J\bigr\}
:=\Bigl|\bigcap_{j=1}^{J}\bigl\{\tu\in\mathbb{Z}^2:\ \tu+\un d_j\in\nset\bigr\}\Bigr| .
\end{equation}
where $G_{\un n}= I(1:n_1)\times I(1:n_2)$. Since the membership of $\tu+\un d_j\in\nset$ means that we impose one constraint per axis, and the constraints of the two
axes involve disjoint coordinates $\{1, \dots, n_1\}\times \{1, \dots, n_2\}$. The intersection in \eqref{eq:nu_lattice} therefore
factorizes into a Cartesian product, 
\begin{equation}
\label{eq:nu_factorises}
\nu_n\bigl\{\un d_1,\dots,\un d_J\bigr\}
=\prod_{i=1}^{2}\nu_{n_i}\bigl\{d_{1,i},\dots,d_{J,i}\bigr\}
=\prod_{i=1}^{2}\Bigl(n_i-\max_{j,j'\leq J}\bigl|d_{j,i}-d_{j',i}\bigr|\Bigr)^{+} .
\end{equation}
\end{remark}

    

\begin{proof}[Proof of Theorem \ref{prop:single_bias}]
From Assumption~\ref{ass:stationarity}, $\{X_{\un t}:  \tu \in \mathbb{Z}^2\}$ is a second-order stationary random field
with mean $\mu=\E[X_{\un t}]$ and covariance function $C(\un h)=\Cov(X_{\un t+ \un h},\,X_{\un t})$, which only depends on the lag/shift and not on the coordinates on the $(n_1\times n_2)$ grid $G_{\un n}$, where
\begin{equation*}
    G_{\un n} = \{\un t\in \mathbb{Z}^2: 1\leq t_i\leq n_1, \; i=\{1,2\} \} = I(1:n_1)\times I(1:n_2), \quad  \text{with}  |G_{\un n}| = N = n_1n_2, 
\end{equation*}
while the set of all admissible lags is $H_{\un n}$ defined by
\begin{equation*}
    H_{\un n} = \{\un h\in \mathbb{Z}^2: |h_i|\leq n_i-1, \; i=\{1,2\}\}.
\end{equation*}
The subset $G_{\un n}(\un h)$ represents the set of coordinates $\un t$ that admit the shift by a lag $\un h \in H_{\un n}$, that is,
\begin{equation*}
\nhset:=\bigl\{\un t\in \nset : \un t+\un h\in \nset\bigr\}\subseteq \nset,
 \quad |\nhset|= N_h= \prod_{i=1}^{2}(n_i-|h_i|).
\end{equation*}
More generally, for an arbitrary shift $\un k$, we write $G_{\un n}(\un k) := \{\un t\in G_{\un n}: \un t + \un k\in G_{\un n}\}$, with the case that $G_{\un n}(\un 0)= G_{\un n}$ and by Remark~\ref{remark:2dim_count}, $|G_{\un n}(\un k)|=\nu_{\un n}\{\un 0, \un k\}$.

The standard sample covariance estimator at lag $\underline{h}\in H_{\un n}$ by
\begin{equation*}
\widehat{C}(\underline{h}) = \frac{1}{N_h}\sum_{\tu \in \nhset}(X_{\tu +\hu}- \Xbar)(X_{\tu}- \Xbar)
\label{eq:prop_displayChat}
\end{equation*}
with the sample mean defined 
\begin{equation*}
\Xbar=\frac{1}{N}\sum_{\uu \in \nset}X_{\un u}=\frac{1}{N}\sum_{u_1=1}^{n_1}\sum_{u_2=1}^{n_2}X_{u_1,u_2}
\end{equation*}
Now we expand the product and take expectations
\begin{eqnarray}
\E\bigl[\widehat{C}(\underline{h})\bigr]
&=& 
  \underbrace{\frac{1}{N_h}
\sum_{\tu \in \nhset}\E\bigr[X_{\tu +\hu}X_{\tu}\bigl]
  }_{\text{=:C}}
- \underbrace{\frac{1}{NN_h}\sum_{\un t \in \nhset}\sum_{\uu \in \nset}
\E\bigr[X_{\tu +\hu}X_{\uu}\bigl]
  }_{\text{=:A}}\notag\\
&&- \underbrace{
\frac{1}{N_hN}
\sum_{\tu \in \nhset}\sum_{\uu \in \nset}\E\bigr[X_{\tu}X_{\uu}\bigl]
  }_{\text{=:B}}
  + \underbrace{\frac{1}{N^2}
\sum_{\uu, \vu \in  \nset}\E\bigr[X_{\uu}X_{\vu}\bigl]
  }_{\text{=:D}}.
\label{eq:ABCD}
\end{eqnarray}
Term C is a double summation and is immediate: since the summand $\E[X_{\un t+\un h}\,X_{\un t}]$ is constant in $\un t$, there are exactly $N_h$ terms, which cancel with the prefactor $1/N_h$ of the summation and we obtain
\begin{equation*}
\text{C} = C(\un h) + \mu^2.
\label{eq:termI_result}
\end{equation*}
The inner sum over $\uu$ is $N\Xbar$, so
\begin{equation}\label{eq:ABD_as_Xbar}
\text{A}=\frac{1}{N_h}\sum_{\tu\in\nhset}\E\bigl[X_{\tu+\hu}\Xbar\bigr],\quad
\text{B}=\frac{1}{N_h}\sum_{\tu\in\nhset}\E\bigl[X_{\tu}\Xbar\bigr],\quad
\text{D}=\E\bigl[\Xbar^{\,2}\bigr].
\end{equation}
Using $\E[X_{\pu}\Xbar]=\Cov(X_{\pu},\Xbar)+\mu^{2}$ and $\E[\Xbar^{2}]=\Var(\Xbar)+\mu^{2}$, we have that each of the four terms in \eqref{eq:ABCD}
has exactly one $\mu^{2}$ such that they cancel. Hence, $\E[\widehat{C}(\underline{h})]$ does not depend on $\mu$.
Then, by subtracting $C(\hu)$, term C cancels and we obtain 
\begin{align}
b(\hu)=\E\bigl[\widehat C(\hu)\bigr]-C(\hu)
= -\frac{1}{N_h}\sum_{\tu \in \nhset}
   \bigl[\Cov(X_{\tu+\hu}, \Xbar)+ \Cov(X_{\tu}, \Xbar)\bigr]+\Var(\Xbar),
\label{eq:bias_cov_form1}
\end{align}
which is the bias representation \eqref{eq:biasautocov_explicit_mainprop}.

To get also the reduced form, note that terms A, B and D are fourfold summations. We want to exploit the stationarity property to reduce them to a double summation over a new lag index $\un r \in H_{\un n}$. However, in practice, we need to count how many times the new lag $\un r$ appears in the new double summation. We proceed term by term.

\paragraph{Term A} By stationarity, 
\begin{equation}\label{eq:A_opened}
\textup{A}-\mu^2=\frac{1}{NN_h}\sum_{\un t\in\nhset}\sum_{\uu\in\nset}\E\bigl[X_{\tu+\hu}X_{\uu}\bigr]=\frac{1}{NN_h}\sum_{\tu\in\nhset}\sum_{\uu\in\nset}C\bigl(\uu-\tu-\hu\bigr).
\end{equation}
We apply a change of argument by fixing $\un t$ and setting
 \begin{equation*}
\ru:=\uu-\tu-\hu
\quad\Longleftrightarrow\quad
\uu=\tu+\hu+\ru
\quad\Longrightarrow\quad
\bigl(\uu\in\nset\iff \tu+\hu+\ru\in\nset\bigr).
\end{equation*}
The site $\tu+\hu+\ru$ is $\uu$ rewritten, namely $\uu\mapsto\ru$ is a bijection of $\nset$
onto $\{\ru:\tu+\hu+\ru\in\nset\}$, so the inner sum of \eqref{eq:A_opened} is
\begin{equation*}\label{eq:A_inner}
\sum_{\uu\in\nset}C\bigl(\uu-\tu-\hu\bigr)
=\sum_{\ru\in\mathbb{Z}^{2}}C(\ru)\mathbbm{1}\{\tu+\hu+\ru\in\nset\}.
\end{equation*}
Exchanging the two summations in \eqref{eq:A_opened}, for one with $\un r$ as index,
\begin{equation}\label{eq:A_counted}
\sum_{\tu\in\nhset}\sum_{\uu\in\nset}C\bigl(\uu-\tu-\hu\bigr)
=\sum_{\ru\in\mathbb{Z}^{2}}C(\ru)
\bigl|\bigl\{\tu: \underbrace{\tu\in\nset, \tu+\hu\in\nset}_{\tu\,\in\,\nhset}, \underbrace{\tu+\hu+\ru\in\nset}_{\uu\in\nset}\bigr\}\bigr|.
\end{equation}
The three constraints in $\eqref{eq:A_counted}$ are the shifts $\un 0, \un h, \un h+\un r$ of $\un t$ and we apply the count function from Lemma~\ref{lemma:count_nu}.
By Remark~\ref{remark:2dim_count}, the cardinality factorizes on the axis dimensions and on axis $i$ the shifts $0, h_i, h_i+r_i$ have pairwise absolute distance $|h_i|, |h_i+r_i|, |r_i|$, such that Lemma~\ref{lemma:count_nu} gives

\begin{equation}\label{eq:A_weight}
\nu_{\un n}\{\un 0, \un h, \un h +\un r\}
=\prod_{i=1}^{2}\bigl(n_i-\max(|r_i|,|h_i|,|r_i+h_i|)\bigr)^{+}
=:\prod_{i=1}^{2}\nu_{n_i}(r_i,h_i),
\end{equation}
which vanishes as soon as $|r_i|\ge n_i$; the sum over $\ru$ is therefore a sum over
$\lagset$. Combining \eqref{eq:A_opened}--\eqref{eq:A_weight}
\begin{equation}\label{eq:termA_final}
\frac{1}{N_h}\sum_{\tu\in\nhset}\Cov\bigl(X_{\tu+\hu},\Xbar\bigr)
=\frac{1}{NN_h}\sum_{\ru\in\lagset}C(\ru)\prod_{i=1}^{2}\nu_{n_i}(r_i,h_i)
=\frac{S(\hu)}{NN_h}.
\end{equation}

\paragraph{Term B}
By stationarity,
\begin{equation}\label{eq:B_opened}
\textup{B}-\mu^2=\frac{1}{N_h}\sum_{\tu\in\nhset}\Cov\bigl(X_{\tu},\Xbar\bigr)
=\frac{1}{NN_h}\sum_{\tu\in\nhset}\sum_{\uu\in\nset}C\bigl(\uu-\tu\bigr).
\end{equation}
The term differs from \eqref{eq:A_opened} only in being indexed by $\tu$ instead of $\tu+\hu$. 
Thus, with $\ru = \uu -\tu$, it follows in the same steps as before, that

\begin{equation}\label{eq:B_weight}
\nu_{\un n}\{\un 0,\hu,\ru\}
=\prod_{i=1}^{2}\bigl(n_i-\max(|r_i|,|h_i|,|r_i-h_i|)\bigr)^{+}
=\prod_{i=1}^{2}\nu_{n_i}(r_i,-h_i),
\end{equation}
where the last equality is the definition \eqref{eq:A_weight} of $\nu_{n_i}$ evaluated at
$-h_i$.
The weight again vanishes as soon as $|r_i|\ge n_i$, so the sum over $\ru$ is a sum
over $\lagset$, thus we get
\begin{equation}\label{eq:termB_final}
\frac{1}{N_h}\sum_{\tu\in\nhset}\Cov\bigl(X_{\tu},\Xbar\bigr)
=\frac{1}{NN_h}\sum_{\ru\in\lagset}C(\ru)\prod_{i=1}^{2}\nu_{n_i}(r_i,-h_i)
=\frac{S(-\hu)}{NN_h}.
\end{equation}

\paragraph{Term D}
By
stationarity,
\begin{equation}\label{eq:D_opened}
\textup{D}-\mu^2=\Var(\Xbar)
=\frac{1}{N^{2}}\sum_{\vu\in\nset}\sum_{\uu\in\nset}C\bigl(\uu-\vu\bigr).
\end{equation}
Here both summation sites are free in $\nset$ and there is no lag $\hu$ to respect. In a similar fashion, with $\ru = \uu-\vu$, we obtain

\begin{equation}\label{eq:D_weight}
\nu_{\un n}\{\un 0,\ru\}
=\prod_{i=1}^{2}\bigl(n_i-|r_i|\bigr)^{+}
=\prod_{i=1}^{2}\nu_{n_i}(r_i,0),
\end{equation}
which is $N_{\ru}$ for  $\ru\in\lagset$,
the middle equality being the definition \eqref{eq:A_weight} at $h_i=0$.
Hence, combining
\eqref{eq:D_opened}--\eqref{eq:D_weight}, leads to
\begin{equation}\label{eq:termD_final}
\Var(\Xbar)
=\frac{1}{N^{2}}\sum_{\ru\in\lagset}C(\ru)\prod_{i=1}^{2}\nu_{n_i}(r_i,0)
=\frac{S(\un 0)}{N^{2}} .
\end{equation}

The three weights obtained in \eqref{eq:A_weight}, \eqref{eq:B_weight} and
\eqref{eq:D_weight} are the same counting factor evaluated at $h_i$, at $-h_i$ and at
$0$, so we define the counting factor with the arguments the new index $r_i$ and the lag $h_i$
\begin{align}
\nu_{n_i}(r_i,h_i) := \bigl(n_i-\max(|r_i|,|h_i|,|r_i+h_i|)\bigr)^{+},\quad i\in \{1,2\},
\label{eq:counting_factor_proof}
\end{align}
which is $\eqref{eq:counting_factor}$.
\eqref{eq:termA_final}, \eqref{eq:termB_final} and
\eqref{eq:termD_final} are $S(\hu)$, $S(-\hu)$ and $S(\un 0)$ as defined in
\eqref{eq:S-def}. Inserting them into the bias representation
\eqref{eq:bias_cov_form1} gives
\begin{equation}\label{eq:bias_reduced}
b(\hu)=-\left[\frac{S(\hu)}{NN_h}+\frac{S(-\hu)}{NN_h}\right]+\frac{S(\un 0)}{N^{2}}
=-\frac{S(\hu)+S(-\hu)}{NN_h}+\frac{S(\un 0)}{N^{2}},
\end{equation}
which is \eqref{eq:biasautocov} and completes the proof.
\end{proof}

\begin{proof}[Proof of Corollary~\ref{cor:single_bias_const}]
The proof follows by applying a simple prefactor change $N_h/N$ to connect the expectations of the two estimators and the definition of the bias of the standard estimator in Proposition~\ref{prop:single_bias}, such that 
\begin{equation}
\label{eq:connection_standard_constant}
   \frac{N_h}{N}\E\left(\widehat{C}(\hu)\right) =  \E\left(\widetilde{C}(\hu)\right)=\frac{N_h}{N}\left(C(\hu) +  b(\hu)\right)
\end{equation}
Subtracting $C(\hu)$ on both sides of the equation gives
\eqref{eq:biasautocov_constant}.
\end{proof}

\section{Proof of Theorem \ref{prop:joint_bias}}
Before proving Theorem \ref{prop:joint_bias}, we need the following lemma.

\begin{lemma}[Entrywise bounds on the bias weights]\label{lemma:entrywise}
For all $\un{h},\un{r}\in H_{\un{n}}$,
\begin{equation}\label{eq:entrywise}
\bigl|\bigl[\mathbf{W}^{*}_{\un{n}}-\mathbf{I}_{L_{\un{n}}}\bigr]_{\un{h},\un{r}}\bigr|
\leq\frac{1}{N}.
\end{equation}
The bound is attained at $\un r = \un 0$ where $\bigl[\mathbf{W}_{\un n}^*- \mathbf{I}_{L_{\un{n}}}\bigr]_{\un{h},\un{0}}=-1/N$ for every
$\un{h}\in H_{\un{n}}$, so that
\begin{equation}\label{eq:maxnorm_weight}
\bigl\lVert\mathbf{W}^{*}_{\un{n}}-\mathbf{I}_{L_{\un{n}}}\bigr\rVert_{\max}
=\frac{1}{N}.
\end{equation}
\end{lemma}
\begin{proof}    
Fix $\un{h},\un{r}\in H_{\un{n}}$ so that
$|h_i|,|r_i|\leq n_i-1$ for $i \in \{1,2\}$. In particular, as $n_i-|h_i|\geq 1$, all
denominators below are positive, and $\bigl(n_i-|r_i|\bigr)^{+}=n_i-|r_i|$. For $i\in\{1,2\}$,
we define the entries of the normalized counting matrices as
\begin{equation}\label{eq:entries_labeled_countingmatrices}
a_i:=[\mathbf{A}^{*}_i]_{h_i,r_i}
=\frac{\nu_{n_i}(r_i,h_i)}{n_i(n_i-|h_i|)},\quad
b_i:=[\mathbf{B}^{*}_i]_{h_i,r_i}
=\frac{\nu_{n_i}(r_i,-h_i)}{n_i(n_i-|h_i|)},\quad
d_i:=[\mathbf{D}^{*}_i]_{h_i,r_i}
=\frac{\nu_{n_i}(r_i,0)}{n_i^2},
\end{equation}
where
$\nu_{n_i}(r_i,h_i)=\bigl(n_i-\max(|r_i|,|h_i|,|r_i+h_i|)\bigr)^{+}$
by \eqref{eq:counting_factor}. By the definition of the weight matrix $\mathbf{W}_{\un n}^*$ from \eqref{eq:W_n_star}, the entries of interest are
\begin{equation}\label{eq:entry_formula}
\bigl[\mathbf{W}^{*}_{\un{n}}-\mathbf{I}_{L_{\un{n}}}\bigr]_{\un{h},\un{r}}
=-a_1a_2-b_1b_2+d_1d_2.
\end{equation}
We use two elementary properties of $a_i$, $b_i$, $d_i$, namely
\begin{equation}\label{eq:properties_entries}
0\leq a_i,\,b_i,\,d_i\leq\frac{1}{n_i}
\quad\text{and}\quad
\min(a_i,b_i)\leq d_i
\end{equation}
that are argued in the following. By construction, each element is a normalized count thus it is nonnegative and also $\nu_{n_i}(r_i,\pm h_i)\leq n_i-|h_i|$, which gives $a_i,b_i\leq 1/n_i$. While, $\nu_{n_i}(r_i,0)=n_i-|r_i|\leq n_i$ gives $d_i\leq 1/n_i$. This
shows the first part of \eqref{eq:properties_entries}.

For the second part note that 
for any $r, h \in \mathbb Z$,
\begin{equation}\label{eq:shiftmax}
    \max\bigl(|r+h|, |r-h|\bigr)= |r|+|h|\implies \max\bigl(|r|, |h|, |r+h|, |r-h|\bigr) = |r|+|h|.
\end{equation}
Since $\min(x^+,y^+)=(\min(x,y))^+$, the minimum between the two counting factors is 
\begin{align*}
\min\bigl(\nu_{n_i}(r_i,h_i),\nu_{n_i}(r_i,-h_i)\bigr)
&=\min\Bigl(
\bigl(n_i-\max\{|r_i|,|h_i|,|r_i+h_i|\}\bigr)^{+},
\bigl(n_i-\max\{|r_i|,|h_i|,|r_i-h_i|\}\bigr)^{+}
\Bigr)\\
&=\bigl(n_i-\max\{|r_i|,|h_i|,|r_i+h_i|,|r_i-h_i|\}\bigr)^{+}\\
&=\bigl(n_i-|r_i|-|h_i|\bigr)^{+}.
\end{align*}
using \eqref{eq:shiftmax} in the last step.
We want to show that 
\begin{equation}\label{eq:min_ab_d}
\min(a_i,b_i)
=\frac{(n_i-|r_i|-|h_i|)^{+}}{n_i(n_i-|h_i|)}
\le\frac{n_i-|r_i|}{n_i^2}
=d_i.
\end{equation}
If $n_i-|r_i|-|h_i|<0$ the left-hand side vanishes and the inequality is trivial.
Otherwise, multiplying both sides by $n_i^{2}(n_i-|h_i|)>0$, we have $n_i\bigl(n_i-|r_i|-|h_i|\bigr)\leq(n_i-|r_i|)(n_i-|h_i|)$, and taking their difference
\begin{align*}
(n_i-|r_i|)(n_i-|h_i|)-n_i\bigl(n_i-|r_i|-|h_i|\bigr)
= |r_i||h_i| \;\ge\;0 ,
\end{align*}
which proves \eqref{eq:min_ab_d} and thus the second part of \eqref{eq:properties_entries}.

\noindent\emph{Upper Bound.} As $a_1a_2,\;b_1b_2\ge0$, we obtain the upper bound 
\begin{equation}\label{eq:upper_bound}
\bigl[\mathbf{W}^{*}_{\un{n}}-\mathbf{I}_{L_{\un{n}}}\bigr]_{\un{h},\un{r}}
\le d_1d_2
\le\frac1{n_1n_2}
=\frac1N.
\end{equation}

\noindent\emph{Lower Bound.} 
From \eqref{eq:entry_formula}, the lower bound $-a_1a_2-b_1b_2+d_1d_2\geq-\frac1N $ is equivalent to
\begin{equation}\label{eq:lower_bound_target}
a_1a_2+b_1b_2\leq \frac1N+d_1d_2,
\end{equation}
By the second part of \eqref{eq:properties_entries}, for each $i$ at least
one of $a_i,b_i$ is bounded by $d_i$. Since \eqref{eq:lower_bound_target} is
invariant under interchanging $(a_1,a_2)$ and $(b_1,b_2)$ this leaves the following two cases.

\begin{enumerate}[label=(\roman*)]
\item Suppose that $a_1\leq d_1$ and $a_2\leq d_2$.
Then $a_1a_2\leq d_1d_2$,
whereas $b_i\leq 1/n_i$, hence $b_1b_2\leq 1/(n_1n_2)=1/N$,
which gives \eqref{eq:lower_bound_target}
The case $b_1\leq d_1$ and $b_2\le d_2$ follows in the same way.

\item Suppose that $a_1\le d_1$ and
$b_2\le d_2$. Since $a_2\le1/n_2$ and $b_1\le1/n_1$, respectively,
\[
a_1a_2\le d_1a_2\leq\frac{d_1}{n_2},
\quad
b_1b_2\le b_1d_2\leq\frac{d_2}{n_1},
\]
so that it suffices to show 
\[
a_1a_2+b_1b_2\leq \frac{d_1}{n_2}+\frac{d_2}{n_1}\leq d_1d_2+\frac1N.
\]
Since $d_i=(n_i-|r_i|)/n_i^2$ and $N=n_1n_2$, multiplication by the
positive quantity $n_1^2n_2^2$ transforms this inequality into
\[
n_2(n_1-|r_1|)+n_1(n_2-|r_2|)
\le n_1n_2+(n_1-|r_1|)(n_2-|r_2|).
\]
The difference between the right-hand and left-hand sides is
$|r_1||r_2|\ge0$. The case $b_1\le d_1$ and $a_2\le d_2$ is identical.
\end{enumerate}
Hence \eqref{eq:lower_bound_target} holds in both cases, which together with
\eqref{eq:upper_bound} proves \eqref{eq:entrywise}..

\noindent For the attainment, let
$\un{r}=\un{0}$. Then, for $i \in \{1,2\}$,
\[
\nu_{n_i}(0,h_i)=\nu_{n_i}(0,-h_i)=n_i-|h_i|,
\qquad
\nu_{n_i}(0,0)=n_i,
\]
so that $a_i=b_i=d_i=1/n_i$
and \[
\bigl[\mathbf{W}_{\un n}^*-\mathbf{I}_{L_{\un n}}\bigr]_{\un{h},\un{0}}=-\frac{1}{N}-\frac{1}{N}+\frac{1}{N}
=-\frac{1}{N}, \quad \un{h}\in H_{\un{n}}
\]
so the bound~\eqref{eq:entrywise}
is attained. This proves~\eqref{eq:maxnorm_weight}.
\end{proof}

\begin{proof}[Proof of Theorem\ref{prop:joint_bias}]
While the first part of Theorem \ref{prop:joint_bias} are direct consequences of the formulas derived before stating the results, it remains to show the rank deficiency of $\mathbf{W}_{\underline{n}}^*$ and the bound in the following.
For this purpose, we begin by noting a few properties of the counting matrices  \[\mathbf{A}_i,\mathbf{B}_i,\mathbf{D}_i\in\mathbb{R}^{L_i\times L_i},\quad \text{with}\quad L_i=2n_i-1
   \] 
and their normalized counterparts 
   \[
\mathbf{A}_i^*=\mathbf{N}_{h,i}^{-1}\mathbf{A}_i,\qquad
\mathbf{B}_i^*=\mathbf{N}_{h,i}^{-1}\mathbf{B}_i,\qquad
\mathbf{D}_i^*=n_i^{-2}\mathbf{D}_i, \quad \text{for}\; i =1,2
\]
where 
\[
\mathbf{N}_{h,i}:=\mathrm{diag}\!\big(N_{h,1,i},\dots,N_{h,L_i,i}\big),\qquad
N_{h,p,i}=n_i\bigl(n_i-|p-n_i|\bigr).
\]
First, as they are counting matrices, $\mathbf{A}_i, \mathbf{B}_i$ (respectively $\mathbf{A}_i^*, \mathbf{B}_i^*$) have non-negative entries by construction. Additionally, $\mathbf{D}_i$ (respectively $\mathbf{D}_i^*$) has only strictly positive entries.

By construction, the $p$-th row–sum of $\mathbf{A}_i$ and $\mathbf{B}_i$ equals $n_i(n_i-|p-n_i|$, while the $p$-th row sum of $\mathbf{D}_i$ is $n_i^2$. This means that the normalization matrices $\mathbf{A}_i^*,\mathbf{B}_i^*, \mathbf{D}_i^* $ have row sums equal to 1 and thus are row–stochastic. Hence
\[
\mathbf{A}_i^*\eones_i=\eones_i,\qquad \mathbf{B}_i^*\eones_i=\eones_i,\qquad \mathbf{D}_i^*\eones_i=\eones_i,
\]
where $\eones_i=(1,1,\dots, 1)$ is the $L_i \times 1$ vector of ones. We define the weight matrix $\mathbf{W}^*$ as a linear mapping
\[W^*: \mathbb{R}^{L_1\times L_2}\to\mathbb{R}^{L_1\times L_2}, \quad W^*(\mathbf{X})= \mathbf{X} - \mathbf{A}_1^*\mathbf{X}\mathbf{A}_2^{*\top} -\mathbf{B}_1^*\mathbf{X}\mathbf{B}_2^{*\top}+\mathbf{D}_1^*\mathbf{X}\mathbf{D}_2^{*\top}. \]

In particular, for the all-ones matrix $\eones_1\eones_2^\top$ of size $L_1\times L_2$,
\begin{align*}
\mathbf{A}^*\eones_1\eones_2^\top\mathbf{A}^{*\top} &= \bigl(\mathbf{A}^*\eones_1\bigr)\bigl(\mathbf{A}^*\eones_2\bigr)^\top = \eones_1\eones_2^\top\\
\mathbf{B}^*\eones_1\eones_2^\top\mathbf{B}^{*\top} &= \bigl(\mathbf{B}^*\eones_1\bigr)\bigl(\mathbf{B}^*\eones_2\bigr)^\top = \eones_1\eones_2^\top\\
\mathbf{D}^*\eones_1\eones_2^\top\mathbf{D}^{*\top} &= \bigl(\mathbf{D}^*\eones_1\bigr)\bigl(\mathbf{D}^*\eones_2\bigr)^\top = \eones_1\eones_2^\top
\end{align*}
so we have
\[
W^*(\eones_1\eones_2^\top) = \eones_1\eones_2^\top - \eones_1\eones_2^\top- \eones_1\eones_2^\top + \eones_1\eones_2^\top = \mathbf{0}.
\]
Thus, $\textup{span}\bigl(\eones_1\eones_2^\top\bigr)$ is in the nullspace of the linear operator $W^*$ and we have $\textup {dim}(\textup{ker}\bigl(W^*\bigr))\geq 1$.
   Hence, 
\begin{align*}
\text{rank}\bigl(\mathbf{W}^*\bigr) = L_{\un n} - \text{dim}\,\textup{ker}\;(W^*)\implies \text{rank}\bigl(\mathbf{W}^*\bigr) \leq L_{\un n} -1.
\end{align*}
Finally, by 
\[b_{\vect(\widehat{\boldsymbol{\Gamma}})} := \E\big(\vect(\widehat{\boldsymbol{\Gamma}})\big) - \vect(\boldsymbol{\Gamma}) =  \left(\mathbf{W}_{\underline{n}}^*-\mathbf{I}_{L_{\un n}}\right)\vect(\boldsymbol{\Gamma}),\]
and Lemma~\ref{lemma:entrywise}, we apply the Hölder inequality\footnote{Row-wise Hölder inequality $\lVert\mathbf{M}v\rVert_\infty\leq\lVert\mathbf{M}\rVert_{\max}\lVert v\rVert_1$} such that
\[
\bigl\lVert b_{\vect(\widehat{\boldsymbol{\Gamma}})}\bigr\rVert_\infty
=\bigl\lVert\bigl(\mathbf{W}_{\un n}^*-\mathbf{I}_{L_{\un n}}\bigr)\vect(\boldsymbol{\Gamma})\bigr\rVert_\infty
\leq\bigl\lVert\mathbf{W}_{\un n}^*-\mathbf{I}_{L_{\un n}}\bigr\rVert_{\max}\,
\bigl\lVert\vect(\boldsymbol{\Gamma})\bigr\rVert_1
=\frac{1}{N}\bigl\lVert\vect(\boldsymbol{\Gamma})\bigr\rVert_1,
\]
which is the bound \eqref{eq:naive_bias_bound}. Under Assumption~\ref{ass:abs_summability} it is bounded uniformly.

\end{proof}

\begin{proof}[Proof of Corollary~\ref{cor:joint_bias_const}]
The scaling \eqref{eq:connection_standard_constant} 
$N_h/N\E\left(\widehat{C}(\hu)\right) =  \E\left(\widetilde{C}(\hu)\right)$
holds entrywise, i.e., $\E(\widetilde{\boldsymbol{\Gamma}})_{p,q}=N_h/N\E(\widehat{\boldsymbol{\Gamma}})_{p,q}$.
By the definition of the Bartlett window matrices
\begin{equation*}
    \boldsymbol{\Delta}_i := \diag\left(1-\frac{|p-n_i|}{n_i},\; p=1, \dots, L_i\right) = \diag\left(1-\frac{|h_i|}{n_i},\; h_i \in I(1-n_i:n_i-1)\right),\quad i \in \{1,2\},
\end{equation*}
where we have the factorization 
\begin{equation}
\label{eq:barlett_window_matrices_proof}
    \frac{N_h}{N}=\frac{n_1-|h_1|}{n_1}\cdot\frac{n_2-|h_2|}{n_2}  = \left[\boldsymbol{\Delta}_1\right]_{p,p}\left[\boldsymbol{\Delta}_2\right]_{q,q}, \quad (h_1,h_2) = (p-n_1,q-n_2),
\end{equation}
we obtain 
\[\E(\widetilde{\boldsymbol{\Gamma}})=\boldsymbol{\Delta}_1\E(\widehat{\boldsymbol{\Gamma}})\boldsymbol{\Delta}_2.\]
By vectorizing the above and using the identity $\vect(\boldsymbol{\Delta}_1\mathbf{X}\boldsymbol{\Delta}_2)=(\boldsymbol{\Delta}_2\otimes\boldsymbol{\Delta}_1)\vect(\mathbf{X})$ as well as $\E(\vect(\widehat{\boldsymbol{\Gamma}})) = \mathbf{W}_{\underline{n}}^*\vect(\boldsymbol{\Gamma})$, we get \[\E(\vect(\widetilde{\boldsymbol{\Gamma}}))=(\boldsymbol{\Delta}_2\otimes\boldsymbol{\Delta}_1)\mathbf{W}_{\underline{n}}^*\vect(\boldsymbol{\Gamma}),\] 
which is \eqref{eq:W_n_tilde}. 

Finally, since $\boldsymbol{\Delta}_i$ are diagonal with positive entries $(n_i-|h_i|)/n_i>0$, for $|h_i|<n_i$, so is their Kronecker product $\boldsymbol{\Delta}_2\otimes \boldsymbol{\Delta}_1$, and since $\mathbf{W}_{\underline{n}}^*\eones=\underline{\mathbf 0}$ this implies $\widetilde{\mathbf{W}}_{\underline{n}}^*\eones=(\boldsymbol{\Delta}_2\otimes\boldsymbol{\Delta}_1)\mathbf{W}_{\underline{n}}^*\eones=\underline{\mathbf 0}$, so the null space of $\widetilde{\mathbf{W}}_{\underline{n}}^*$ is contained in that of $\mathbf{W}_{\underline{n}}^*$.
    
\end{proof}

\section{Proof of Theorem \ref{prop:covofcov}}

\begin{proof}[Proof of Theorem~\ref{prop:covofcov}]
Each summand of the standard estimator \eqref{eq:sampleAutocov} is a product
of two deviations from the sample mean, so that
\begin{equation}\label{eq:bilinear}
\Cov\bigl(\widehat C(\hu),\widehat C(\gu)\bigr)
=\frac{1}{N_{\hu}N_{\gu}}\sum_{\tu\in\nhset}\ \sum_{\su\in\ngset}
\Cov\Bigl(\bigl(X_{\tu+\hu}-\Xbar\bigr)\bigl(X_{\tu}-\Xbar\bigr),\;
          \bigl(X_{\su+\gu}-\Xbar\bigr)\bigl(X_{\su}-\Xbar\bigr)\Bigr),
\end{equation}
by the bilinearity of the covariance.  Fix $\tu\in\nhset$ and $\su\in\ngset$.
Under Assumption~\ref{ass:gaussianity}, the four centered variables in  \eqref{eq:bilinear} are jointly Gaussian. They have constant mean, because
$\E(X_{\tu})=\E(\Xbar)=\mu$ under Assumption~\ref{ass:stationarity}; and their
pairwise covariances are, by the definition \eqref{eq:C_centered} of the
centered covariance function, the values $\Cctr$ evaluated at the
corresponding pairs of locations.  Isserlis' theorem
\citep{Isserlis1918} applied to their fourth moment therefore gives three
pairings, of which the first cancels against the product of expectations in
the covariance, leaving the two cross pairings
\begin{align}
&\Cov\Bigl(\bigl(X_{\tu+\hu}-\Xbar\bigr)\bigl(X_{\tu}-\Xbar\bigr),\;
          \bigl(X_{\su+\gu}-\Xbar\bigr)\bigl(X_{\su}-\Xbar\bigr)\Bigr)
\notag\\[2pt]
&\qquad=\Cov\bigl(X_{\tu+\hu}-\Xbar,\,X_{\su+\gu}-\Xbar\bigr)\,
        \Cov\bigl(X_{\tu}-\Xbar,\,X_{\su}-\Xbar\bigr)\notag\\
&\qquad\quad+\Cov\bigl(X_{\tu+\hu}-\Xbar,\,X_{\su}-\Xbar\bigr)\,
        \Cov\bigl(X_{\tu}-\Xbar,\,X_{\su+\gu}-\Xbar\bigr)\notag\\[2pt]
&\qquad=\Cctr(\tu+\hu,\su+\gu)\,\Cctr(\tu,\su)
       +\Cctr(\tu+\hu,\su)\,\Cctr(\tu,\su+\gu).
\label{eq:isserlis_step}
\end{align}
Inserting \eqref{eq:isserlis_step} into \eqref{eq:bilinear} proves the centered
form \eqref{eq:cov_of_cov}.
 
For the lag form \eqref{eq:cov_of_cov_lag}, we apply the form given by
\eqref{eq:C_centered}, that is by
$\Cov(X_{\tu},X_{\su})=C(\su-\tu)$,, as in \eqref{eq:covariance}. Then, with the shift
$\du:=\su-\tu$, the four lags are
\begin{alignat*}{3}
(\su+\gu)-(\tu+\hu)&=\du+\gu-\hu,&\qquad
\su-\tu&=\du,\\
\su-(\tu+\hu)&=\du-\hu,&\qquad
(\su+\gu)-\tu&=\du+\gu,
\end{alignat*}
so that
\begin{align}
\Cctr(\tu+\hu,\su+\gu)&=C(\du+\gu-\hu)-\Cbar(\tu+\hu)-\Cbar(\su+\gu)+D,
\label{eq:CZ1}\\
\Cctr(\tu,\su)&=C(\du)-\Cbar(\tu)-\Cbar(\su)+D,
\label{eq:CZ2}\\
\Cctr(\tu+\hu,\su)&=C(\du-\hu)-\Cbar(\tu+\hu)-\Cbar(\su)+D,
\label{eq:CZ3}\\
\Cctr(\tu,\su+\gu)&=C(\du+\gu)-\Cbar(\tu)-\Cbar(\su+\gu)+D.
\label{eq:CZ4}
\end{align}
Substituting \eqref{eq:CZ1}--\eqref{eq:CZ4} into \eqref{eq:cov_of_cov} yields
\eqref{eq:cov_of_cov_lag}.  Finally, neither $C$, nor $\Cbar$, nor $D$ depends
on $\mu$, so neither does the right-hand side of
\eqref{eq:cov_of_cov}--\eqref{eq:cov_of_cov_lag}.
\end{proof}

\begin{proof}[Proof of Corollary~\ref{cor:cov_const}]
By \eqref{eq:connection_standard_constant}, $\widetilde C(\hu)=N_{\hu}/N\,\widehat C(\hu)$
with $N_{\hu}/N$ deterministic, so the bilinearity of the covariance yields
\[
\Cov\bigl(\widetilde C(\hu),\widetilde C(\gu)\bigr)
=\frac{N_{\hu}N_{\gu}}{N^2}\;
\Cov\bigl(\widehat C(\hu),\widehat C(\gu)\bigr).
\]
Substituting the
representation \eqref{eq:cov_of_cov} of Proposition~\ref{prop:covofcov}, whose
prefactor is $1/(N_{\hu}N_{\gu})$, the factors $N_{\hu}N_{\gu}$ cancel and leave
$1/N^{2}$, which is \eqref{eq:cov_of_cov_const}.
\end{proof}

\section{Proof of Theorem~\ref{prop:bias_of_bce}}
In this Appendix, we prove Theorem~\ref{prop:bias_of_bce} regarding the finite-sample bias of the bias-corrected estimator and its bound. For this purpose, first, we prove in Lemma~\ref{lemma:invertibility_W_nm} the invertibility of the weight matrix $\mathbf{W}_{\un n, \un m}^*$
when the sufficient condition $M_{\un n}/N<1$ is satisfied. 

\begin{proof}[Proof of Lemma~\ref{lemma:invertibility_W_nm}]
    Let $\mathbf{S}_{\un m}$ be the selection matrix as in \eqref{eq:selection_matrices} with $\mathbf{S}_{\un m}\mathbf{S}_{\un m}^\top$ and the reduced weight matrix $\mathbf{S}_{\un m}\mathbf{W}_{\un m, \un n}^{*\top}\mathbf{S}_{\un m}$ as in \eqref{eq:Wnm_star}, hence we have
\begin{equation}\label{eq:truncated_difference}
\mathbf{W}^{*}_{\un{n},\un{m}}-\mathbf{I}_{M_{\un{m}}}
=\mathbf{S}_{\un{m}}\bigl(\mathbf{W}^{*}_{\un{n}}-\mathbf{I}_{L_{\un{n}}}\bigr)
\mathbf{S}_{\un{m}}^{\top},
\end{equation}
so that every entry of $\mathbf{W}^{*}_{\un{n},\un{m}}-\mathbf{I}_{M_{\un{m}}}$ is an entry of $\mathbf{W}_{\un n}^*-\mathbf{I}_{L_{\un n}}$and is
bounded in absolute value by $1/N$ by Lemma~\ref{lemma:entrywise}. Since each row and each column in $\mathbf{W}^{*}_{\un{n},\un{m}}-\mathbf{I}_{M_{\un{m}}}$ has $M_{\un{m}}$ entries, we have the bounds in \eqref{eq:bounds_lemma_W_I}.

Suppose that $M_{\un m}/N<1$. By Lemma 2.3.3 of \citet{golub2013matrix} (4th edition), since
$\lVert \mathbf{I}_{M_{\un{m}}}-\bigl(\mathbf{I}_{M_{\un{m}}}-\mathbf{W}^{*}_{\un{n},\un{m}}\bigr)\rVert_{\infty}=\lVert\mathbf{W}^{*}_{\un{n},\un{m}}-\mathbf{I}_{M_{\un{m}}}\rVert_{\infty}
\leq M_{\un{m}}/N<1$ and $\lVert\cdot\rVert_{\infty}$ is submultiplicative, the
Neumann series applied to
$\mathbf{W}^{*}_{\un{n},\un{m}}
=\mathbf{I}_{M_{\un{m}}}-\bigl(\mathbf{I}_{M_{\un{m}}}-\mathbf{W}^{*}_{\un{n},\un{m}}\bigr)$
shows that $\mathbf{W}^{*}_{\un{n},\un{m}}$ is nonsingular with
\begin{align}
\bigl\lVert(\mathbf{W}^{*}_{\un n,\un m})^{-1}\bigr\rVert_\infty
&\leq\frac{1}{1-\bigl\lVert\mathbf{W}^{*}_{\un n,\un m}-\mathbf{I}_{M_{\un m}}\bigr\rVert_\infty}
\leq\frac{1}{1-(M_{\un m}/N)}=\frac{N}{N-M_{\un m}}\\
\bigl\lVert(\mathbf{W}^{*}_{\un n,\un m})^{-1}-\mathbf{I}_{M_{\un m}}\bigr\rVert_\infty
&\leq\frac{\bigl\lVert\mathbf{W}^{*}_{\un n,\un m}-\mathbf{I}_{M_{\un m}}\bigr\rVert_\infty}
{1-\bigl\lVert\mathbf{W}^{*}_{\un n,\un m}-\mathbf{I}_{M_{\un m}}\bigr\rVert_\infty}
\leq\frac{M_{\un m}/N}{1-(M_{\un m}/N)}= \frac{M_{\un m}}{N-M_{\un m}}.
\end{align}
which completes the proof.
\end{proof}

\begin{proof}[Proof of Theorem~\ref{prop:bias_of_bce}] By applying the expectation to \eqref{eq:bias_corrected_estimator} then,
\begin{align*}
\E\bigl(\widehat{\un C}^{*}_{\un m}\bigr)&=\bigl(\mathbf{W}^{*}_{\un n,\un m}\bigr)^{-1}
   \Bigl(\mathbf{W}^{*}_{\un n,\un m}\,\un C_{\un m}+\mathbf{W}^{*,c}_{\un n,\un m}\,\un C^{c}_{\un m}\Bigr)=\un C_{\un m}+\bigl(\mathbf{W}^{*}_{\un n,\un m}\bigr)^{-1}\mathbf{W}^{*,c}_{\un n,\un m}\,\un C^{c}_{\un m},
\end{align*}
Then subtracting $\widehat{\un C}_{\un m}$ on both sides gives \eqref{eq:bias_corrected_m}.\\
\emph{Bound.} By Lemma~\ref{lemma:invertibility_W_nm} and applying the Hölder inequality we have 
\begin{align*}
\bigl\lVert\underline{b}^{*}_{\un m}\bigr\rVert_\infty
&=\bigl\lVert(\mathbf{W}^{*}_{\un n,\un m})^{-1}
   \mathbf{W}^{*,c}_{\un n,\un m}\underline{C}^{c}_{\un m}\bigr\rVert_\infty
\leq\bigl\lVert(\mathbf{W}^{*}_{\un n,\un m})^{-1}\bigr\rVert_\infty
   \bigl\lVert\mathbf{W}^{*,c}_{\un n,\un m}\underline{C}^{c}_{\un m}\bigr\rVert_\infty\leq\frac{N}{N-M_{\un m}}\bigl\lVert\mathbf{W}^{*,c}_{\un n,\un m}\bigr\rVert_{\max}
\bigl\lVert\underline{C}^{c}_{\un m}\bigr\rVert_1\\
&\leq\Big(\frac{N}{N-M_{\un m}}\Big)\frac{1}{N}
   \sum_{\un h\in W^c_{\un m}}\bigl|C(\un h)\bigr| =\frac{1}{N-M_{\un m }}\bigl\lVert\underline{C}^{c}_{\un m}\bigr\rVert_1
\end{align*}
\end{proof}

\begin{proof}[Proof of Corollary~\ref{cor:coincide}]
Given that the deterministic scaling $\widetilde{C}(h_1,h_2) = (N_h/N)\cdot\widehat{C}(h_1,h_2)$ from \eqref{eq:connection_standard_constant} holds entrywise in $\vect(\widetilde{\boldsymbol{\Gamma}})$, as we have that $\vect(\widetilde{\boldsymbol{\Gamma}}) = (\boldsymbol{\Delta}_2\otimes \boldsymbol{\Delta}_1)\vect(\widehat{\boldsymbol{\Gamma}})$ and by \eqref{eq:selection_matrices} and \eqref{eq:bartlett_commuting} we have
\[
\vect(\widetilde{\boldsymbol{\Gamma}}_{\underline{m}})
=\mathbf{S}_{\underline{m}}(\boldsymbol{\Delta}_2\otimes\boldsymbol{\Delta}_1)\vect(\widehat{\boldsymbol{\Gamma}})
=(\boldsymbol{\Delta}_{2,m_2}\otimes\boldsymbol{\Delta}_{1,m_1})
\mathbf{S}_{\underline m}\vect(\widehat{\boldsymbol{\Gamma}})=(\boldsymbol{\Delta}_{2,m_2}\otimes\boldsymbol{\Delta}_{1,m_1})\,\vect(\widehat{\boldsymbol{\Gamma}}_{\underline{m}}).
\]
Substituting this together with \eqref{eq:Wnm_tilde} into \eqref{eq:bias_corrected_estimator_const}, and using that the Kronecker product  of the diagonal Bartlett window matrices $\boldsymbol{\Delta}_{2,m_2}\otimes\boldsymbol{\Delta}_{1,m_1}$ is invertible, since its diagonal entries are $(n_i-|h_i|)/n_i>0$ for $|h_i|\leq m_i<n_i$), we have 
\begin{align*}
\widetilde{\underline{C}}_{\underline{m}}^{*}
&=\bigl[(\boldsymbol{\Delta}_{2,m_2}\otimes\boldsymbol{\Delta}_{1,m_1})\,\mathbf{W}_{\underline{n},\underline{m}}^{*}\bigr]^{-1}
(\boldsymbol{\Delta}_{2,m_2}\otimes\boldsymbol{\Delta}_{1,m_1})\,\vect(\widehat{\boldsymbol{\Gamma}}_{\underline{m}})\\
&=\bigl(\mathbf{W}_{\underline{n},\underline{m}}^{*}\bigr)^{-1}(\boldsymbol{\Delta}_{2,m_2}\otimes\boldsymbol{\Delta}_{1,m_1})^{-1}(\boldsymbol{\Delta}_{2,m_2}\otimes\boldsymbol{\Delta}_{1,m_1})\,\vect(\widehat{\boldsymbol{\Gamma}}_{\underline{m}})
=\bigl(\mathbf{W}_{\underline{n},\underline{m}}^{*}\bigr)^{-1}\vect(\widehat{\boldsymbol{\Gamma}}_{\underline{m}})
=\widehat{\underline{C}}_{\underline{m}}^{*}.
\end{align*}
\end{proof}

\section{Proof of Proposition \ref{prop:asymptotics_without_and_with_bias_correction}}

By exploiting a common truncation argument (see, e.g., \citet{brockwell1991}, Proposition 7.3.4), for $\un q := (q_1,q_2)\in\mathbb{N}^2$, we define the truncated process $\mathbf{X}^{(\qu)}=\{X_{\tu}^{(\qu)}, \tu\in \mathbb{Z}^2\}$ by
\begin{align*}
X_{\tu}=\mu+\sum_{j_1=-q_1}^{q_1}\sum_{j_2=-q_2}^{q_2}\psi_{\un j}\epsilon_{\tu - \un j},    \quad   (\tu)\in\mathbb{Z}^2.
\end{align*}
Then, by exploiting the $m$-dependence property (in two-dimensions) of $\mathbf{X}^{(\qu)}$, it is straightforward to show that 
\begin{align*}
\sqrt{N}\left(\widehat{\underline{C}}_m^{(\qu)}-\underline{C}_m^{(\qu)}\right)\overset{d}{\rightarrow}\mathcal{N}\left(\underline{0},V(\widehat{\underline{C}}_m^{(\qu)})\right),
\end{align*}
where the entries of $V(\widehat{\underline{C}}_m^{(\qu)})$ are of the form
\begin{align*}
V(\widehat{\underline{C}}_m^{(\qu)})\big(\hu,\gu\big)
=& (\eta - 3)C^{(\qu)}(\hu)C^{(\qu)}(\gu) + \sum_{\du\in \mathbb Z^2} \Big\{C^{(\qu)}(\du)C^{(\qu)}(\du+\gu-\hu)\Big\}  \\
& + \sum_{\du\in \mathbb Z^2} \Big\{C^{(\qu)}(\du+\gu)C^{(\qu)}(\du-\hu)\Big\},
\end{align*}
where $C^{(\qu)}(\du)$ denotes the corresponding covariances of lag $\du$ of the truncated process $\mathbf{X}^{(\qu)}$. Hence, as $\mathbf{X}^{(\qu)}$ is $m$-dependent (in two-dimensions with $\un m=(m_1,m_2)$ and $m_i=2q_i+1$, in contrast to $C(\du)$ and $V(\widehat{\underline{C}}_m)$, where all sums are infinite, all sums for $C^{(\qu)}(\du)$ and $V(\widehat{\underline{C}}_m^{(\qu)})$ become finite. Further, by letting $\qu\rightarrow \infty$, we get $V(\widehat{\underline{C}}_m^{(\qu)})\rightarrow V(\widehat{\underline{C}}_m)$ and, finally, by using standard arguments (i.e., the absolute summability of the sequence $\{\psi_{\un j},\un j\in\mathbb{Z}^2\}$) and showing that for all $\delta>0$, we have
\begin{align*}
    \lim_{\qu\rightarrow\infty} \limsup_{n_1,n_2\rightarrow \infty} P\left(\left|\sqrt{N}\left(\widehat{\underline{C}}_m-\underline{C}_m\right)-\sqrt{N}\left(\widehat{\underline{C}}_m^{(\qu)}-\underline{C}_m^{(\qu)}\right)\right|>\delta\right)=0,
\end{align*}
the proof of $\sqrt{N}\left(\widehat{\underline{C}}_m-\underline{C}_m\right)\overset{d}{\rightarrow}\mathcal{N}\left(\underline{0},V(\widehat{\underline{C}}_m)\right)$ is complete. Further, for showing $\sqrt{N}\left(\widehat{\underline{C}}_m^*-\underline{C}_m\right)\overset{d}{\rightarrow}\mathcal{N}\left(\underline{0},V(\widehat{\underline{C}}_m)\right)$, using \eqref{eq:bias_corrected_estimator}, we get
\begin{align*}
\sqrt{N}\left(\widehat{\underline{C}}_m^*-\underline{C}_m\right) =& \sqrt{N}\left(\bigl(\mathbf{W}_{\underline{n},m}^{*}\bigr)^{-1}\widehat{\underline{C}}_m-\widehat{\underline{C}}_m+\widehat{\underline{C}}_m-\underline{C}_m\right) \\
=& \sqrt{N}\left(\bigl(\mathbf{W}_{\underline{n},m}^{*}\bigr)^{-1}-I_{M_{\un m}}\right)\widehat{\underline{C}}_m+\sqrt{N}\left(\widehat{\underline{C}}_m-\underline{C}_m\right).
\end{align*}
Hence, in view of the asymptotic normality result derived above
and the consistency of $\widehat{\underline{C}}_m$, that is, $\widehat{\underline{C}}_m\rightarrow \underline{C}_m$, it remains to show
\begin{align*}
\sqrt{N}\left(\bigl(\mathbf{W}_{\underline{n},m}^{*}\bigr)^{-1}-I_{M_{\un m}}\right)=o(1).
\end{align*}
From 
\begin{align*}
I_{M_{\un m}} = \mathbf{W}_{\underline{n},m}^{*}\bigl(\mathbf{W}_{\underline{n},m}^{*}\bigr)^{-1}    
\end{align*}
and 
\begin{equation*}
\mathbf{W}_{\underline{n},m}^{*}
=
\mathbf{I}_{M_{\un m}}
-
\mathbf{A}_{2,m}^{*}\otimes \mathbf{A}_{1,m}^{*}
-
\mathbf{B}_{2,m}^{*}\otimes \mathbf{B}_{1,m}^{*}
+
\mathbf{D}_{2,m}^{*}\otimes \mathbf{D}_{1,m}^{*},
\end{equation*}
we get
\begin{align}
& \sqrt{N}\left(\bigl(\mathbf{W}_{\underline{n},m}^{*}\bigr)^{-1}-I_{M_{\un m}}\right)   \notag   \\
=& -\sqrt{N}\left(-
\mathbf{A}_{2,m}^{*}\otimes \mathbf{A}_{1,m}^{*}
-
\mathbf{B}_{2,m}^{*}\otimes \mathbf{B}_{1,m}^{*}
+
\mathbf{D}_{2,m}^{*}\otimes \mathbf{D}_{1,m}^{*}\right)\bigl(\mathbf{W}_{\underline{n},m}^{*}\bigr)^{-1}.   \label{eq:W_nm_inv_calculation}
\end{align}
Finally, for all $i \in \{1,2\}$, note that all entries of $\mathbf{A}_{i,m}^{*}$, $\mathbf{B}_{i,m}^{*}$ and $\mathbf{D}_{i,m}^{*}$ are of order $O(1/n_i)$, respectively, due to the application of the selection matrices $\mathbf{P}_{i,m}$ to $\mathbf{A}_i^{*}$, $\mathbf{B}_i^{*}$ and $\mathbf{D}_i^{*}$. Hence, the Kronecker products $\mathbf{A}_{2,m}^{*}\otimes \mathbf{A}_{1,m}^{*}$, $\mathbf{B}_{2,m}^{*}\otimes \mathbf{B}_{1,m}^{*}$ and $\mathbf{D}_{2,m}^{*}\otimes \mathbf{D}_{1,m}^{*}$ are of order $O(1/N)$. Hence, \eqref{eq:W_nm_inv_calculation} becomes
\begin{align*}
-\sqrt{N} O\left(\frac{1}{N}\right)O(1) = O\left(\frac{1}{\sqrt{N}}\right),
\end{align*}
which completes the proof.  \hfill $\square$

\section{Supplementary figures}\label{sec:additional_plots}

In Figure~\ref{fig:masks10}, a heatmap for each spatial dependence level is shown, each related to observations in a grid of size $(10\times 10)$. In the axis, the $(h_1, h_2)$-lag pairs are represented. Colored cells indicates if the bias-correction estimator $ \widehat C_{\un m}^*$ has achieved a better MSE than the standard estimator $\widehat{ C}$, and, if so, it also shows which $m$ values achieve the best possible MSE reduction for each specific lag pair. Recall that the higher $m$, the more covariances are used in the weighting of the bias-corrected estimator and therefore the more bias it may reduce. We see that especially for high levels of spatial dependence, i.e., for $\lambda \in \{6, 8\}$, the bias corrected method can afford to use higher values of $m$ and thus is able to correct more of the bias. The crossed out, colored cells, corresponding to larger lag pairs, indicate cases where a better MSE is achieved by simply tapering the estimator, i.e., by setting $\widehat{C}(\un h):=0$. This behavior is typical especially for processes with short-range spatial dependence, where the true covariance function decays rapidly to zero, such that the estimation noise dominates any small signal. Consequently, tapering to zero these entries reduces the variance more than the bias it introduces, leading to an overall improvement in the MSE. Finally the cells that are left in white (without crosses) have a better MSE either with the standard estimation or with the tapered to zero estimation, but in those cases, bias correction performs worse than either of those two options.
Figure~\ref{fig:masks20} shows the same MSE heatmap, but for a grid of size $(20\times 20)$. On this larger lattice, the symmetry implied by the chosen (separable) covariance model becomes more clearly visible. Also here the optimal choice of $m$ for the same lag-pair increases together with the strength of the spatial dependence.

\begin{figure}[t]
\begin{subfigure}{\textwidth}
  \centering
 \begin{subfigure}{0.495\linewidth}
    \centering
    \includegraphics[width=\linewidth]{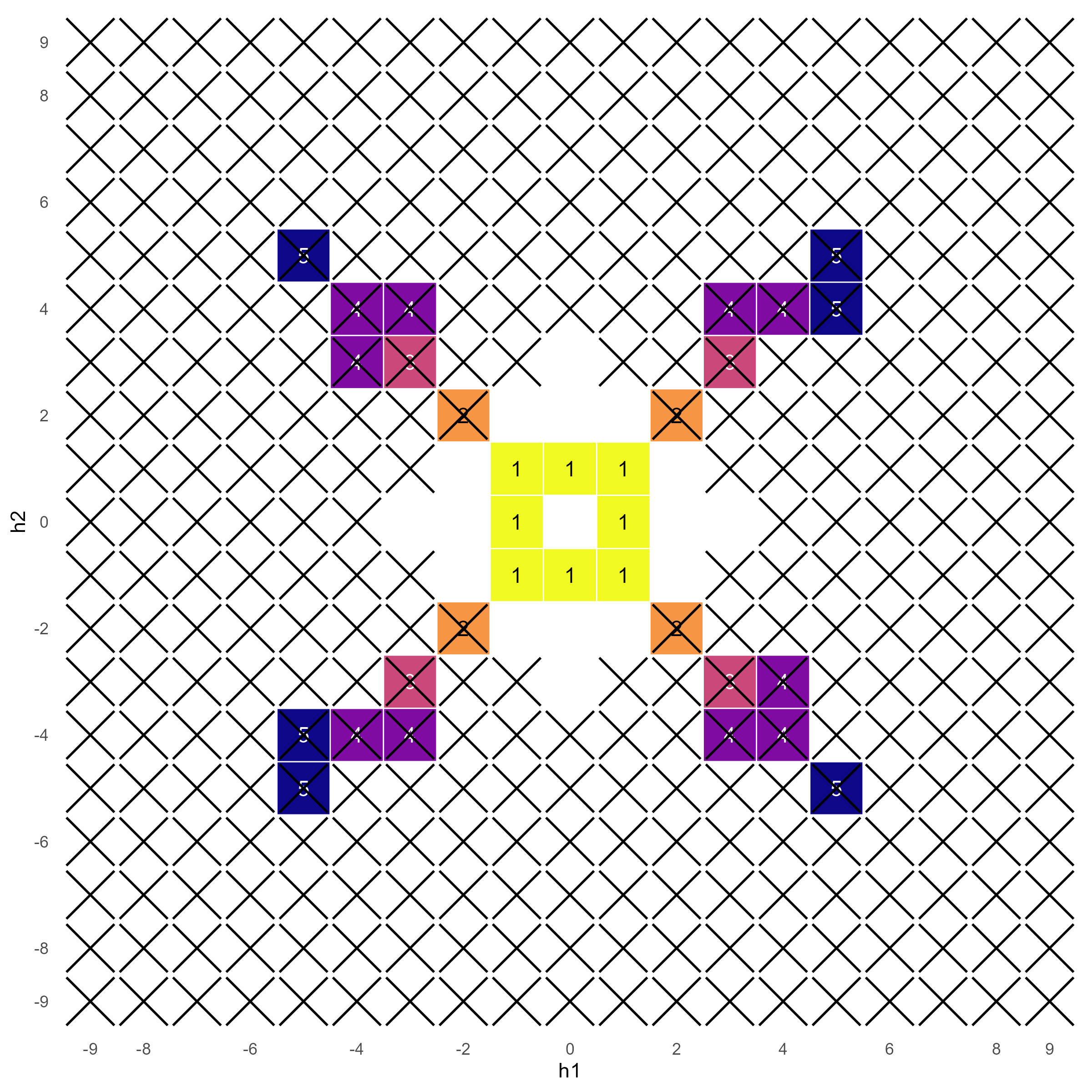}
    \caption{$\lambda=2$}
  \end{subfigure}\hfill
  \begin{subfigure}{0.495\linewidth}
    \centering
    \includegraphics[width=\linewidth]{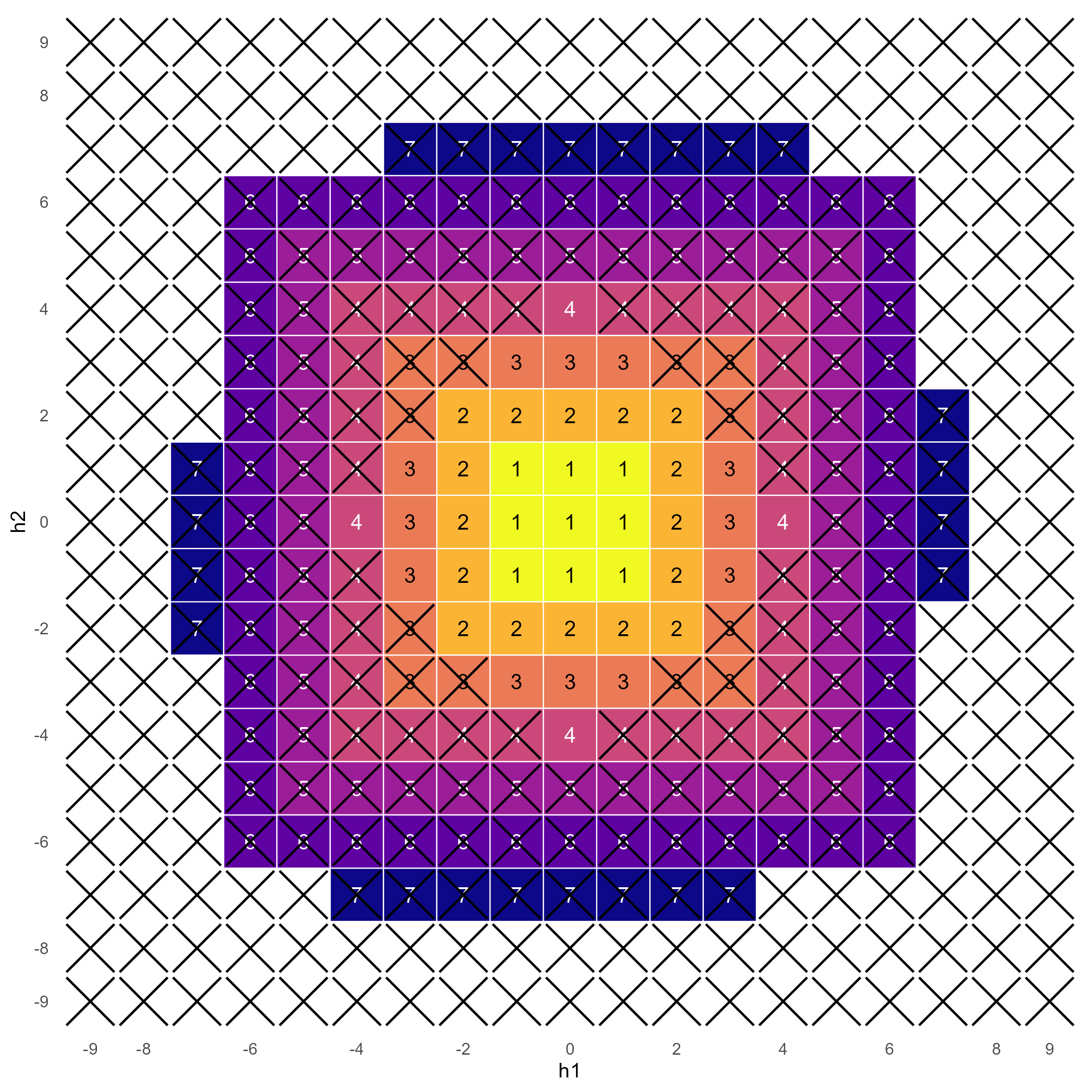}
    \caption{$\lambda=4$}
  \end{subfigure}

  \medskip

  \begin{subfigure}{0.495\linewidth}
    \centering
    \includegraphics[width=\linewidth]{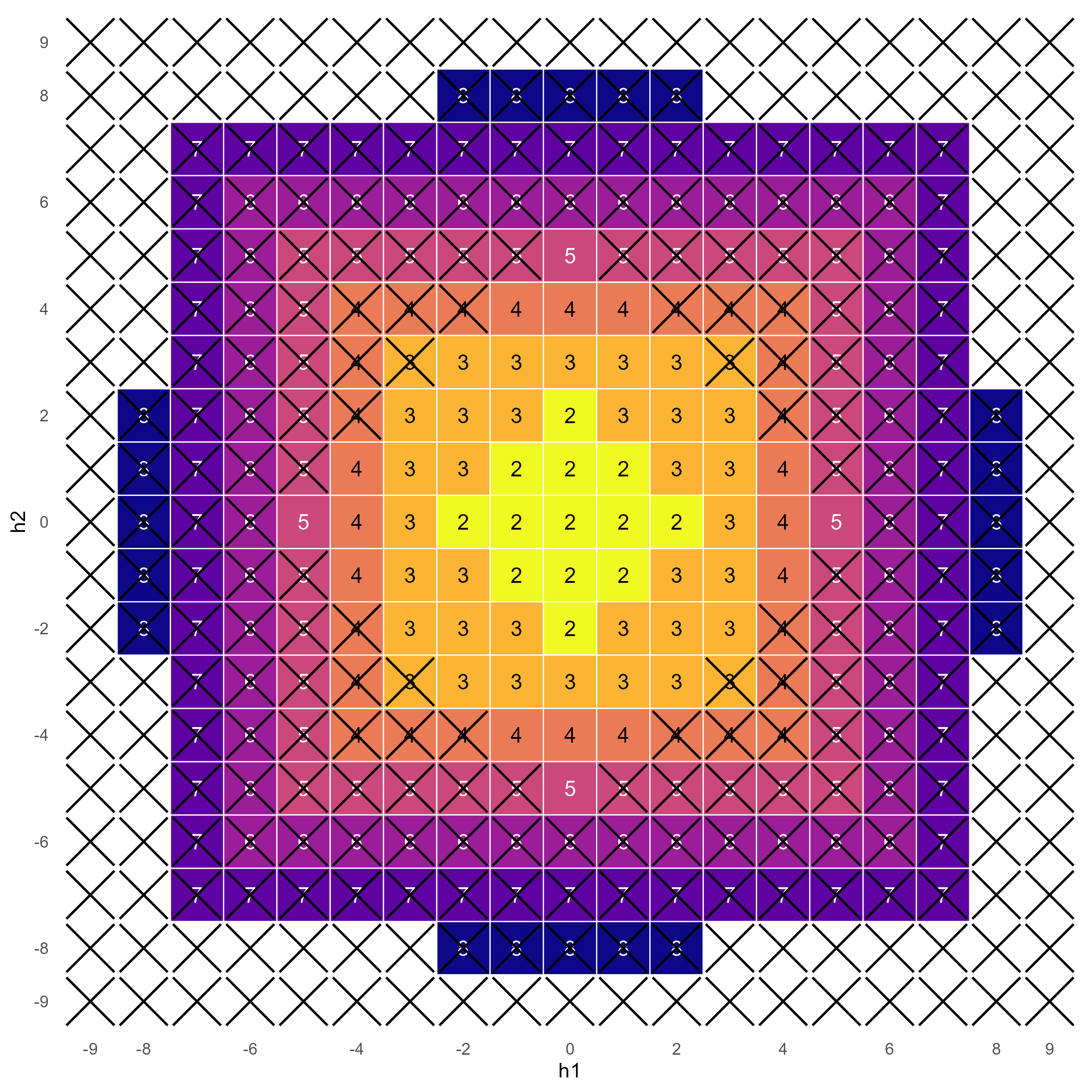}
    \caption{$\lambda=6$}
  \end{subfigure}\hfill
  \begin{subfigure}{0.495\linewidth}
    \centering
    \includegraphics[width=\linewidth]{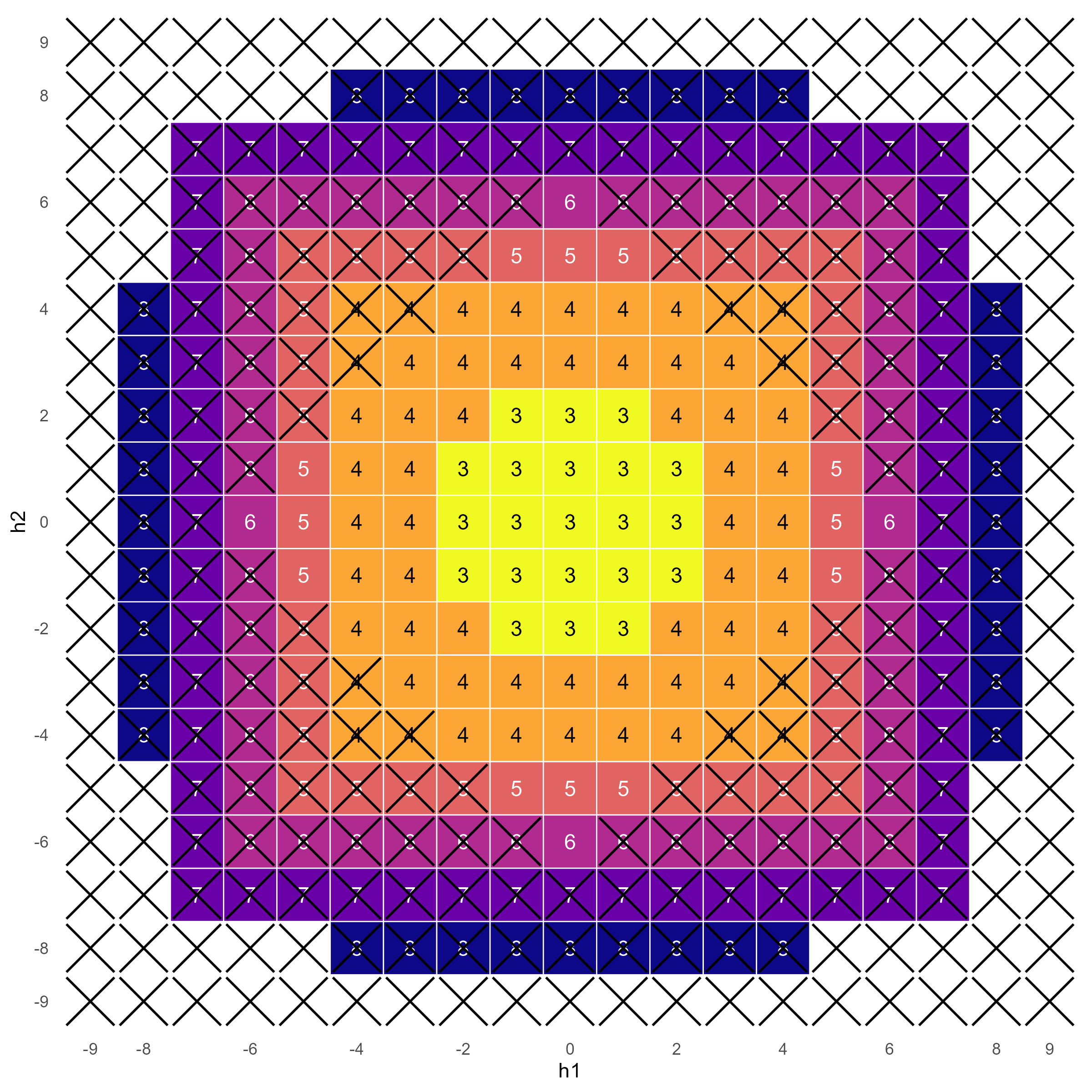}
    \caption{$\lambda=8$}
  \end{subfigure}
\end{subfigure}

\caption{Values of $m$ for which $\widehat{\underline{C}}_{m}^{*}$ achieves the best MSE reduction compared to the naive estimator $\widehat{C}$. Crossed-out cells indicate that tapering to zero  ($\widehat C(h_1,h_2):=0$) yields the best MSE. Grid size $n_1=n_2=10$.}
\label{fig:masks10}
\end{figure}

\begin{figure}[h!]
    \centering
\begin{subfigure}{0.495\textwidth}
  \centering
  \includegraphics[width=\linewidth]{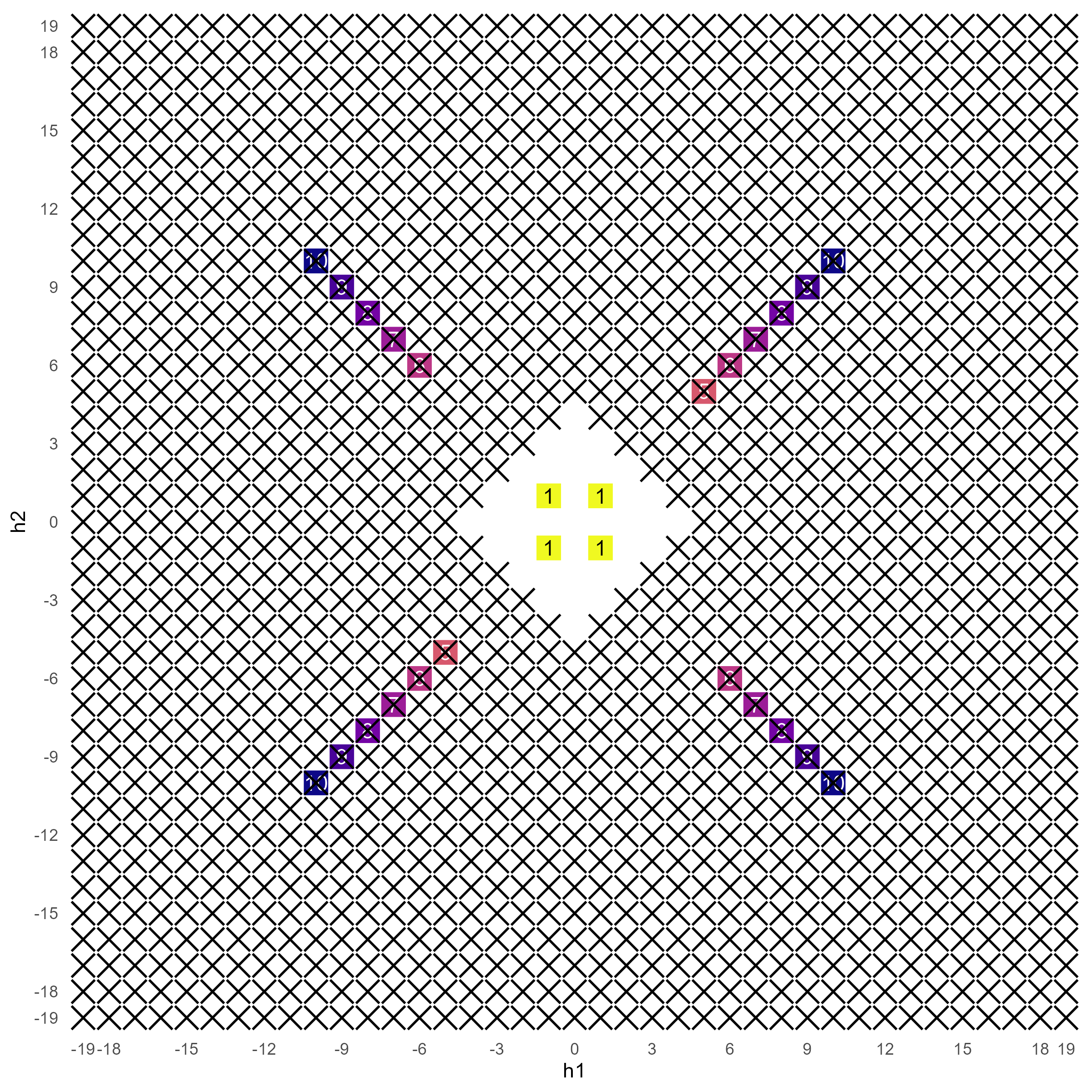}
  \caption{ $\lambda=2$}
  \label{fig:truecov-l2}
\end{subfigure}\hfill
\begin{subfigure}{0.495\textwidth}
  \centering
  \includegraphics[width=\linewidth]{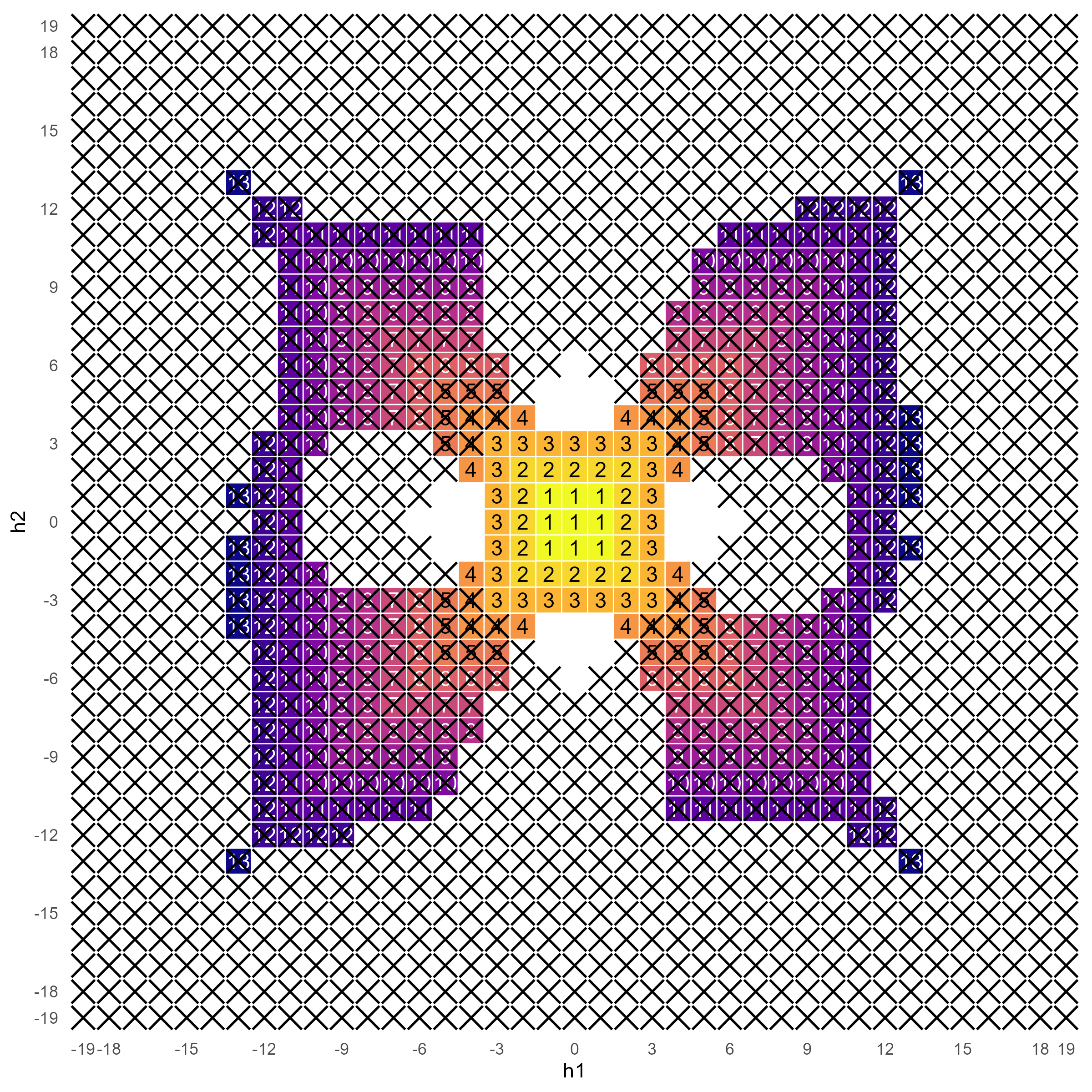}
  \caption{$\lambda=4$}
  \label{fig:truecov-l4}
\end{subfigure}
\begin{subfigure}{0.495\textwidth}
  \centering
  \includegraphics[width=\linewidth]{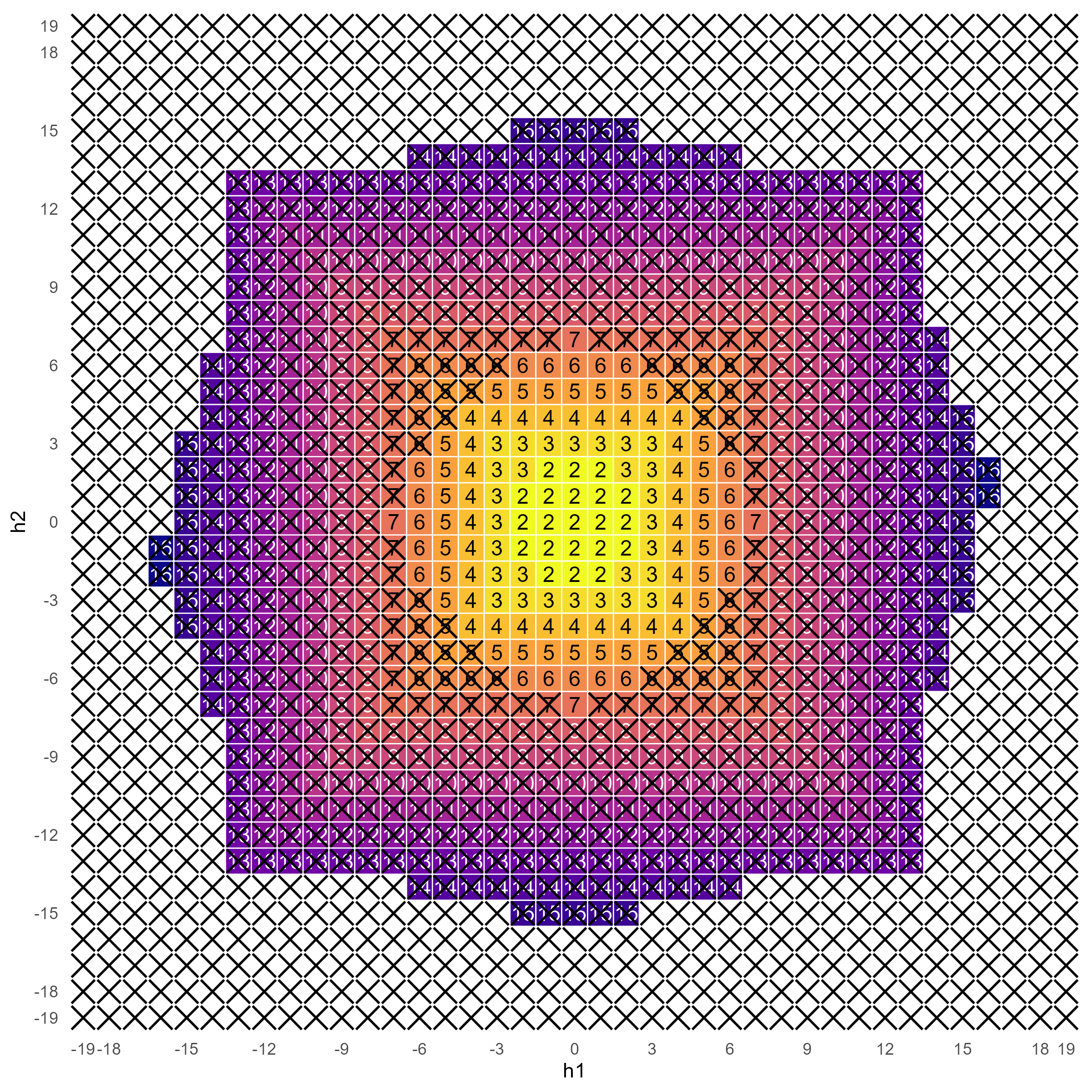}
  \caption{$\lambda=6$}
  \label{fig:truecov-l6}
\end{subfigure}\hfill
\begin{subfigure}{0.495\textwidth}
  \centering
  \includegraphics[width=\linewidth]{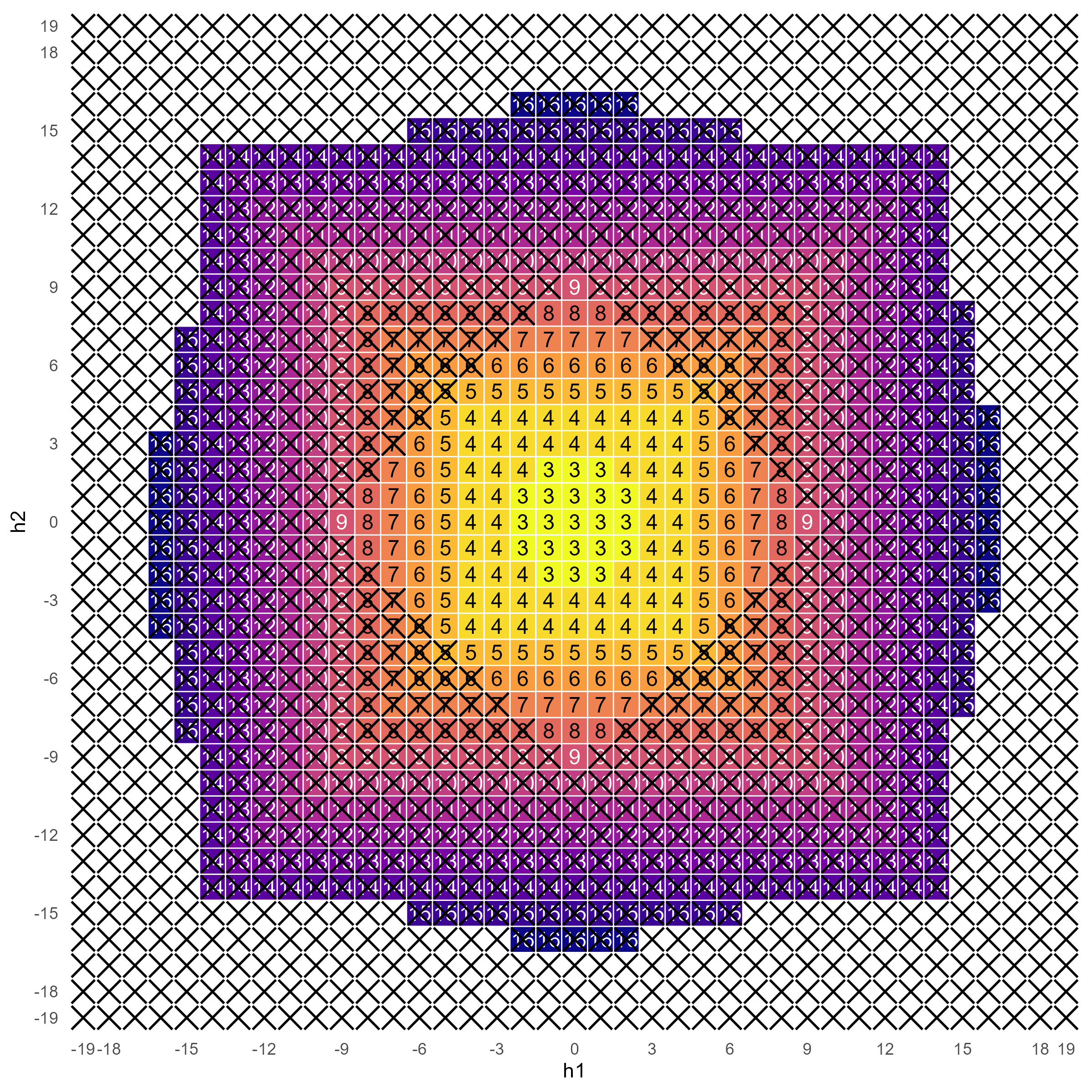}
  \caption{$\lambda=8$}
  \label{fig:masks20}
\end{subfigure}
\caption{Values of $m$ for which $\widehat{\underline{C}}_{m}^{*}$ achieves the best MSE reduction compared to the naive estimator $\widehat{C}$. Crossed-out cells indicate that tapering to zero  ($\widehat C(h_1,h_2):=0$) yields the best MSE. Grid size $n_1=n_2=20$.}
\label{fig:masks20}
\end{figure}

\FloatBarrier
\end{document}